\documentclass[fleqn]{article}
\usepackage{latexsym}
\usepackage{amssymb}
\usepackage{stmaryrd}
\usepackage{graphicx}
\usepackage{color}
\usepackage{url}
\usepackage{phonetic}
\usepackage{amsmath}

\newcommand{\be}{\begin{enumerate}}
\newcommand{\ee}{\end{enumerate}}
\newcommand{\bi}{\begin{itemize}}
\newcommand{\ei}{\end{itemize}}
\newcommand{\bc}{\begin{center}}
\newcommand{\ec}{\end{center}}
\newcommand{\bsp}{\begin{sloppypar}}
\newcommand{\esp}{\end{sloppypar}}

\newtheorem{thm}{Theorem}[subsection]
\newtheorem{cor}[thm]{Corollary}
\newtheorem{lem}[thm]{Lemma}
\newtheorem{prop}[thm]{Proposition}
\newtheorem{rem}[thm]{Remark}

\newenvironment{proof}{\par\noindent{\bf Proof\sglsp}}{\hfill$\Box$}

\newcommand{\sglsp}{\ }
\newcommand{\dblsp}{\ \ }

\newcommand{\sC}{\mbox{$\cal C$}}
\newcommand{\sD}{\mbox{$\cal D$}}
\newcommand{\sE}{\mbox{$\cal E$}}

\newcommand{\sH}{\mbox{$\cal H$}}

\newcommand{\sM}{\mbox{$\cal M$}}

\newcommand{\sT}{\mbox{$\cal T$}}

\newcommand{\sV}{\mbox{$\cal V$}}

\renewcommand{\phi}{\varphi}

\newcommand{\churchqe}{$\mbox{\sc ctt}_{\rm qe}$}
\newcommand{\churchuqe}{$\mbox{\sc ctt}_{\rm uqe}$}
\newcommand{\churcheps}{$\mbox{\sc ctt}_\epsilon$}
\newcommand{\qzero}{${\cal Q}_0$}
\newcommand{\qzerou}{${\cal Q}^{\rm u}_{0}$}
\newcommand{\qzerouqe}{${\cal Q}^{\rm uqe}_{0}$}

\newcommand{\set}[1]{{\{ #1 \}}}
\newcommand{\sembrack}[1]{\llbracket#1\rrbracket}
\newcommand{\synbrack}[1]{\ulcorner#1\urcorner}
\newcommand{\commabrack}[1]{\lfloor#1\rfloor}
\newcommand{\mname}[1]{\mbox{\sf #1}}

\newcommand{\mdot}{\mathrel.}
\newcommand{\tarrow}{\rightarrow}
\newcommand{\LambdaApp}{\lambda\,}
\newcommand{\Neg}{\neg}
\newcommand{\NegAlt}{{\sim}}
\renewcommand{\And}{\wedge}
\newcommand{\Implies}{\supset}
\newcommand{\Or}{\vee}
\newcommand{\Iff}{\equiv}
\newcommand{\Forall}{\forall}
\newcommand{\ForallApp}{\forall\,}
\newcommand{\Forsome}{\exists}
\newcommand{\ForsomeApp}{\exists\,}

\newcommand{\imps}{\mbox{\sc imps}}

\newcommand{\lutins}{\mbox{\sc lutins}}

\newcommand{\proves}[2]{#1 \vdash #2}

\newcommand{\TRUE}{\mbox{{\sc t}}}
\newcommand{\FALSE}{\mbox{{\sc f}}}

\title{{\bf Incorporating Quotation and Evaluation Into Church's Type
    Theory}\thanks{This paper is a greatly extended version of
    \cite{Farmer16}. This research was supported by NSERC.}}

\author{William M. Farmer\thanks{Address: Department of Computing and
    Software, McMaster University, 1280 Main Street West, Hamilton,
    Ontario L8S 4K1, Canada.  E-mail:
    {\texttt{wmfarmer@mcmaster.ca}.}}}

\date{3 March 2018}

\begin{document}

\maketitle

\begin{abstract}

\noindent
{\churchqe} is a version of Church's type theory that includes
quotation and evaluation operators that are similar to quote and eval
in the Lisp programming language.  With quotation and evaluation it is
possible to reason in {\churchqe} about the interplay of the syntax
and semantics of expressions and, as a result, to formalize
syntax-based mathematical algorithms.  We present the syntax and
semantics of {\churchqe} as well as a proof system for {\churchqe}.
The proof system is shown to be sound for all formulas and complete
for formulas that do not contain evaluations.  We give several
examples that illustrate the usefulness of having quotation and
evaluation in {\churchqe}.  

\bigskip


\end{abstract}

\newpage

\tableofcontents
\listoftables

\newpage

\section{Introduction}\label{sec:introduction}

The Lisp programming language is famous for its use of
\emph{quotation} and \emph{evaluation}.\footnote{Lisp is also famous
  for its use of \emph{quasiquotation}, a more powerful form of
  quotation which we discuss in section~\ref{sec:examples}.}  From
code the Lisp quotation operator called \emph{quote} produces
meta-level data (i.e., S-expressions) that represents the code, and
from this data the Lisp evaluation operator called \emph{eval}
produces the code that the data represents.  In Lisp,
\emph{metaprogramming} (i.e., programming at the meta-level) is
performed by manipulating S-expressions and is \emph{reflected} (i.e.,
integrated) into object-level programming by the use of quote and
eval.

Metaprogramming with reflection is a very powerful programming tool.
Besides Lisp, several other programming languages employ quotation and
evaluation mechanisms to enable metaprogramming with reflection.
Examples include Agda~\cite{Norell07,Norell09,VanDerWalt12},
Archon~\cite{Stump09}, Elixir~\cite{Elixir18}, F\#~\cite{FSharp18},
Idris~\cite{Christiansen:2016,Christiansen:2014,Christiansen:2016:Thesis},
MetaML~\cite{TahaSheard00}, MetaOCaml~\cite{MetaOCaml11},
reFLect~\cite{GrundyEtAl06}, Scala~\cite{Odersky16,Scalameta18}, and
Template Haskell~\cite{SheardJones02}.  Indeed nearly all major
programming languages today provide some level of support for
metaprogramming.  Quotation is not crucial for metaprogramming since
strings or abstract syntax trees (ASTs) can be used directly as data
representing code.  However, quotation is convenient for connecting
code to its representation.

Analogous to metaprogramming in a programming language,
\emph{metareasoning} is performed in a logic by manipulating
meta-level values (e.g., ASTs) that represent expressions in the logic
and is \emph{reflected} into object-level reasoning using quotation
and evaluation\footnote{Evaluation in this context is also called
  unquoting, interpretation, dereferencing, and dereification.}
mechanisms~\cite{Costantini02}.  Metareasoning with reflection has
been used in several proof assistants --- including
Agda~\cite{VanDerWalt12,VanDerWaltSwierstra12},
Coq~\cite{Boutin97,BraibantPous11,Chlipala13,GonthierEtAl15,GregoireMahboubi05,JamesHinze09,OostdijkGeuvers02},
Idris~\cite{Christiansen:2016,Christiansen:2014,Christiansen:2016:Thesis},
Isabelle/HOL~\cite{ChaiebNipkow08},
Lean~\cite{ebner2017metaprogramming}, Maude~\cite{ClavelMeseguer02},
Nqthm/ACL1~\cite{BoyerMoore81,HuntEtAl05},
Nuprl/MetaPRL~\cite{AllenEtAl90,Barzilay05,Constable95,Howe92,KnoblockConstable86,Nogin05,Yu07},
PVS~\cite{VonHenkeEtAl98}, reFLect~\cite{MelhamEtAl13}, and
Theorema~\cite{GieseBuchberger07} --- for formalizing metalogical
techniques (logical reflection) and incorporating symbolic computation
into proofs (computational reflection)~\cite{Farmer13,Harrison95}.

In proof assistants such as Coq and Agda, metareasoning with
reflection is implemented in the logic by defining an infrastructure
consisting of (1) an \emph{inductive type of syntactic values} that
represent certain object-level expressions, (2) an \emph{informal
  quotation operator} (residing only at the meta-level) that maps
these object-level expressions to syntactic values, and (3) a
\emph{formal evaluation operator} (residing at the object-level) that
maps syntactic values to the values of the object-level expressions
that they represent.  (The three components of this approach form an
instance of a \emph{syntax framework}~\cite{FarmerArxiv13a}, a
mathematical structure that models systems for reasoning about the
syntax of an interpreted language.)

The reflection infrastructures that have been employed in today's
proof assistants are usually \emph{local} in the sense that the
syntactic values of the inductive type represent only the expressions
of the logic that are relevant to a particular problem, the quotation
operator can only be applied to these expressions, and the evaluation
operator can only be applied to the syntactic values of this inductive
type.  Can metareasoning with reflection be implemented in a
traditional logic like first-order logic or simple type
theory~\cite{Farmer08} using a \emph{global} infrastructure for the
entire set of expressions of the logic with global quotation and
evaluation operators like Lisp's quote and eval?  This is largely an
open question.

We have proposed a version of NBG set theory named
Chiron~\cite{FarmerArxiv13} and a version of Alonzo Church's type
theory~\cite{Church40}\footnote{Church's type theory~\cite{Church40}
  is a version of simple type theory with lambda notation.  It is a
  classical form of type theory in contrast to constructive type
  theories, like Martin-L\"of type theory~\cite{Martin-Lof84} and the
  calculus of constructions~\cite{CoquandHuet88}.} named
{\qzerouqe}~\cite{FarmerArxiv14} that include global quotation and
evaluation operators, but these logics have a high level of complexity
and are not easy to implement.  This paper presents a logic named
{\churchqe}, a version of Church's type theory with quotation and
evaluation which is much simpler than {\qzerouqe}.  We believe
{\churchqe} is the first readily implementable version of simple type
theory that includes global quotation and evaluation.
See~\cite{GieseBuchberger07} for research in a similar direction.

Several challenging design problems face the logic engineer who seeks
to incorporate global quotation and evaluation into a traditional
logic.  The three design problems that most concern us are the
following.  We will write the quotation and evaluation operators
applied to an expression $e$ as $\synbrack{e}$ and $\sembrack{e}$,
respectively.

\be

  \item \emph{Evaluation Problem.}  An evaluation operator is
    applicable to syntactic values that represent formulas and thus is
    effectively a truth predicate.  Hence, by the proof of Alfred
    Tarski's theorem on the undefinability of
    truth~\cite{Tarski33,Tarski35,Tarski35a}, if the evaluation
    operator is total in the context of a sufficiently strong theory
    like first-order Peano arithmetic, then it is possible to express
    the liar paradox using the quotation and evaluation operators.
    Therefore, the evaluation operator must be partial and the law of
    disquotation\footnote{There are different meanings for the law of
      disquotation depending, for example, on how free variables in a
      quotation are handled.  We are assuming a law of disquotation in
      which $\sembrack{\synbrack{e}}$ has the same value as $e$ for
      any variable assignment.} cannot hold universally (i.e., for
    some expressions $e$, $\sembrack{\synbrack{e}} \not= e$).  As a
    result, reasoning with evaluation is cumbersome and leads to
    undefined expressions.

  \item \emph{Variable Problem.}  The variable $x$ is not free in the
    expression $\synbrack{x + 3}$ (or in any quotation).  However, $x$
    is free in $\sembrack{\synbrack{x + 3}}$ because
    $\sembrack{\synbrack{x + 3}} = x + 3$.  If the value of a constant
    $c$ is $\synbrack{x + 3}$, then $x$ is free in $\sembrack{c}$
    because $\sembrack{c} = \sembrack{\synbrack{x + 3}} = x + 3$.
    Hence, in the presence of an evaluation operator, whether or not a
    variable is free in an expression may depend on the values of the
    expression's components.  As a consequence, the substitution of an
    expression for the free occurrences of a variable in another
    expression depends on the semantics (as well as the syntax) of the
    expressions involved and must be integrated with the proof system
    for the logic.  That is, a logic with quotation and evaluation
    requires a semantics-dependent form of substitution in which side
    conditions, like whether a variable is free in an expression, are
    proved within the proof system.  This is a major departure from
    traditional logic.

  \item \emph{Double Substitution Problem.}  By the semantics of
    evaluation, the value of $\sembrack{e}$ is the \emph{value} of the
    expression whose syntax tree is represented by the \emph{value} of
    $e$.  Hence the semantics of evaluation involves a double
    valuation (see condition 6 of the definition of a general model in
    subsection~\ref{subsec:gen-models}).  If the value of a variable $x$
    is $\synbrack{x}$, then $\sembrack{x} = \sembrack{\synbrack{x}} =
    x = \synbrack{x}$ (see Proposition~\ref{prop:double-sub}).  Hence
    the substitution of $\synbrack{x}$ for $x$ in $\sembrack{x}$
    requires one substitution inside the argument of the evaluation
    operator and another substitution after the evaluation operator is
    eliminated.  This double substitution is another major departure
    from traditional logic.

\ee

To solve the Evaluation Problem, it is necessary to restrict the
application of either the quotation operator or the evaluation
operator.  In {\qzerouqe}, $\synbrack{e}$ is defined for every
expression $e$, but $\sembrack{\synbrack{e}}$ is defined only if $e$
is an expression in which every occurrence of the evaluation operator
is within a quotation.  The Variable and Double Substitution Problems
are then solved by defining an explicit substitution operator that
operates on syntactic values.  This results in a proof system for
{\qzerouqe} that is very complex and very difficult to implement.  On
the other hand, if $\synbrack{e}$ is defined only when $e$ is a closed
evaluation-free expression, then all three of the design problem
disappear.  However, most of the usefulness of having quotation and
evaluation in a logic also disappear --- which is illustrated by the
examples in section~\ref{sec:examples}.

The logic {\churchqe} takes a middle path to solve the three design
problems: $\synbrack{e}$ is defined only if $e$ is a (possible open)
expression that does not contain the evaluation operator.  It is much
simpler than {\qzerouqe}, but also much more useful than a logic in
which only closed expressions can be quoted.  Like {\qzerouqe},
{\churchqe} is based on {\qzero}~\cite{Andrews02}, Peter Andrews'
version of Church's type theory.  In this paper, we present the syntax
and semantics of {\churchqe} as well as a proof system for
{\churchqe}.  We also give several examples that demonstrate the
benefits of having quotation and evaluation in {\churchqe}.

The rest of the paper is organized as follows.  The syntax of
{\churchqe} is defined in section~\ref{sec:syntax}.  A Henkin-style
general models semantics for {\churchqe} is defined in
section~\ref{sec:semantics}.  Four examples that illustrate the
utility of the quotation and evaluation facility in {\churchqe} are
presented in section~\ref{sec:examples}.  A proof system for
{\churchqe} is given in section~\ref{sec:proof-system}.  Various
proof-theoretic results about the proof system are proved in
section~\ref{sec:pt-results}.  The proof system is proved in
sections~\ref{sec:soundness} and~\ref{sec:completeness} to be,
respectively, sound with respect to the semantics of {\churchqe} and
complete with respect to the semantics of {\churchqe} for
evaluation-free formulas.  In section~\ref{sec:revisit} the results
stated about the examples discussed in section~\ref{sec:examples} are
proved within the proof system for {\churchqe}.  The extensive body of
literature related to {\churchqe} is briefly surveyed in
section~\ref{sec:related-work}.  And the paper ends with some final
remarks in section~\ref{sec:conclusion} including a brief discussion
on future work.

\section{Syntax}\label{sec:syntax}

{\churchqe} has the syntax of Church's type theory plus an inductive
type of syntactic values, a partial quotation operator, and a typed
evaluation operator.  The syntax of {\churchqe} is very similar to the
syntax of {\qzero}~\cite[pp.~210--211]{Andrews02}.  Like {\qzero}, the
propositional connectives and quantifiers are defined using function
application, function abstraction, and equality.  For the sake of
simplicity, {\churchqe} does not contain, as in {\qzero}, a definite
description operator or, as in the logic of the HOL proof
assistant~\cite{GordonMelham93}, an indefinite description (choice)
operator and type variables.

\subsection{Types}

A \emph{type} of {\churchqe} is a string of symbols defined
inductively by the following formation rules:
\be

  \item \emph{Type of individuals}: $\iota$ is a type.

  \item \emph{Type of truth values}: $\omicron$ is a type.

  \item \emph{Type of constructions}: $\epsilon$ is a type.

  \item \emph{Function type}: If $\alpha$ and $\beta$ are types, then
    $(\alpha \tarrow \beta)$ is a type.\footnote{In Andrews'
    {\qzero}~\cite{Andrews02} and Church's original
    system~\cite{Church40}, the function type $(\alpha \tarrow \beta)$
    is written as $(\beta\alpha)$.}

\ee
Let $\sT$ denote the set of types of {\churchqe}.
$\alpha,\beta,\gamma, \ldots$ are syntactic variables ranging over
types.  When there is no loss of meaning, matching pairs of
parentheses in types may be omitted.  We assume that function type
formation associates to the right so that a type of the form $(\alpha
\tarrow (\beta \tarrow \gamma))$ may be written as $\alpha \tarrow
\beta \tarrow \gamma$.

We will see in the next subsection that in {\churchqe} types are
directly assigned to variables and constants and thereby indirectly
assigned to expressions.

\subsection{Expressions}\label{subsec:expressions}

A \emph{typed symbol} is a symbol with a subscript from $\sT$.  Let
$\sV$ be a set of typed symbols such that, for each $\alpha \in \sT$,
$\sV$ contains denumerably many typed symbols with subscript~$\alpha$.
A \emph{variable of type $\alpha$} of {\churchqe} is a member of $\sV$
with subscript~$\alpha$.  $\textbf{f}_\alpha, \textbf{g}_\alpha,
\textbf{h}_\alpha, \textbf{u}_\alpha, \textbf{v}_\alpha,
\textbf{w}_\alpha,\textbf{x}_\alpha, \textbf{y}_\alpha,
\textbf{z}_\alpha,\ldots$ are syntactic variables ranging over
variables of type~$\alpha$.  We will assume that $f_\alpha, g_\alpha,
h_\alpha, u_\alpha, v_\alpha, w_\alpha, x_\alpha, y_\alpha,
z_\alpha,\ldots$ are actual variables of type~$\alpha$ of {\churchqe}.

Let $\sC$ be a set of typed symbols disjoint from $\sV$ that includes
the typed symbols in Table~\ref{tab:log-con}.  A \emph{constant of
  type~$\alpha$} of {\churchqe} is a member of $\sC$ with
subscript~$\alpha$.  The typed symbols in Table~\ref{tab:log-con} are
the \emph{logical constants} of {\churchqe}.  $\textbf{c}_\alpha,
\textbf{d}_\alpha, \ldots$ are syntactic variables ranging over
constants of type~$\alpha$.

\begin{table}
\bc
\begin{tabular}{|ll|}
\hline
$\mname{=}_{\alpha \tarrow \alpha \tarrow o}$ 
& for all $\alpha \in \sT$\\
$\mname{is-var}_{\epsilon \tarrow o}$
&\\
$\mname{is-var}_{\epsilon \tarrow o}^{\alpha}$
& for all $\alpha \in \sT$\\
$\mname{is-con}_{\epsilon \tarrow o}$
&\\
$\mname{is-con}_{\epsilon \tarrow o}^{\alpha}$
& for all $\alpha \in \sT$\\
$\mname{app}_{\epsilon \tarrow \epsilon \tarrow \epsilon}$
&\\
$\mname{abs}_{\epsilon \tarrow \epsilon \tarrow \epsilon}$
&\\
$\mname{quo}_{\epsilon \tarrow \epsilon}$
&\\
$\mname{is-expr}_{\epsilon \tarrow o}$
&\\
$\mname{is-expr}_{\epsilon \tarrow o}^{\alpha}$
& for all $\alpha \in \sT$\\
$\sqsubset_{\epsilon \tarrow \epsilon \tarrow o}$
&\\
$\mname{is-free-in}_{\epsilon \tarrow \epsilon \tarrow o}$
&\\
\hline
\end{tabular}
\ec
\caption{Logical Constants}\label{tab:log-con}
\end{table}

An \emph{expression of type $\alpha$} of {\churchqe} is a string of
symbols defined inductively by the formation rules below.
$\textbf{A}_\alpha, \textbf{B}_\alpha, \textbf{C}_\alpha, \ldots$ are
syntactic variables ranging over expressions of type $\alpha$.  An
expression is \emph{eval-free} if it is constructed using just the
first five formation rules.
\be

  \item \emph{Variable}: $\textbf{x}_\alpha$ is an expression of type
    $\alpha$.

  \item \emph{Constant}: $\textbf{c}_\alpha$ is an expression of type
    $\alpha$.

  \item \emph{Function application}: $(\textbf{F}_{\alpha \tarrow
    \beta} \, \textbf{A}_\alpha)$ is an expression of type $\beta$.

  \item \emph{Function abstraction}: $(\LambdaApp \textbf{x}_\alpha
    \mdot \textbf{B}_\beta)$ is an expression of type $\alpha \tarrow
    \beta$.

  \item \emph{Quotation}: $\synbrack{\textbf{A}_\alpha}$ is an
    expression of type $\epsilon$ if $\textbf{A}_\alpha$ is eval-free.

  \item \emph{Evaluation}: $\sembrack{\textbf{A}_\epsilon}_{{\bf
      B}_\beta}$ is an expression of type $\beta$.

\ee 

\noindent
The sole purpose of the second component $\textbf{B}_\beta$ in an
evaluation $\sembrack{\textbf{A}_\epsilon}_{{\bf B}_\beta}$ is to
establish the type of the evaluation.  A \emph{formula} is an
expression of type $o$.  When there is no loss of meaning, matching
pairs of parentheses in expressions may be omitted.  We assume that
function application formation associates to the left so that an
expression of the form $((\textbf{G}_{\alpha \tarrow \beta \tarrow
  \gamma} \, \textbf{A}_\alpha) \, \textbf{B}_\beta)$ may be written
as $\textbf{G}_{\alpha \tarrow \beta \tarrow \gamma} \,
\textbf{A}_\alpha \, \textbf{B}_\beta$.

Let $\textbf{A}_\alpha$ and $\textbf{B}_\beta$ be eval-free
expressions.  An occurrence of a variable $\textbf{x}_\alpha$ in
$\textbf{B}_\beta$ is \emph{bound} [\emph{free}] if (1) it is not in a
quotation and (2) it is [not] in a subexpression of $\textbf{B}_\beta$
of the form $\LambdaApp \textbf{x}_\alpha \mdot \textbf{C}_\gamma$.  A
variable $\textbf{x}_\alpha$ is \emph{bound} [\emph{free}] \emph{in}
$\textbf{B}_\beta$ if there is a bound [free] occurrence of
$\textbf{x}_\alpha$ in $\textbf{B}_\beta$. $\textbf{A}_\alpha$ is
\emph{free for $\textbf{x}_\alpha$ in} $\textbf{B}_\beta$ if no free
occurrence of $\textbf{x}_\alpha$ in $\textbf{B}_\beta$ is within a
subexpression of $\textbf{B}_\beta$ of the form $\LambdaApp
\textbf{y}_\gamma \mdot \textbf{C}_\delta$ such that
$\textbf{y}_\gamma$ is free in $\textbf{A}_\alpha$.

\begin{rem}[Bound and Free Variables]\em\bsp
For expressions that are not eval-free (i.e., contain one or more
evaluations), it is not possible to define in a purely syntactic way
the notion of a bound or free variable due to the Variable Problem
(see section~\ref{sec:introduction}).  Hence notions concerning bound
and free variables, such as substitution and alpha-equivalence, cannot
be readily extended to non-eval-free expressions.\esp
\end{rem}

\subsection{Constructions}

A \emph{construction} of {\churchqe} is an expression of type
$\epsilon$ defined inductively as follows:

\be

  \item $\synbrack{\textbf{x}_\alpha}$ is a construction.

  \item $\synbrack{\textbf{c}_\alpha}$ is a construction.

  \item If $\textbf{A}_\epsilon$ and $\textbf{B}_\epsilon$ are
    constructions, then $\mname{app}_{\epsilon \tarrow \epsilon
      \tarrow \epsilon} \, \textbf{A}_\epsilon \,
    \textbf{B}_\epsilon$, $\mname{abs}_{\epsilon \tarrow \epsilon
      \tarrow \epsilon} \, \textbf{A}_\epsilon \,
    \textbf{B}_\epsilon$, and $\mname{quo}_{\epsilon \tarrow \epsilon}
    \, \textbf{A}_\epsilon$ are constructions.

\ee

\noindent
The set of constructions is thus an inductive type whose base elements
are quotations of variables and constants and whose constructors are
$\mname{app}_{\epsilon \tarrow \epsilon \tarrow \epsilon}$,
$\mname{abs}_{\epsilon \tarrow \epsilon \tarrow \epsilon}$, and
$\mname{quo}_{\epsilon \tarrow \epsilon}$.  We will call these three
constants \emph{syntax constructors}.

Let $\sE$ be the function mapping eval-free expressions to
constructions that is defined inductively as follows:

\be

  \item $\sE(\textbf{x}_\alpha) = \synbrack{\textbf{x}_\alpha}$.

  \item $\sE(\textbf{c}_\alpha) = \synbrack{\textbf{c}_\alpha}$.

  \item $\sE(\textbf{F}_{\alpha \tarrow \beta} \, \textbf{A}_\alpha) =
    \mname{app}_{\epsilon \tarrow \epsilon \tarrow \epsilon} \,
    \sE(\textbf{F}_{\alpha \tarrow \beta}) \, \sE(\textbf{A}_\alpha)$.

  \item $\sE(\LambdaApp \textbf{x}_\alpha \mdot \textbf{B}_\beta) =
    \mname{abs}_{\epsilon \tarrow \epsilon \tarrow \epsilon} \,
    \sE(\textbf{x}_\alpha) \, \sE(\textbf{B}_\beta)$.

  \item $\sE(\synbrack{\textbf{A}_\alpha}) = \mname{quo}_{\epsilon
    \tarrow \epsilon} \, \sE(\textbf{A}_\alpha)$.

\ee

\noindent
$\sE$ is clearly injective.  When $\textbf{A}_\alpha$ is eval-free,
$\sE(\textbf{A}_\alpha)$ is the unique construction that represents
the syntax tree of $\textbf{A}_\alpha$.  That is,
$\sE(\textbf{A}_\alpha)$ is a syntactic value that represents how
$\textbf{A}_\alpha$ is syntactically constructed.  For every eval-free
expression, there is a construction that represents its syntax tree,
but not every construction represents the syntax tree of an eval-free
expression.  For example, $\mname{app}_{\epsilon \tarrow \epsilon
  \tarrow \epsilon} \, \synbrack{\textbf{x}_\alpha} \,
\synbrack{\textbf{x}_\alpha}$ represents the syntax tree of
$(\textbf{x}_\alpha \, \textbf{x}_\alpha)$ which is not an expression
of {\churchqe} since the types are mismatched.  A construction is
\emph{proper} if it is in the range of $\sE$, i.e., it represents the
syntax tree of an eval-free expression.  Whether a construction is
proper is easily decided by examining it syntactic structure.

The five kinds of eval-free expressions and the syntactic values that
represent their syntax trees are given in
Table~\ref{tab:eval-free-exprs}.  The logical constants
$\mname{is-var}_{\epsilon \tarrow o}$, $\mname{is-var}_{\epsilon
  \tarrow o}^{\alpha}$, $\mname{is-con}_{\epsilon \tarrow o}$,
$\mname{is-con}_{\epsilon \tarrow o}^{\alpha}$,
$\mname{is-expr}_{\epsilon \tarrow o}$, $\mname{is-expr}_{\epsilon
  \tarrow o}^{\alpha}$, $\sqsubset_{\epsilon \tarrow \epsilon \tarrow
  o}$, and $\mname{is-free-in}_{\epsilon \tarrow \epsilon \tarrow o}$
are used to make assertions about the expressions that constructors
represent.  Their meanings are given in
subsection~\ref{subsec:interpretations}.  $\sqsubset_{\epsilon \tarrow
  \epsilon \tarrow o}$ is needed to express the induction principle
for constructions (Axiom B6 in section~\ref{sec:proof-system}).

\begin{table}[b]
\bc
\begin{tabular}{|lll|}
\hline

\textbf{Kind}
& \textbf{Syntax}
& \textbf{Syntactic Value}\\

Variable \hspace*{15ex}
& $\textbf{x}_\alpha$  \hspace*{9ex}
& $\synbrack{\textbf{x}_\alpha}$\\

Constant
& $\textbf{c}_\alpha$
& $\synbrack{\textbf{c}_\alpha}$\\

Function application
& $\textbf{F}_{\alpha \tarrow \beta} \, \textbf{A}_\alpha$
& $\mname{app}_{\epsilon \tarrow \epsilon \tarrow \epsilon} \,
  \sE(\textbf{F}_{\alpha \tarrow \beta}) \, \sE(\textbf{A}_\alpha)$\\

Function abstraction
& $\LambdaApp \textbf{x}_\alpha \mdot \textbf{B}_\beta$
& $\mname{abs}_{\epsilon \tarrow \epsilon \tarrow \epsilon} \,
  \sE(\textbf{x}_\alpha) \, \sE(\textbf{B}_\beta)$\\

Quotation
& $\synbrack{\textbf{A}_\alpha}$
& $\mname{quo}_{\epsilon \tarrow \epsilon} \, \sE(\textbf{A}_\alpha)$\\

\hline
\end{tabular}
\ec
\caption{Five Kinds of Eval-Free Expressions}\label{tab:eval-free-exprs}
\end{table}

\begin{rem}[Type of Constructions]\em\bsp
All constructions have the same type $\epsilon$.  Thus,
$\sE(A_\alpha)$ and $\sE(B_\beta)$ both have type $\epsilon$ even when
$A_\alpha$ and $B_\beta$ are eval-free expressions with different
types.  An alternate approach would be to parameterize $\epsilon$ so
that all constructors of the form $\sE(A_\alpha)$ would have the type
$\epsilon_\alpha$.  Instead of parameterizing $\epsilon$ by the type of
expressions, we have chosen to partition $\epsilon$ by the set
$\set{\mname{is-expr}_{\epsilon \tarrow o}^{\alpha} \;|\; \alpha \in
  \sT}$ of unary predicates on $\epsilon$.\esp
\end{rem}

\subsection{Theories}

Let $\sD \subseteq \sC$.  An expression $\textbf{A}_\alpha$ of
{\churchqe} is a \emph{\sD-expression} if each constant occurring in
$\textbf{A}_\alpha$ is a member of $\sD$.  Let $L_{\cal D}$ be the set
of all \sD-expressions.  A \emph{language} of {\churchqe} is $L_{\cal
  D}$ for some $\sD \subseteq \sC$ such that $\sD$ contains all the
logical constants of {\churchqe}.  A \emph{theory} of {\churchqe} is a
pair $T=(L,\Gamma)$ where $L$ is a language of {\churchqe} and
$\Gamma$ is a set of formulas in $L$.  The \emph{theory of the logic}
is the theory $T_{\rm logic} = (L_{\cal C},\emptyset)$.  A theory
$T=(L,\Gamma)$ is \emph{eval-free} if each member of $\Gamma$ is
eval-free.  $\textbf{A}_\alpha$ is an \emph{expression of a theory
$T = (L,\Gamma)$} if $\textbf{A}_\alpha \in L$.

\subsection{Definitions and Abbreviations} \label{subsec:definitions}

As Andrews does in~\cite[p.~212]{Andrews02}, we introduce in
Table~\ref{tab:defs} several defined logical constants and
abbreviations.  The former includes constants for true and false and
the propositional connectives.  The latter includes notation for
equality, the propositional connectives, universal and existential
quantification, $\sqsubset_{\epsilon \tarrow \epsilon \tarrow o}$ as
an infix operator, and a simplified notation for evaluations.

\begin{table}
\bc
\begin{tabular}{|lll|}
\hline

$(\textbf{A}_\alpha = \textbf{B}_\alpha)$ \hspace*{1ex}
& stands for  \hspace*{1ex}
& $=_{\alpha \tarrow \alpha \tarrow o} \, \textbf{A}_\alpha \, \textbf{B}_\alpha$.\\

$(\textbf{A}_o \Iff \textbf{B}_o)$ 
& stands for 
& $=_{o \tarrow o \tarrow o} \, \textbf{A}_o \, \textbf{B}_o$.\\

$T_o$ 
& stands for
& $=_{o \tarrow o \tarrow o} \; = \; =_{o \tarrow o \tarrow o}$.\\

$F_o$ 
& stands for
& $(\LambdaApp x_o \mdot T_o) = (\LambdaApp x_o \mdot x_o).$\\

$(\ForallApp \textbf{x}_\alpha \mdot \textbf{A}_o)$ 
& stands for
& $(\LambdaApp \textbf{x}_\alpha \mdot T_o) = (\LambdaApp \textbf{x}_\alpha \mdot \textbf{A}_o)$.\\

$\wedge_{o \tarrow o \tarrow o}$ 
& stands for
& $\LambdaApp x_o \mdot \LambdaApp y_o \mdot {}$\\
& 
& \hspace*{2ex}$((\LambdaApp g_{o \tarrow o \tarrow o} \mdot 
g_{o \tarrow o \tarrow o} \, T_o \, T_o) = {}$\\
&
& \hspace*{3ex}$(\LambdaApp g_{o \tarrow o \tarrow o} \mdot 
g_{o \tarrow o \tarrow o} \, x_o \, y_o)).$\\

$(\textbf{A}_o \And \textbf{B}_o)$ 
& stands for
& $\wedge_{o \tarrow o \tarrow o} \, \textbf{A}_o \, \textbf{B}_o$.\\

$\Implies_{o \tarrow o \tarrow o}$ 
& stands for
& $\LambdaApp x_o \mdot \LambdaApp y_o \mdot (x_o = (x_o \And y_o)).$\\ 

$(\textbf{A}_o \Implies \textbf{B}_o)$ 
& stands for
& ${\Implies_{o \tarrow o \tarrow o}} \, \textbf{A}_o \,\textbf{B}_o$.\\

$\Neg_{o \tarrow o}$ 
& stands for
& ${=_{o \tarrow o \tarrow o}} \, F_o$.\\

$(\Neg\textbf{A}_o)$ 
& stands for
& $\Neg_{o \tarrow o} \, \textbf{A}_o$.\\

$\vee_{o \tarrow o \tarrow o}$ 
& stands for
& $\LambdaApp x_o \mdot \LambdaApp y_o \mdot \Neg (\Neg x_o \And \Neg y_o).$\\

$(\textbf{A}_o \Or \textbf{B}_o)$ 
& stands for
& ${\vee_{o \tarrow o \tarrow o}} \, \textbf{A}_o \, \textbf{B}_o$.\\

$(\ForsomeApp \textbf{x}_\alpha \mdot \textbf{A}_o)$ 
& stands for
& $\Neg(\ForallApp \textbf{x}_\alpha \mdot \Neg\textbf{A}_o)$.\\

$(\textbf{A}_\alpha \not= \textbf{B}_\alpha)$ 
& stands for 
& $\Neg(\textbf{A}_\alpha = \textbf{B}_\alpha)$.\\

$\textbf{A}_\epsilon \sqsubset_{\epsilon \tarrow \epsilon \tarrow o} \textbf{B}_\epsilon$
& stands for
& $\sqsubset_{\epsilon \tarrow \epsilon \tarrow o} \, 
\textbf{A}_\epsilon \, \textbf{B}_\epsilon$.\\

$\sembrack{\textbf{A}_\epsilon}_\beta$ 
& stands for
& $\sembrack{\textbf{A}_\epsilon}_{{\bf B}_\beta}$.\\

\hline
\end{tabular}
\ec
\caption{Definitions and Abbreviations}\label{tab:defs}
\end{table}

\section{Semantics}\label{sec:semantics}

The semantics of {\churchqe} extends the semantics of
{\qzero}~\cite[pp.~238--239]{Andrews02} by defining the domain of the
type $\epsilon$ and what quotations and evaluations mean.

\subsection{Frames}

A \emph{frame} of {\churchqe} is a collection $\set{D_\alpha \;|\;
  \alpha \in \sT}$ of domains such that:

\be

  \item $D_\iota$ is a nonempty set of values (called \emph{individuals}).

  \item $D_o = \set{\TRUE,\FALSE}$, the set of standard \emph{truth
    values}.

  \item $D_\epsilon$ is the set of \emph{constructions} of
    {\churchqe}.

  \item For $\alpha, \beta \in \sT$, $D_{\alpha \tarrow \beta}$ is
    some set of \emph{total functions} from $D_\alpha$ to $D_\beta$.

\ee
$\sD_\iota$ is the \emph{domain of individuals}, $\sD_o$ is the
\emph{domain of truth values}, $\sD_\epsilon$ is the \emph{domain of
  constructions}, and, for $\alpha, \beta \in \sT$, $\sD_{\alpha
  \tarrow \beta}$ is a \emph{function domain}.

\subsection{Interpretations}\label{subsec:interpretations}

An \emph{interpretation} of {\churchqe} is a pair $(\set{D_\alpha
  \;|\; \alpha \in \sT},I)$ consisting of a frame and an
interpretation function $I$ that maps each constant in $\sC$ of type
$\alpha$ to an element of $D_\alpha$ such that:

\be

  \item For all $\alpha \in \sT$, $I(=_{\alpha \tarrow \alpha
    \tarrow o})$ is the function $f \in D_{\alpha \tarrow
    \alpha \tarrow o}$ such that, for all $d_1,d_2 \in D_\alpha$,
    $f(d_1)(d_2) = \TRUE$ iff $d_1 = d_2$.  That is, $I(=_{\alpha
    \tarrow \alpha \tarrow o})$ is the identity relation on
    $D_\alpha$.

  \item $I(\mname{is-var}_{\epsilon \tarrow o})$ is the function $f
    \in D_{\epsilon \tarrow o}$ such that, for all constructions
    $\textbf{A}_\epsilon \in D_\epsilon$, $f(\textbf{A}_\epsilon) =
    \TRUE$ iff $\textbf{A}_\epsilon = \synbrack{\textbf{x}_\alpha}$
    for some variable $\textbf{x}_\alpha \in \sV$ (where $\alpha$ can
    be any type).

  \item For all $\alpha \in \sT$, $I(\mname{is-var}_{\epsilon \tarrow
    o}^{\alpha})$ is the function $f \in D_{\epsilon \tarrow o}$
    such that, for all constructions $\textbf{A}_\epsilon \in
    D_\epsilon$, $f(\textbf{A}_\epsilon) = \TRUE$ iff
    $\textbf{A}_\epsilon = \synbrack{\textbf{x}_\alpha}$ for some
    variable $\textbf{x}_\alpha \in \sV$.

  \item $I(\mname{is-con}_{\epsilon \tarrow o})$ is the function $f \in
    D_{\epsilon \tarrow o}$ such that, for all constructions
    $\textbf{A}_\epsilon \in D_\epsilon$, $f(\textbf{A}_\epsilon) =
    \TRUE$ iff $\textbf{A}_\epsilon = \synbrack{\textbf{c}_\alpha}$
    for some constant $\textbf{c}_\alpha \in \sC$ (where $\alpha$ can
    be any type).

  \item For all $\alpha \in \sT$, $I(\mname{is-con}_{\epsilon \tarrow
    o}^{\alpha})$ is the function $f \in D_{\epsilon \tarrow o}$
    such that, for all constructions $\textbf{A}_\epsilon \in
    D_\epsilon$, $f(\textbf{A}_\epsilon) = \TRUE$ iff
    $\textbf{A}_\epsilon = \synbrack{\textbf{c}_\alpha}$ for some
    constant $\textbf{c}_\alpha \in \sC$.

  \item \bsp $I(\mname{app}_{\epsilon \tarrow \epsilon \tarrow
    \epsilon})$ is the function $f \in D_{\epsilon \tarrow \epsilon
    \tarrow \epsilon}$ such that, for all constructions
    $\textbf{A}_\epsilon, \textbf{B}_\epsilon \in D_\epsilon$,
    $f(\textbf{A}_\epsilon)(\textbf{B}_\epsilon)$ is the construction
    $\mname{app}_{\epsilon \tarrow \epsilon \tarrow \epsilon} \,
    \textbf{A}_\epsilon \, \textbf{B}_\epsilon$. \esp

  \item $I(\mname{abs}_{\epsilon \tarrow \epsilon \tarrow \epsilon})$
    is the function $f \in D_{\epsilon \tarrow \epsilon \tarrow
      \epsilon}$ such that, for all constructions
    $\textbf{A}_\epsilon, \textbf{B}_\epsilon \in D_\epsilon$,
    $f(\textbf{A}_\epsilon)(\textbf{B}_\epsilon)$ is the construction
    $\mname{abs}_{\epsilon \tarrow \epsilon \tarrow \epsilon} \,
    \textbf{A}_\epsilon \, \textbf{B}_\epsilon$.

  \item $I(\mname{quo}_{\epsilon \tarrow \epsilon})$ is the function $f
    \in D_{\epsilon \tarrow \epsilon}$ such that, for all
    constructions $\textbf{A}_\epsilon \in D_\epsilon$,
    $f(\textbf{A}_\epsilon)$ is the construction
    $\mname{quo}_{\epsilon \tarrow \epsilon} \, \textbf{A}_\epsilon$.

  \item $I(\mname{is-expr}_{\epsilon \tarrow o})$ is the function $f \in
    D_{\epsilon \tarrow o}$ such that, for all constructions
    $\textbf{A}_\epsilon \in D_\epsilon$, $f(\textbf{A}_\epsilon) =
    \TRUE$ iff $\textbf{A}_\epsilon = \sE(\textbf{B}_\alpha)$ for some
    (eval-free) expression $\textbf{B}_\alpha$ (where $\alpha$ can be
    any type).

  \item For all $\alpha \in \sT$, $I(\mname{is-expr}_{\epsilon \tarrow
    o}^{\alpha})$ is the function $f \in D_{\epsilon \tarrow o}$
    such that, for all constructions $\textbf{A}_\epsilon \in
    D_\epsilon$, $f(\textbf{A}_\epsilon) = \TRUE$ iff
    $\textbf{A}_\epsilon = \sE(\textbf{B}_\alpha)$ for some
    (eval-free) expression $\textbf{B}_\alpha$.

  \item $I(\sqsubset_{\epsilon \tarrow \epsilon \tarrow o})$ is the
    function $f \in D_{\epsilon \tarrow \epsilon \tarrow o}$ such
    that, for all constructions $\textbf{A}_\epsilon,
    \textbf{B}_\epsilon \in D_\epsilon$,
    $f(\textbf{A}_\epsilon)(\textbf{B}_\epsilon) = \TRUE$ iff
    $\textbf{A}_\epsilon$ is a proper subexpression of
    $\textbf{B}_\epsilon$.

  \item $I(\mname{is-free-in}_{\epsilon \tarrow \epsilon \tarrow o})$
    is the function $f \in D_{\epsilon \tarrow \epsilon \tarrow o}$
    such that, for all constructions $\textbf{A}_\epsilon,
    \textbf{B}_\epsilon \in D_\epsilon$,
    $f(\textbf{A}_\epsilon)(\textbf{B}_\epsilon) = \TRUE$ iff
    $\textbf{A}_\epsilon = \synbrack{\textbf{x}_\alpha}$ for some
    $\textbf{x}_\alpha \in \sV$, $\textbf{B}_\epsilon =
    \sE(\textbf{C}_\beta)$ for some (eval-free) expression
    $\textbf{C}_\beta$, and $\textbf{x}_\alpha$ is free in
    $\textbf{C}_\beta$.

\ee

\begin{rem}[Domain of Constructions]\em
We would prefer that $D_\epsilon$ contains just the set of proper
constructions because we need only proper constructions to represent
the syntax trees of eval-free expressions.  However, then the natural
interpretations of the three syntax constructors ---
$\mname{app}_{\epsilon \tarrow \epsilon \tarrow \epsilon}$,
$\mname{abs}_{\epsilon \tarrow \epsilon \tarrow \epsilon}$, and
$\mname{quo}_{\epsilon \tarrow \epsilon}$ --- would be partial
functions.  Since {\churchqe} admits only total functions, it is more
convenient to allow $D_\epsilon$ to include improper constructions
than to interpret the syntax constructors as total functions that
represent partial functions.
\end{rem}

An \emph{assignment} into a frame $\set{D_\alpha \;|\; \alpha \in
  \sT}$ is a function $\phi$ whose domain is $\sV$ such that
$\phi(\textbf{x}_\alpha) \in D_\alpha$ for each $\textbf{x}_\alpha \in
\sV$.  Given an assignment $\phi$, $\textbf{x}_\alpha \in \sV$, and $d
\in D_\alpha$, let $\phi[\textbf{x}_\alpha \mapsto d]$ be the
assignment $\psi$ such that $\psi(\textbf{x}_\alpha) = d$ and
$\psi(\textbf{y}_\beta) = \phi(\textbf{y}_\beta)$ for all variables
$\textbf{y}_\beta$ distinct from $\textbf{x}_\alpha$.  Given an
interpretation $\sM = (\set{D_\alpha \;|\; \alpha \in \sT}, I)$,
$\mname{assign}(\sM)$ is the set of assignments into the frame of
$\sM$.

\subsection{General Models} \label{subsec:gen-models}

We are now ready to define the notion of a ``general model'' for
{\churchqe} in which function domains need not contain all possible
total functions.  General models were introduced by Leon Henkin
in~\cite{Henkin50}.

An interpretation $\sM = (\set{D_\alpha \;|\; \alpha \in \sT), I}$ is
a \emph{general model} for {\churchqe} if there is a binary valuation
function $V^{\cal M}$ such that, for all assignments $\phi \in
\mname{assign}(\sM)$ and expressions $\textbf{C}_\gamma$, $V^{\cal
  M}_{\phi}(\textbf{C}_\gamma) \in D_\gamma$\footnote{We write
  $V^{\cal M}_{\phi}(\textbf{C}_\gamma)$ instead of $V^{\cal
    M}(\phi,\textbf{C}_\gamma)$.} and each of the following conditions
is satisfied:

\be

  \item If $\textbf{C}_\gamma \in \sV$, then $V^{\cal
    M}_{\phi}(\textbf{C}_\gamma) = \phi(\textbf{C}_\gamma)$.

  \item If $\textbf{C}_\gamma \in \sC$, then $V^{\cal
    M}_{\phi}(\textbf{C}_\gamma) = I(\textbf{C}_\gamma)$.

  \item If $\textbf{C}_\gamma$ is $\textbf{F}_{\alpha \tarrow \beta} \,
    \textbf{A}_\alpha$, then $V^{\cal M}_{\phi}(\textbf{C}_\gamma) =
    V^{\cal M}_{\phi}(\textbf{F}_{\alpha \tarrow \beta})(V^{\cal
      M}_{\phi}(\textbf{A}_\alpha))$.

  \item If $\textbf{C}_\gamma$ is $\LambdaApp \textbf{x}_\alpha \mdot
    \textbf{B}_\beta$, then $V^{\cal M}_{\phi}(\textbf{C}_\gamma)$ is
    the function $f \in D_{\alpha \tarrow \beta}$ such that, for each
    $d \in D_\alpha$, $f(d) = V^{\cal M}_{\phi[{\bf x}_\alpha \mapsto
      d]}(\textbf{B}_\beta)$.

  \item If $\textbf{C}_\gamma$ is $\synbrack{\textbf{A}_\alpha}$, then
    $V^{\cal M}_{\phi}(\textbf{C}_\gamma) = \sE(\textbf{A}_\alpha)$.

  \item If $\textbf{C}_\gamma$ is
    $\sembrack{\textbf{A}_\epsilon}_\beta$ and $V^{\cal
    M}_{\phi}(\mname{is-expr}_{\epsilon \tarrow o}^{\beta} \,
    \textbf{A}_\epsilon) = \TRUE$, then \[V^{\cal
      M}_{\phi}(\textbf{C}_\gamma) = V^{\cal
      M}_{\phi}(\sE^{-1}(V^{\cal M}_{\phi}(\textbf{A}_\epsilon))).\]

  \item For each $\beta \in \sT$, there is some $d_\beta \in D_\beta$
    such that, if $\textbf{C}_\gamma$ is
    $\sembrack{\textbf{A}_\epsilon}_\beta$ and $V^{\cal
      M}_{\phi}(\mname{is-expr}_{\epsilon \tarrow o}^{\beta} \,
    \textbf{A}_\epsilon) = \FALSE$, then $V^{\cal
      M}_{\phi}(\textbf{C}_\gamma) = d_\beta$.

\ee 

\begin{prop} \label{prop:gen-models-exist}
General models for {\churchqe} exist.
\end{prop}

\begin{proof} 
It is easy to construct an interpretation $\sM = (\set{\sD_\alpha
  \;|\; \alpha \in \sT}, I)$ such that $\sD_{\alpha \tarrow \beta}$ is
the set of all total functions from $\sD_\alpha$ to $\sD_\beta$ for
all $\alpha,\beta \in \sT$.  It is also easy to define by induction on
the structure of expressions a valuation function $V^{\cal M}$ such
that $\sM$ is a general model for {\churchqe}.
\end{proof}

\begin{rem}[Semantics of Evaluations]\em\bsp
When $V^{\cal M}_{\phi}(\mname{is-expr}_{\epsilon \tarrow o}^{\beta}
\, \textbf{A}_\epsilon) = \TRUE$, the semantics of $V^{\cal
  M}_{\phi}(\sembrack{\textbf{A}_\epsilon}_\beta)$ involves a double
valuation as mentioned in the Double Substitution Problem (see
section~\ref{sec:introduction}).\esp
\end{rem}

\begin{rem}[Undefined Evaluations]\em
Suppose $V^{\cal M}_{\phi}(\textbf{A}_\epsilon)$ is an improper
construction.  Then $V^{\cal M}_{\phi}(\sE^{-1}(V^{\cal
  M}_{\phi}(\textbf{A}_\epsilon)))$ is undefined and $V^{\cal
  M}_{\phi}(\sembrack{\textbf{A}_\epsilon}_\beta)$ has no natural
value.  Since {\churchqe} does not admit undefined expressions,
$V^{\cal M}_{\phi}(\sembrack{\textbf{A}_\epsilon}_\beta)$ is defined
to be $d_\beta$, an unspecified error value for the type
$\beta$. Similarly, if $V^{\cal M}_{\phi}(\textbf{A}_\epsilon)$ is a
proper construction of the form $\sE(\textbf{B}_\gamma)$ with $\gamma
\not= \beta$, $V^{\cal
  M}_{\phi}(\sembrack{\textbf{A}_\epsilon}_\beta) = d_\beta$.
\end{rem}

Let $\sM$ be a general model for {\churchqe}.  $\textbf{A}_o$ is
\emph{valid in $\sM$}, written $\sM \vDash \textbf{A}_o$, if $V^{\cal
  M}_{\phi}(\textbf{A}_o) = \TRUE$ for all $\phi \in
\mname{assign}(\sM)$.  $\textbf{A}_o$ is \emph{valid in {\churchqe}},
written ${} \vDash \textbf{A}_o$, if $\textbf{A}_o$ is valid in every
general model for {\churchqe}.

Let $T=(L,\Gamma)$ be a theory of {\churchqe}, $\textbf{A}_o$ be a
formula of $T$, and $\sH$ be a set of formulas of $T$.  A
\emph{general model for $T$} is a general model $\sM$ for {\churchqe}
such that $\sM \vDash \textbf{A}_o$ for all $\textbf{A}_o \in
\Gamma$.  $\textbf{A}_o$ is \emph{valid in $T$}, written $T \vDash
\textbf{A}_o$, if $\textbf{A}_o$ is valid in every general model for
$T$.  \emph{$\sH$ entails $\textbf{A}_o$ in $T$}, written $T,\sH
\vDash \textbf{A}_o$, if, for all general models $\sM$ for $T$ and
all $\phi \in \mname{assign}(\sM)$, $V^{\cal M}_{\phi}(\textbf{H}_o) =
\TRUE$ for all $\textbf{H}_o \in \sH$ implies $V^{\cal
  M}_{\phi}(\textbf{A}_o) = \TRUE$.

\begin{prop}\label{prop:val-const}
Let $\sM$ be a general model for {\churchqe}, $\textbf{A}_\epsilon$ be
a construction, and $\phi \in \mname{assign}(\sM)$.  Then $V^{\cal
  M}_{\phi}(\textbf{A}_\epsilon) = \textbf{A}_\epsilon$.
\end{prop}

\begin{proof}
Follows immediately from conditions 6--8 of the definition of an
interpretation and conditions 3 and 5 of the definition of a general
model.
\end{proof}

\begin{thm}[Law of Quotation] \label{thm:sem-quotation}
$\synbrack{\textbf{A}_\alpha} = \sE(\textbf{A}_\alpha)$ is valid in
  {\churchqe}.
\end{thm}

\begin{proof}
Let $\sM$ be a general model for {\churchqe} and $\phi \in
\mname{assign}(\sM)$.  Then
\begin{align} \setcounter{equation}{0}
&
V^{\cal M}_{\phi}(\synbrack{\textbf{A}_\alpha}) \\
&=
\sE(\textbf{A}_\alpha) \\
&=
V^{\cal M}_{\phi}(\sE(\textbf{A}_\alpha))
\end{align}
(2) follows from condition 5 of the definition of a general model and
(3) follows from Proposition~\ref{prop:val-const}. Hence $V^{\cal
  M}_{\phi}(\synbrack{\textbf{A}_\alpha}) = V^{\cal
  M}_{\phi}(\sE(\textbf{A}_\alpha))$ for all $\phi \in
\mname{assign}(\sM)$, and so $\synbrack{\textbf{A}_\alpha} =
\sE(\textbf{A}_\alpha)$ is valid in every general model for
   {\churchqe}.
\end{proof}

\begin{thm}[Law of Disquotation] \label{thm:sem-disquotation}
$\sembrack{\synbrack{\textbf{A}_\alpha}}_\alpha = \textbf{A}_\alpha$
  is valid in {\churchqe}.
\end{thm}

\begin{proof}\bsp
Let $\sM$ be a general model for {\churchqe} and $\phi \in
\mname{assign}(\sM)$.  Then
\begin{align} \setcounter{equation}{0}
&
V^{\cal M}_{\phi}(\sembrack{\synbrack{\textbf{A}_\alpha}}_\alpha) \\
&=
V^{\cal M}_{\phi}(\sE^{-1}(V^{\cal M}_{\phi}(\synbrack{\textbf{A}_\alpha}))) \\
&=
V^{\cal M}_{\phi}(\sE^{-1}(\sE(\textbf{A}_\alpha))) \\
&=
V^{\cal M}_{\phi}(\textbf{A}_\alpha)
\end{align}
(a) $V^{\cal M}_{\phi}(\synbrack{\textbf{A}_\alpha}) =
\sE(\textbf{A}_\alpha)$ is by condition 5 of the definition of a
general model.  (b) $V^{\cal M}_{\phi}(\mname{is-expr}_{\epsilon
  \tarrow o}^{\alpha} \, \synbrack{\textbf{A}_\alpha}) = \TRUE$
follows from (a). (2) follows from (b) and condition 6 of the
definition of a general model; (3) follows from (a); and (4) is
immediate.  Hence $V^{\cal
  M}_{\phi}(\sembrack{\synbrack{\textbf{A}_\alpha}}_\alpha) = V^{\cal
  M}_{\phi}(\textbf{A}_\alpha)$ for all $\phi \in
\mname{assign}(\sM)$, and so
$\sembrack{\synbrack{\textbf{A}_\alpha}}_\alpha = \textbf{A}_\alpha$
is valid in every general model for {\churchqe}.\esp
\end{proof}

\begin{rem}[Evaluation Problem]\em\bsp
Theorem~\ref{thm:sem-disquotation} shows that disquotation holds
universally in {\churchqe} contrary to the Evaluation Problem (see
section~\ref{sec:introduction}).  We have avoided the Evaluation
Problem in {\churchqe} by admitting only quotations of eval-free
expressions and thus making it impossible to express the liar paradox.
If quotations of non-eval-free expressions were allowed in
{\churchqe}, the logic would be significantly more expressive, but
also much more complicated, as seen in
{\qzerouqe}~\cite{FarmerArxiv14}.\esp
\end{rem}

\subsection{Standard Models}

An interpretation $\sM = (\set{\sD_\alpha \;|\; \alpha \in \sT), I}$
is a \emph{standard model} for {\churchqe} if $\sD_{\alpha \tarrow
  \beta}$ is the set of all total functions from $\sD_\alpha$ to
$\sD_\beta$ for all $\alpha,\beta \in \sT$.

\begin{lem}
Standard models for {\churchqe} exist, and every standard model for
{\churchqe} is also a general model for {\churchqe}.
\end{lem}

\begin{proof}
By the proof of Proposition~\ref{prop:gen-models-exist}.
\end{proof}

\bigskip

A general model for {\churchqe} is a \emph{nonstandard model} for
{\churchqe} if it is not a standard model.

\section{Examples}\label{sec:examples}

We will present in this section four examples that illustrate the
utility of the quotation and evaluation facility in {\churchqe}.

\subsection{Reasoning about Syntax}

Reasoning about the syntax of expressions is normally performed in the
metalogic, but in {\churchqe} reasoning about the syntax of eval-free
expressions can be performed in the logic itself.  This is done by
reasoning about constructions (which represent the syntax trees of
eval-free expressions) using quotation and the machinery of
constructions.  Algorithms that manipulate eval-free expressions can
be formalized as functions that manipulate constructions.  The
functions can be executed using beta-reduction, rewriting, and other
kinds of simplification.

As an example, consider the constant
$\mname{make-implication}_{\epsilon \tarrow \epsilon \tarrow
  \epsilon}$ defined as
\[\LambdaApp x_\epsilon \mdot \LambdaApp y_\epsilon \mdot
(\mname{app}_{\epsilon \tarrow \epsilon \tarrow \epsilon} \,
(\mname{app}_{\epsilon \tarrow \epsilon \tarrow \epsilon} \,
\synbrack{\Implies_{o \tarrow o \tarrow o}} \, x_\epsilon) \,
y_\epsilon).\] It can be used to build constructions that
represent implications.  As another example, consider the constant
$\mname{is-app}_{\epsilon \tarrow o}$ defined as
\[\LambdaApp x_\epsilon \mdot
\ForsomeApp y_\epsilon \mdot \ForsomeApp z_\epsilon
\mdot x_\epsilon = (\mname{app}_{\epsilon \tarrow \epsilon
  \tarrow \epsilon} \, y_\epsilon \, z_\epsilon).\]
It can be used to test whether a construction represents a function
application.

Reasoning about syntax is a two-step process: First, a construction is
built using quotation and the machinery of constructions, and second,
the construction is employed using evaluation.  Continuing the example
above, \[\mname{make-implication}_{\epsilon \tarrow \epsilon \tarrow
  \epsilon} \, \synbrack{\textbf{A}_o} \, \synbrack{\textbf{B}_o}\]
builds a construction equivalent to the quotation
$\synbrack{\textbf{A}_o \Implies \textbf{B}_o}$ and
\[\sembrack{\mname{make-implication}_{\epsilon \tarrow \epsilon \tarrow
  \epsilon} \, \synbrack{\textbf{A}_o} \, \synbrack{\textbf{B}_o}}_o\]
employs the construction as the implication $\textbf{A}_o \Implies
\textbf{B}_o$.  We prove these two results within the proof system for
       {\churchqe} in subsection~\ref{subsec:revisit-syntax} (see
       Propositions~\ref{prop:make-impl-a}
       and~\ref{prop:make-impl-b}).

Using this mixture of quotation and evaluation, it is possible to
express the interplay of syntax and semantics that is needed to
formalize syntax-based algorithms that are commonly used in
mathematics~\cite{Farmer13}.  See
subsection~\ref{secsub:meaning-formulas} for an example.

\subsection{Quasiquotation}

Quasiquotation is a parameterized form of quotation in which the
parameters serve as holes in a quotation that are filled with
expressions that denote syntactic values.  It is a very powerful
syntactic device for specifying expressions and defining macros.
Quasiquotation was introduced by Willard Van Orman Quine in 1940 in
the first version of his book \emph{Mathematical
  Logic}~\cite{Quine03}.  It has been extensively employed in the Lisp
family of programming languages~\cite{Bawden99}.\footnote{In Lisp, the
  standard symbol for quasiquotation is the backquote ({\tt `})
  symbol, and thus in Lisp, quasiquotation is usually called
  \emph{backquote}.}

In {\churchqe}, constructing a large quotation from smaller quotations
can be tedious because it requires many applications of syntax
constructors.  Quasiquotation provides a convenient way to construct
big quotations from little quotations.  It can be defined
straightforwardly in {\churchqe}.

A \emph{quasi-expression} of {\churchqe} is defined inductively as
follows:

\be

  \item $\commabrack{\textbf{A}_\epsilon}$ is a quasi-expression
    called an \emph{antiquotation}.

  \item $\textbf{x}_\alpha$ is a quasi-expression.

  \item $\textbf{c}_\alpha$ is a quasi-expression.

  \item If $M$ and $N$ are quasi-expressions, then $(M \, N)$,
    $(\LambdaApp \textbf{x}_\alpha \mdot N)$, $(\LambdaApp
    \commabrack{\textbf{A}_\epsilon} \mdot N)$, and $\synbrack{M}$ are
    quasi-expressions.

\ee

\noindent
A quasi-expression is thus an eval-free expression where one or more
subexpressions have been replaced by antiquotations.  For example,
$\Neg(\textbf{A}_o \And \commabrack{\textbf{B}_\epsilon})$ is a
quasi-expression.  Obviously, every eval-free expression is a
quasi-expression.

Let $\sE'$ be the function mapping quasi-expressions to expressions
of type~$\epsilon$ that is defined inductively as follows:

\be

  \item $\sE'(\commabrack{\textbf{A}_\epsilon}) = \textbf{A}_\epsilon$.

  \item $\sE'(\textbf{x}_\alpha) = \synbrack{\textbf{x}_\alpha}$.

  \item $\sE'(\textbf{c}_\alpha) = \synbrack{\textbf{c}_\alpha}$.

  \item $\sE'(M \, N) = \mname{app}_{\epsilon \tarrow \epsilon \tarrow
    \epsilon} \, \sE'(M) \, \sE'(N)$.

  \item $\sE'(\LambdaApp M \mdot N) = \mname{abs}_{\epsilon \tarrow
    \epsilon \tarrow \epsilon} \, \sE'(M) \, \sE'(N)$.

  \item $\sE(\synbrack{M}) = \mname{quo}_{\epsilon \tarrow \epsilon} \,
    \sE'(M)$.

\ee

\noindent
Notice that $\sE'(M) = \sE(M)$ when $M$ is an eval-free expression.
Continuing our example above, $\sE'(\Neg(\textbf{A}_o \And
\commabrack{\textbf{B}_\epsilon})) = {}$
\[\mname{app}_{\epsilon \tarrow
  \epsilon \tarrow \epsilon} \, \synbrack{\Neg_{o \tarrow o}} \,
(\mname{app}_{\epsilon \tarrow \epsilon \tarrow \epsilon} \,
(\mname{app}_{\epsilon \tarrow \epsilon \tarrow \epsilon}
\synbrack{\wedge_{o \tarrow o \tarrow o}} \, \sE'(\textbf{A}_o))
\, \textbf{B}_\epsilon).\]

A \emph{quasiquotation} is an expression of the form $\synbrack{M}$
where $M$ is a quasi-expression.  Thus every quotation is a
quasiquotation.  The quasiquotation $\synbrack{M}$ serves as an
alternate notation for the expression $\sE'(M)$.  So
$\synbrack{\Neg(\textbf{A}_o \And \commabrack{\textbf{B}_\epsilon})}$
stands for the significantly more verbose expression in the previous
paragraph.  It represents the syntax tree of a negated conjunction in
which the part of the tree corresponding to the second conjunct is
replaced by the syntax tree represented by $\textbf{B}_\epsilon$.  If
$\textbf{B}_\epsilon$ is a quotation $\synbrack{\textbf{C}_o}$, then
the quasiquotation $\synbrack{\Neg(\textbf{A}_o \And
  \commabrack{\synbrack{\textbf{C}_o}})}$ is equivalent to the
quotation $\synbrack{\Neg(\textbf{A}_o \And \textbf{C}_o)}$.

The use of quasiquotation is further illustrated in the next two
subsections.

\subsection{Schemas}

A \emph{schema} is a metalogical expression containing syntactic
variables.  An instance of a schema is a logical expression obtained
by replacing the syntactic variables with appropriate logical
expressions.  In {\churchqe}, a schema can be formalized as a single
logical expression.

For example, consider the \emph{law of excluded middle (LEM)} that is
expressed as the formula schema $A \Or \Neg A$ where $A$ is a
syntactic variable ranging over all formulas.  LEM can be
formalized in {\churchqe} as the universal statement
\[\ForallApp x_\epsilon \mdot 
\mname{is-expr}_{\epsilon \tarrow o}^{o} \, x_\epsilon \Implies
\sembrack{x_\epsilon}_o \Or \Neg \sembrack{x_\epsilon}_o.\] An
instance of this formalization of LEM is any instance of the universal
statement.  Using quasiquotation, LEM could also be formalized in
{\churchqe} as
\[\ForallApp x_\epsilon \mdot 
\mname{is-expr}_{\epsilon \tarrow o}^{o} \, x_\epsilon \Implies
\sembrack{\synbrack{\commabrack{x_\epsilon} \Or \Neg
    \commabrack{x_\epsilon}}}_o.\] In
subsection~\ref{subsec:revisit-schemas}, we prove the first
formalization of LEM within the proof system for {\churchqe}, and we
show how instances of LEM can be derived by instantiating this
formalization (see Theorem~\ref{thm:lem} and
Proposition~\ref{prop:lem-instance}).

If we assume that the domain of the type $\iota$ is the natural
numbers and $\sC$ includes the usual constants of natural number
arithmetic (including a constant $S_{\iota \tarrow \iota}$
representing the successor function), then the (first-order)
\emph{induction schema for Peano arithmetic} can be formalized in
     {\churchqe} as
\begin{align*}
&
\ForallApp f_\epsilon \mdot 
((\mname{is-expr}_{\epsilon \tarrow o}^{\iota \tarrow o} \, f_\epsilon \And
\mname{is-peano}_{\epsilon \tarrow o} \, f_\epsilon) \Implies {} \\
&
\hspace*{2ex}
((\sembrack{f_\epsilon}_{\iota \tarrow o} \, 0 \And
(\ForallApp x_\iota \mdot \sembrack{f_\epsilon}_{\iota \tarrow o} \, x_\iota \Implies
\sembrack{f_\epsilon}_{\iota \tarrow o} \, 
(S_{\iota \tarrow \iota} \, x_\iota)))
\Implies 
\ForallApp x_\iota \mdot \sembrack{f_\epsilon}_{\iota \tarrow o} \, x_\iota))
\end{align*}
where $\mname{is-peano}_{\epsilon \tarrow o} \, f_\epsilon$ holds iff
$f_\epsilon$ represents the syntax tree of a predicate of
first-order Peano arithmetic.  The \emph{induction schema for
  Presburger arithmetic} is exactly the same as the induction schema
for Peano arithmetic except that the predicate
$\mname{is-peano}_{\epsilon \tarrow o}$ is replaced by an appropriate
predicate $\mname{is-presburger}_{\epsilon \tarrow o}$.  Hence it is
possible to directly express both first-order Peano arithmetic and
Presburger arithmetic as theories in {\churchqe}.

\subsection{Meaning Formulas} \label{secsub:meaning-formulas}

Many symbolic algorithms work by manipulating mathematical expressions
in a mathematically meaningful way.  A \emph{meaning formula} for such
an algorithm is a statement that captures the mathematical
relationship between the input and output expressions of the
algorithm.  For example, consider a symbolic differentiation algorithm
that takes as input an expression representing a function (say $x^2$),
repeatedly applies syntactic differentiation rules to the expression,
and then returns as output the final expression that is produced
($2x$).  The intended meaning formula of this algorithm states that
the function ($\LambdaApp x : \mathbb{R} \mdot 2x$) represented by the
output expression $2x$ is the derivative of the function ($\LambdaApp
x : \mathbb{R} \mdot x^2$) represented by the input expression $x^2$.

Meaning formulas are difficult to express in a traditional logic like
first-order logic or simple type theory since there is no way to
directly refer to the syntactic structure of the expressions in the
logic~\cite{Farmer13}.  However, meaning formulas can be easily
expressed in {\churchqe}.

\bsp 
Consider the following example on the differentiation of polynomials.
Let $T_{\mathbb R} = (L_{\cal D},\Gamma)$ be a theory of {\churchqe}
in which the type $\iota$ is specified to be the real numbers
(see~\cite[p.~276--277]{Farmer08}).  We assume that $\sD$ contains the
following constants:

\be

  \item Constants for zero and one: $0_\iota$ and $1_\iota$.

  \item The usual constants for the negation, addition, and
    multiplication of real numbers: $-_{\iota \tarrow \iota}$,
    $+_{\iota \tarrow \iota \tarrow \iota}$, and $*_{\iota \tarrow
      \iota \tarrow \iota}$.

  \item The usual constants for the negation, addition, and
    multiplication of functions: $-_{(\iota \tarrow \iota) \tarrow
      (\iota \tarrow \iota)}$, $+_{(\iota \tarrow \iota) \tarrow
      (\iota \tarrow \iota) \tarrow (\iota \tarrow \iota)}$, and
    $*_{(\iota \tarrow \iota) \tarrow (\iota \tarrow \iota) \tarrow
      (\iota \tarrow \iota)}$.

  \item $\mname{is-diff}_{(\iota \tarrow \iota) \tarrow o}$ such that
    $\mname{is-diff}_{(\iota \tarrow \iota) \tarrow o} \,
    \textbf{F}_{\iota \tarrow \iota}$ holds iff the function
    $\textbf{F}_{\iota \tarrow \iota}$ is differentiable (at every
    point in the domain of $\iota$).

  \item $\mname{deriv}_{(\iota \tarrow \iota) \tarrow (\iota \tarrow
    \iota)}$ such that $\mname{deriv}_{(\iota \tarrow \iota) \tarrow
    (\iota \tarrow \iota)} \, \textbf{F}_{\iota \tarrow \iota}$
    denotes the derivative of the function $\textbf{F}_{\iota \tarrow
      \iota}$.

  \item $\mname{is-poly}_{\epsilon \tarrow o}$ such that
    $\mname{is-poly}_{\epsilon \tarrow o} \, \textbf{A}_\epsilon$
    holds iff $\textbf{A}_\epsilon$ represents a syntax tree of an
    expression of type $\iota$ that has the form of a polynomial with
    variables from $\set{x_\iota,y_\iota}$ and constants from
    $\set{0_\iota,1_\iota}$.

  \item $\mname{poly-diff}_{\epsilon \tarrow \epsilon \tarrow
    \epsilon}$ such that $\mname{poly-diff}_{\epsilon \tarrow \epsilon
    \tarrow \epsilon} \, \textbf{A}_\epsilon \,
    \synbrack{\textbf{x}_\iota}$ denotes the result of applying the
    usual differentiation rules for polynomials to
    $\textbf{A}_\epsilon$ with respect to $\textbf{x}_\iota$.

\ee
When $n$ is a
natural number with $n \ge 2$, let $n_\iota$ be an abbreviation defined
by:

\be

  \item $2_\iota$ stands for $1_\iota + 1_\iota$.

  \item $(n+1)_\iota$ stands for $n_\iota + 1_\iota$.

\ee
To improve readability, $+_{\iota \tarrow \iota \tarrow \iota}$,
$*_{\iota \tarrow \iota \tarrow \iota}$, $+_{(\iota \tarrow \iota)
  \tarrow (\iota \tarrow \iota) \tarrow (\iota \tarrow \iota)}$, and
$*_{(\iota \tarrow \iota) \tarrow (\iota \tarrow \iota) \tarrow (\iota
  \tarrow \iota)}$ are written as infix operators.
$\mname{is-poly}_{\epsilon \tarrow o}$ is recursively defined in
$T_{\mathbb R}$ (using quasiquotation) as:
\begin{align*}
\LambdaApp u_\epsilon \mdot {}
& u_\epsilon = \synbrack{x_\iota} \Or 
  u_\epsilon = \synbrack{y_\iota} \Or 
  u_\epsilon = \synbrack{0_\iota} \Or 
  u_\epsilon = \synbrack{1_\iota} \Or {}\\
& \ForsomeApp v_\epsilon \mdot
  (\mname{is-poly}_{\epsilon \tarrow o} \, v_\epsilon \And
  u_\epsilon = \synbrack{-_{\iota \tarrow \iota} \, 
  \commabrack{v_\epsilon}}) \Or {}\\
& \ForsomeApp v_\epsilon \mdot 
  \ForsomeApp w_\epsilon \mdot 
  (\mname{is-poly}_{\epsilon \tarrow o} \, v_\epsilon \And
  \mname{is-poly}_{\epsilon \tarrow o} \, w_\epsilon \And {}\\
& \hspace*{2ex}u_\epsilon =
  \synbrack{\commabrack{v_\epsilon} +_{\iota \tarrow \iota \tarrow \iota}
  \commabrack{w_\epsilon}}) \Or {}\\
& \ForsomeApp v_\epsilon \mdot
  \ForsomeApp w_\epsilon \mdot 
  (\mname{is-poly}_{\epsilon \tarrow o} \, v_\epsilon \And
  \mname{is-poly}_{\epsilon \tarrow o} \, w_\epsilon \And {}\\
& \hspace*{2ex}u_\epsilon =
  \synbrack{\commabrack{v_\epsilon} *_{\iota \tarrow \iota \tarrow \iota}
  \commabrack{w_\epsilon}}) \Or {}\\
\end{align*}
When $n$ is a natural number, let $\textbf{x}_{\iota}^{n}$ be an
abbreviation defined by:

\be

  \item $\textbf{x}_{\iota}^{0}$ stands for $1_\iota$.

  \item $\textbf{x}_{\iota}^{m+1}$ stands for $
    \textbf{x}_{\iota}^{m} * \textbf{x}_{\iota}$.

\ee
$\mname{poly-diff}_{\epsilon \tarrow \epsilon \tarrow \epsilon}$ is
specified in $T_{\mathbb R}$ (using quasiquotation) by the following
formulas:

\be

  \item $\mname{is-var}_{\epsilon \tarrow o}^{\iota} \, x_{\epsilon}
    \Implies \mname{poly-diff}_{\epsilon \tarrow \epsilon \tarrow \epsilon} \,
    x_{\epsilon}\, x_{\epsilon} = \synbrack{1_\iota}$.

  \item $(\mname{is-var}_{\epsilon \tarrow o}^{\iota} \,
    x_{\epsilon} \And \mname{is-var}_{\epsilon \tarrow
      o}^{\iota} \, y_{\epsilon} \And x_{\epsilon}
    \not= y_{\epsilon})\Implies \mname{poly-diff}_{\epsilon
      \tarrow \epsilon \tarrow \epsilon} \, x_{\epsilon}\,
    y_{\epsilon} = \synbrack{0_\iota}$.

  \item $(\mname{is-con}_{\epsilon \tarrow o}^{\iota} \, x_{\epsilon}
    \And \mname{is-var}_{\epsilon \tarrow o}^{\iota} \,
    y_{\epsilon})\Implies \mname{poly-diff}_{\epsilon \tarrow \epsilon
      \tarrow \epsilon} \, x_{\epsilon}\, y_{\epsilon} =
    \synbrack{0_\iota}$.

  \item $(\mname{is-poly}_{\epsilon \tarrow o} \, x_{\epsilon}
    \And \mname{is-var}_{\epsilon \tarrow o}^{\iota} \, y_{\epsilon})\Implies {}$\\ 
    \hspace*{2ex}$\mname{poly-diff}_{\epsilon \tarrow \epsilon \tarrow \epsilon} \,
    \synbrack{-_{\iota \tarrow \iota} \, \commabrack{x_{\epsilon}}} \, y_{\epsilon} =
    \synbrack{-_{\iota \tarrow \iota} \, 
    \commabrack{\mname{poly-diff}_{\epsilon \tarrow  \epsilon \tarrow \epsilon} \, 
    x_{\epsilon} \, y_{\epsilon}}})$.

  \item $(\mname{is-poly}_{\epsilon \tarrow o} \, x_{\epsilon}
    \And \mname{is-poly}_{\epsilon \tarrow o} \, y_{\epsilon} 
    \And \mname{is-var}_{\epsilon \tarrow o}^{\iota} \, z_{\epsilon}) \Implies {}\\ 
    \hspace*{2ex}\mname{poly-diff}_{\epsilon \tarrow \epsilon \tarrow \epsilon} \,
    \synbrack{\commabrack{x_\epsilon} 
    +_{\iota \tarrow \iota \tarrow \iota} 
    \commabrack{y_\epsilon}} \, z_{\epsilon} = {}\\
    \hspace*{2ex}\synbrack{
    \commabrack{\mname{poly-diff}_{\epsilon \tarrow \epsilon \tarrow \epsilon} \,
    x_\epsilon \, z_\epsilon} 
    +_{\iota \tarrow \iota \tarrow \iota} 
    \commabrack{\mname{poly-diff}_{\epsilon \tarrow \epsilon \tarrow \epsilon} \,
    y_\epsilon \, z_\epsilon}}$.

  \item $(\mname{is-poly}_{\epsilon \tarrow o} \, x_{\epsilon}
    \And \mname{is-poly}_{\epsilon \tarrow o} \, y_{\epsilon} 
    \And \mname{is-var}_{\epsilon \tarrow o}^{\iota} \, z_{\epsilon}) \Implies {}\\ 
    \hspace*{2ex}\mname{poly-diff}_{\epsilon \tarrow \epsilon \tarrow \epsilon} \,
    \synbrack{\commabrack{x_\epsilon} 
    *_{\iota \tarrow \iota \tarrow \iota} 
    \commabrack{y_\epsilon}} \, z_{\epsilon} = {}\\
    \hspace*{2ex}\synbrack{
    (\commabrack{(\mname{poly-diff}_{\epsilon \tarrow \epsilon \tarrow \epsilon} \,
    x_\epsilon \, z_\epsilon)} *_{\iota \tarrow \iota \tarrow \iota} 
    \commabrack{y_\epsilon})
    +_{\iota \tarrow \iota \tarrow \iota} {}\\
    \hspace*{3.2ex}(\commabrack{x_\epsilon} 
    *_{\iota \tarrow \iota \tarrow \iota}
    \commabrack{(\mname{poly-diff}_{\epsilon \tarrow \epsilon \tarrow \epsilon} \,
    y_\epsilon \, z_\epsilon)})}$.

\ee
The value of an application of $\mname{poly-diff}_{\epsilon \tarrow
  \epsilon \tarrow \epsilon}$ can be computed using these formulas as
conditional rewrite rules.
\esp

The meaning formula for $\mname{poly-diff}_{\epsilon \tarrow \epsilon
  \tarrow \epsilon}$ can then be given as
\begin{align*}
&
\ForallApp u_\epsilon \mdot \ForallApp v_\epsilon \mdot
((\mname{is-poly}_{\epsilon \tarrow o} \, u_\epsilon) 
\And \mname{is-var}_{\epsilon \tarrow o}^{\iota} \, v_\epsilon \Implies {} \\
&
\hspace*{2ex}
\mname{deriv}_{(\iota \tarrow \iota) \tarrow (\iota \tarrow \iota)}
(\sembrack{\mname{abs}_{\epsilon \tarrow  \epsilon \tarrow \epsilon} \,
v_\epsilon \, u_\epsilon}_{\iota \tarrow \iota}) = {} \\
&
\hspace*{2ex}
\sembrack{\mname{abs}_{\epsilon \tarrow  \epsilon \tarrow \epsilon} \,
v_\epsilon \, 
(\mname{poly-diff}_{\epsilon \tarrow \epsilon \tarrow \epsilon} \, 
u_\epsilon \, v_\epsilon)}_{\iota \tarrow \iota}).\footnotemark
\end{align*}
\bsp\noindent
Unfortunately, we are not able to prove this formula in the proof
system for {\churchqe} presented later in the paper.  So instead we
will define the \emph{meaning formula for $\mname{poly-diff}_{\epsilon
    \tarrow \epsilon \tarrow \epsilon}$} as the formula schema
\begin{align*}
&
\ForallApp u_\epsilon \mdot
(\mname{is-poly}_{\epsilon \tarrow o} \, u_\epsilon \Implies {} \\
&
\hspace*{2ex}
\mname{deriv}_{(\iota \tarrow \iota) \tarrow (\iota \tarrow \iota)}
(\LambdaApp \textbf{x}_\iota \mdot \sembrack{u_\epsilon}_\iota) = 
\LambdaApp \textbf{x}_\iota \mdot
\sembrack{\mname{poly-diff}_{\epsilon \tarrow \epsilon \tarrow \epsilon} \, 
u_\epsilon \, \synbrack{\textbf{x}_\iota}}_\iota)
\end{align*}
where $\textbf{x}_\iota$ is either $x_\iota$ or $y_\iota$.  As
expected, the equation \[\mname{deriv}_{(\iota \tarrow \iota) \tarrow
  (\iota \tarrow \iota)} (\LambdaApp x_\iota \mdot x_{\iota}^{2}) =
\LambdaApp x_\iota \mdot 2 * x_{\iota}\] follows from this meaning
formula and the definitions of $\mname{is-poly}_{\epsilon \tarrow o}$
and $\mname{poly-diff}_{\epsilon \tarrow \epsilon \tarrow \epsilon}$.
\esp

In subsection~\ref{subsec:revisit-meaning}, we prove the meaning
formula for $\mname{poly-diff}_{\epsilon \tarrow \epsilon \tarrow
  \epsilon}$ and the application of it given above within the proof
system for {\churchqe} (see Theorem~\ref{thm:meaning-form} and
Proposition~\ref{prop:mf-app}).

\footnotetext{We restrict this example to polynomials since polynomial
  functions and their derivatives are always total.  Thus issues of
  undefinedness do not arise in the formulation of the meaning formula
  for $\mname{poly-diff}_{\epsilon \tarrow \epsilon \tarrow
    \epsilon}$.}

\section{Proof System}\label{sec:proof-system}

We present in this section the proof system for {\churchqe} which is
an extension of Andrews' elegant proof system for {\qzero} given
in~\cite[p.~213]{Andrews02}.

\subsection{Requirements}\label{subsec:requirements}

At first glance, it would appear that a proof system for {\churchqe}
could be straightforwardly developed by extending Andrews' proof
system for {\qzero}.  We can define $\mname{is-var}_{\epsilon \tarrow
  o}$, $\mname{is-var}_{\epsilon \tarrow o}^{\alpha}$,
$\mname{is-con}_{\epsilon \tarrow o}$, $\mname{is-con}_{\epsilon
  \tarrow o}^{\alpha}$ by axiom schemas.  We can also define
$\mname{is-expr}_{\epsilon \tarrow o}$, $\mname{is-expr}_{\epsilon
  \tarrow o}^{\alpha}$, $\sqsubset_{\epsilon \tarrow \epsilon \tarrow
  o}$, and $\mname{is-free-in}_{\epsilon \tarrow \epsilon \tarrow o}$,
by axiom schemas using recursion on the construction of expressions.
We can specify that the type $\epsilon$ of constructions is an
inductive type using a set of axioms that say (1)~the constructions
are distinct from each other and (2)~induction holds for
constructions.  We can specify quotation using the Law of Quotation
$\synbrack{\textbf{A}_\alpha} = \sE(\textbf{A}_\alpha)$
(Theorem~\ref{thm:sem-quotation}).  And we can specify evaluation
using the Law of Disquotation
$\sembrack{\synbrack{\textbf{A}_\alpha}}_\alpha = \textbf{A}_\alpha$
(Theorem~\ref{thm:sem-disquotation}).

Andrews' proof system with these added axioms would enable simple
theorems about closed expressions involving quotation and evaluation
to be proved, but the proof system would not be able to substitute
expressions for free variables occurring within evaluations.  Hence
axiom schemas and meaning formulas could be expressed in {\churchqe},
but they would be useless because they could not be instantiated using
the proof system.  Clearly, a useful proof system for {\churchqe}
requires some form of substitution that is applicable to evaluations.
The proof system we have described above would also not be able to
reduce evaluations that contain free variables.  Hence it would not be
possible to prove schemas and meaning formulas.  A useful proof system
for {\churchqe} needs a means to reduce evaluations applied to
expressions that are not just quotations.

To be useful, a proof system for {\churchqe} needs to satisfy the
following requirements:

\be

  \item \textbf{R1}. The proof system is sound, i.e., it only proves
    valid formulas.

  \item \textbf{R2}. The proof system is complete with respect to the
    (general models) semantics for {\churchqe} for eval-free formulas.

  \item \textbf{R3}. The proof system can be used to reason about
    constructions that are denoted by quotations and other expressions of
    type $\epsilon$.

  \item \textbf{R4}. The proof system can instantiate free variables
    occurring within evaluations as found in formulas that
    represent schemas and meaning formulas.

  \item \textbf{R5}. The proof system can prove formulas, such as
    those that represent schemas and meaning formulas, in which free
    variables occur within evaluations.

\ee
In this section we present a proof system that is intended to satisfy
requirements R1--5.  In subsequent sections we will show that this
proof system actually does satisfy these requirements.

\subsection{Challenges}

A proof system for a logic with quotation and evaluation must solve
the three design problems presented in section~\ref{sec:introduction}.
The Evaluation Problem is not an issue in {\churchqe} since quotations
do not contain evaluations.  However, the Variable and Double
Substitution Problems are serious issues for {\churchqe}.  Due to the
Variable Problem, concepts concerning free and bound variables cannot
be purely syntactic, and the due to the Double Substitution Problem, a
substitution must be performed twice in some cases.

The notion of a variable being free in an expression and the notion of
a substitution of an expression for the free occurrences of a variable
in another expression are two crucial concepts affected by the
Variable Problem.  To express the axioms and rules of inference of a
proof system for {\churchqe}, we need a way to say ``a variable
$\textbf{x}_\alpha$ is free in an expression $\textbf{B}_\beta$''.  We
defined this concept in subsection~\ref{subsec:expressions} in the
usual way when $\textbf{B}_\beta$ is eval-free as a purely syntactic
assertion expressed in the metalogic of {\churchqe}.  However, when
$\textbf{B}_\beta$ is not eval-free, this approach does not work
because whether $\textbf{x}_\alpha$ occurs in $\textbf{B}_\beta$ may
depend on the semantics of $\textbf{B}_\beta$.

An approach taking advantage of {\churchqe}'s facility for reasoning
about the syntax of expressions is to introduce a logical constant
$\mname{is-free-in}_{\epsilon \tarrow \epsilon \tarrow o}$ such that
\[\mname{is-free-in}_{\epsilon \tarrow \epsilon
  \tarrow o} \, \synbrack{\textbf{x}_\alpha} \,
\synbrack{\textbf{B}_\beta}\] holds iff ``$\textbf{x}_\alpha$ is free
in $\textbf{B}_\beta$''.  This works fine when $\textbf{B}_\beta$ is
eval-free, but it does not work at all when $\textbf{B}_\beta$ is not
eval-free since once again quotations cannot contain evaluations.
(This approach does work in {\qzerouqe}~\cite{FarmerArxiv14} in which
quotations in {\churchqe} may contain evaluations.)  Nevertheless, the
constant $\mname{is-free-in}_{\epsilon \tarrow \epsilon \tarrow o}$ is
included in {\churchqe} so that semantic statements of the
form \[\mname{is-free-in}_{\epsilon \tarrow \epsilon \tarrow o} \,
\synbrack{\textbf{x}_\alpha} \, \textbf{A}_\epsilon,\] can be
expressed where $\textbf{A}_\epsilon$ is not a quotation and may even
contain free variables.

So what do we do when $\textbf{B}_\beta$ is not eval-free?  Our
approach is to use the more restrictive semantic notion
``$\textbf{x}_\alpha$ is effective in $\textbf{B}_\beta$'' in place of
the syntactic notion ``$\textbf{x}_\alpha$ is free in
$\textbf{B}_\beta$'' when $\textbf{B}_\beta$ is possibly not
eval-free.  ``$\textbf{x}_\alpha$ is effective in $\textbf{B}_\beta$''
means that the value of $\textbf{B}_\beta$ depends on the value of
$\textbf{x}_\alpha$.  Clearly, if $\textbf{B}_\beta$ is eval-free,
``$\textbf{x}_\alpha$ is effective in $\textbf{B}_\beta$'' implies
``$\textbf{x}_\alpha$ is free in $\textbf{B}_\beta$''.  However,
``$\textbf{x}_\alpha$ is effective in ${B}_\beta$'' is a refinement of
``$\textbf{x}_\alpha$ is free in $\textbf{B}_\beta$'' on eval-free
expressions since $\textbf{x}_\alpha$ is free in $\textbf{x}_\alpha =
\textbf{x}_\alpha$, but $\textbf{x}_\alpha$ is not effective in
$\textbf{x}_\alpha = \textbf{x}_\alpha$.

We express ``$\textbf{x}_\alpha$ is effective in $\textbf{B}_\beta$''
in {\churchqe} by the abbreviation
\[\mname{IS-EFFECTIVE-IN}(\textbf{x}_\alpha,\textbf{B}_\beta) 
{\dblsp} \mbox{stands for} {\dblsp} \ForsomeApp \textbf{y}_\alpha \mdot
((\LambdaApp \textbf{x}_\alpha \mdot \textbf{B}_\beta) \,
\textbf{y}_\alpha \not= \textbf{B}_\beta)\] where $\textbf{y}_\alpha$ is
any variable of type $\alpha$ that differs from $\textbf{x}_\alpha$.
We will prove later (Lemma~\ref{lem:not-effective-sem}) that
$\mname{IS-EFFECTIVE-IN}(\textbf{x}_\alpha,\textbf{B}_\beta)$ holds
precisely when the value of $\textbf{B}_\beta$ depends on the value of
$\textbf{x}_\alpha$.  The following are examples of valid formulas in
{\churchqe} involving $\mname{IS-EFFECTIVE-IN}$:

\be

  \item $(\ForsomeApp \textbf{u}_\alpha \mdot \ForsomeApp
    \textbf{v}_\alpha \mdot \textbf{u}_\alpha \not= \textbf{v}_\alpha)
    \Implies \mname{IS-EFFECTIVE-IN}(\textbf{x}_\alpha,
    \textbf{x}_\alpha)$.

  \item $\mname{IS-EFFECTIVE-IN}(\textbf{x}_\epsilon,
    \sembrack{\textbf{x}_\epsilon}_\alpha)$.

  \item $((\ForsomeApp \textbf{u}_\alpha \mdot \ForsomeApp
    \textbf{v}_\alpha \mdot \textbf{u}_\alpha \not= \textbf{v}_\alpha)
    \And \textbf{y}_\epsilon = \synbrack{\textbf{x}_\alpha}) \Implies
    \mname{IS-EFFECTIVE-IN}(\textbf{x}_\alpha,
    \sembrack{\textbf{y}_\epsilon}_\alpha)$.

  \item $\Neg \mname{IS-EFFECTIVE-IN}(\textbf{x}_\alpha, \LambdaApp
    \textbf{x}_\alpha \mdot \textbf{B}_\beta)$.

  \item $\Neg
    \mname{IS-EFFECTIVE-IN}(\textbf{x}_\alpha,\textbf{B}_\beta =
    \textbf{B}_\beta)$.

  \item $\Neg
    \mname{IS-EFFECTIVE-IN}(\textbf{x}_\alpha,\textbf{B}_\beta \Or
    \Neg \textbf{B}_\beta)$.

\ee

As a consequence of the Variable Problem, substitution involving
evaluations cannot be purely syntactic as in a traditional logic.  It
must be a semantics-dependent operation in which side conditions, like
whether a variable is free in an expression, are proved within the
proof system.  Since {\churchqe} supports reasoning about syntax, an
obvious way forward is to add to $\sC$ a logical constant
$\mname{sub}_{\epsilon \tarrow \epsilon \tarrow \epsilon \tarrow
  \epsilon}$ such that \[\mname{sub}_{\epsilon \tarrow \epsilon
  \tarrow \epsilon \tarrow \epsilon} \, \synbrack{\textbf{A}_\alpha}
\, \synbrack{\textbf{x}_\alpha} \, \synbrack{\textbf{B}_\beta} =
\synbrack{\textbf{C}_\beta}\] holds iff ``$\textbf{C}_\beta$ is the
result of substituting $\textbf{A}_\alpha$ for each free occurrence of
$\textbf{x}_\alpha$ in $\textbf{B}_\beta$ without any variable
captures''.  $\mname{sub}_{\epsilon \tarrow \epsilon \tarrow \epsilon
  \tarrow \epsilon}$ thus plays the role of an explicit substitution
operator~\cite{AbadiEtAl91}.  This approach, however, does not work in
{\churchqe} since $\textbf{B}_\beta$ may contain evaluations, but
quotations in {\churchqe} may not contain evaluations.  (Although the
approach does work in {\qzerouqe}~\cite{FarmerArxiv14} in which
quotations in {\churchqe} may contain evaluations, it is extremely
complicated due to the Evaluation Problem.)

A more promising approach is to add axiom schemas to the five
beta-reduction axiom schemas used by Andrews' in his proof system for
{\qzero}~\cite[p.~213]{Andrews02} that specify beta-reduction of
applications of the form $(\LambdaApp \textbf{x}_\alpha \mdot
\synbrack{\textbf{B}_\beta}) \, \textbf{A}_\alpha$ and
$(\LambdaApp \textbf{x}_\alpha \mdot
\sembrack{\textbf{B}_\epsilon}_\beta) \, \textbf{A}_\alpha$.

But how do we solve the Double Substitution Problem in applications of
the second form?  There seems to be no easy way of emulating a double
substitution with beta-reduction, so the best approach appears to be
to consider only cases that do not require a second substitution, as
formalized by the axiom schema B11.2 given below which uses the
constant $\mname{is-free-in}_{\epsilon \tarrow \epsilon \tarrow o}$.
We will see that being able to beta-reduce lambda applications that
correspond to double substitutions is not a necessity and that a proof
system without this capability can prove many useful theorems
involving evaluations.

\bsp In summary, we address the problem of stating
``$\textbf{x}_\alpha$ is free in $\textbf{B}_\beta$'' by (1)
expressing it either syntactically in the usual way or semantically as
$\mname{is-free-in}_{\epsilon \tarrow \epsilon \tarrow o} \,
\synbrack{\textbf{x}_\alpha} \, \synbrack{\textbf{B}_\beta}$ when
$\textbf{B}_\beta$ is eval-free and (2) expressing it more tightly as
$\mname{IS-EFFECTIVE-IN}(\textbf{x}_\alpha,\textbf{B}_\beta)$ when
$\textbf{B}_\beta$ is not restricted.  And we address the problem of
defining substitution by augmenting the beta-reduction axioms of for
{\qzero} to cover quotations and evaluations.
\esp

\subsection{Axioms and Rules of Inference}

The proof system for {\churchqe} consists of the axioms for {\qzero}
(the A axioms), the single rule of inference for {\qzero} named R, and
a set of additional axioms (the B axioms).

\subsubsection{Andrews' Proof System}

Andrews' proof system for {\qzero} is complete with respect to the
general models semantics for {\qzero}~\cite[5502 on
  p.~213]{Andrews02}.  Its axioms and single rule of inference (Rule
R) are the core of the proof system for {\churchqe}.  Axioms A1--3 are
used without modification.  Andrews' five beta-reduction axiom schemas
\cite[p.~213]{Andrews02} have been replaced with a slightly modified
set of axiom schemas.  In particular, the syntactic variables in axiom
schemas A4.1--6 range over all expressions of {\churchqe} (but
restrictions are placed on the syntactic variables in A4.5).  Rule R
is modified slightly to work correctly with quotations and
evaluations.

\bi

  \item[] \textbf{A1 (Truth Values)}

  \be

    \item[] $(g_{o \tarrow o} \, T_o \And g_{o \tarrow o} \, F_o) =
    \ForallApp x_o \mdot g_{o \tarrow o} \, x_o$.

  \ee

  \item[] \textbf{A2 (Leibniz' Law)}

  \be

    \item[] $x_\alpha = y_\alpha \Implies h_{\alpha \tarrow o} \,
      x_\alpha = h_{\alpha \tarrow o} \, y_\alpha$.

  \ee

  \item[] \textbf{A3 (Extensionality)}

  \be

    \item[] $(f_{\alpha \tarrow \beta} = g_{\alpha \tarrow \beta}) =
      \ForallApp x_\alpha \mdot (f_{\alpha \tarrow \beta} \, x_\alpha =
      g_{\alpha \tarrow \beta} \, x_\alpha)$.

  \ee

  \item[] \textbf{A4 (Beta-Reduction)}

  \be

    \item $(\LambdaApp \textbf{x}_\alpha \mdot \textbf{y}_\beta) \,
      \textbf{A}_\alpha = \textbf{y}_\beta$ {\sglsp} where
      $\textbf{x}_\alpha$ and $\textbf{y}_\beta$ are distinct.

    \item $(\LambdaApp \textbf{x}_\alpha \mdot \textbf{x}_\alpha) \,
      \textbf{A}_\alpha = \textbf{A}_\alpha$.

    \item $(\LambdaApp \textbf{x}_\alpha \mdot \textbf{c}_\beta) \,
      \textbf{A}_\alpha = \textbf{c}_\beta$.

    \item $(\LambdaApp \textbf{x}_\alpha \mdot (\textbf{B}_{\beta
      \tarrow \gamma} \, \textbf{C}_\beta)) \, \textbf{A}_\alpha =
      ((\LambdaApp \textbf{x}_\alpha \mdot \textbf{B}_{\beta \tarrow
      \gamma}) \, \textbf{A}_\alpha) \, ((\LambdaApp \textbf{x}_\alpha
      \mdot \textbf{C}_\beta) \, \textbf{A}_\alpha)$.

    \item $(\LambdaApp \textbf{x}_\alpha \mdot \LambdaApp
      \textbf{y}_\beta \mdot \textbf{B}_\gamma) \, \textbf{A}_\alpha =
      \LambdaApp \textbf{y}_\beta \mdot ((\LambdaApp \textbf{x}_\alpha
      \mdot \textbf{B}_\gamma) \, \textbf{A}_\alpha)$ {\sglsp} where
      $\textbf{x}_\alpha$ and $\textbf{y}_\beta$ are distinct and
      either (1) $\textbf{A}_\alpha$ is eval-free and
      $\textbf{y}_\beta$ is not free in $\textbf{A}_\alpha$ or (2)
      $\textbf{B}_\gamma$ is eval-free and $\textbf{x}_\alpha$ is not
      free in $\textbf{B}_\gamma$.

    \item $(\LambdaApp \textbf{x}_\alpha \mdot \LambdaApp
      \textbf{x}_\alpha \mdot \textbf{B}_\beta) \, \textbf{A}_\alpha =
      \LambdaApp \textbf{x}_\alpha \mdot \textbf{B}_\beta$.

  \ee

  \item[] \textbf{Rule R (Equality Substitution)} From
    $\textbf{A}_\alpha = \textbf{B}_\alpha$ and $\textbf{C}_o$ infer
    the result of replacing one occurrence of $\textbf{A}_\alpha$ in
    $\textbf{C}_o$ by an occurrence of $\textbf{B}_\alpha$, provided
    that the occurrence of $\textbf{A}_\alpha$ in $\textbf{C}_o$ is
    not within a quotation, not the first argument of a function
    abstraction, and not the second argument of an evaluation.

\ei

\subsubsection{Axioms for Constructions}

The axioms in this part (1) define the logical constants
$\mname{is-var}_{\epsilon \tarrow o}$, $\mname{is-var}_{\epsilon
  \tarrow o}^{\alpha}$, $\mname{is-con}_{\epsilon \tarrow o}$,
$\mname{is-con}_{\epsilon \tarrow o}^{\alpha}$,
$\mname{is-expr}_{\epsilon \tarrow o}$, and $\mname{is-expr}_{\epsilon
  \tarrow o}^{\alpha}$, (2) specify the type of constructions as an
inductive type, (3) define $\mname{is-free-in}_{\epsilon \tarrow
  \epsilon \tarrow o}$ so that
\[\mname{is-free-in}_{\epsilon \tarrow \epsilon \tarrow o} \,
\synbrack{\textbf{x}_\alpha} \, \synbrack{\textbf{B}_\beta}\] is
provable iff $\textbf{x}_\alpha$ is free in $\textbf{B}_\beta$
(Lemma~\ref{lem:syn-is-free-in}).  B1.2--4, B2.2--4, B3.1-4, B3.6, and
B4.6--7 are axiom schemas; the rest are individual axioms.

By virtue of the axioms given in this part, the proof system for
{\churchqe} satisfies Requirement R3.

\bi

  \item[] \textbf{B1 (Definitions of $\mname{is-var}_{\epsilon \tarrow
      o}$ and $\mname{is-var}_{\epsilon \tarrow o}^{\alpha}$)}

  \be

    \item $\mname{is-var}_{\epsilon \tarrow o}^{\alpha} \, x_\epsilon
      \Implies \mname{is-var}_{\epsilon \tarrow o} \, x_\epsilon$.

    \item $\mname{is-var}_{\epsilon \tarrow o}^{\alpha} \,
      \synbrack{\textbf{x}_\alpha}.$

    \item $\Neg(\mname{is-var}_{\epsilon \tarrow o}^{\alpha} \,
      \synbrack{\textbf{x}_\beta})$ {\sglsp} where $\alpha \not=
      \beta$.

    \item $\NegAlt (\mname{is-var}_{\epsilon \tarrow o} \,
      \textbf{A}_\epsilon)$ {\sglsp} where $\textbf{A}_\epsilon$ is a
      construction not of the form $\synbrack{\textbf{x}_\alpha}$.

  \ee

  \item[] \textbf{B2 (Definitions of $\mname{is-con}_{\epsilon \tarrow
      o}$ and $\mname{is-con}_{\epsilon \tarrow o}^{\alpha}$)}

  \be

    \item $\mname{is-con}_{\epsilon \tarrow o}^{\alpha} \, x_\epsilon
      \Implies \mname{is-con}_{\epsilon \tarrow o} \, x_\epsilon$.

    \item $\mname{is-con}_{\epsilon \tarrow o}^{\alpha} \,
      \synbrack{\textbf{c}_\alpha}.$

    \item $\Neg(\mname{is-con}_{\epsilon \tarrow o}^{\alpha} \,
      \synbrack{\textbf{c}_\beta})$ {\sglsp} where $\alpha \not=
      \beta$.

    \item $\Neg (\mname{is-con}_{\epsilon \tarrow o} \,
      \textbf{A}_\epsilon)$ {\sglsp} where $\textbf{A}_\epsilon$ is a
      construction not of the form $\synbrack{\textbf{c}_\alpha}$.

  \ee

  \item[] \textbf{B3 (Definitions of $\mname{is-expr}_{\epsilon \tarrow
      o}$ and $\mname{is-expr}_{\epsilon \tarrow o}^{\alpha}$)}

  \be

    \item $\mname{is-expr}_{\epsilon \tarrow o}^{\alpha} \, x_\epsilon
      \Implies \mname{is-expr}_{\epsilon \tarrow o} \, x_\epsilon$.

    \item $(\mname{is-var}_{\epsilon \tarrow o}^{\alpha} \, x_\epsilon \Or
      \mname{is-con}_{\epsilon \tarrow o}^{\alpha} \, x_\epsilon)
      \Implies \mname{is-expr}_{\epsilon \tarrow o}^{\alpha} x_\epsilon $.

    \item $(\mname{is-expr}_{\epsilon \tarrow o}^{\alpha \tarrow \beta} \, x_\epsilon \And 
      \mname{is-expr}_{\epsilon \tarrow o}^{\alpha} \, y_\epsilon) \Implies
      \mname{is-expr}_{\epsilon \tarrow o}^{\beta} \, 
      (\mname{app}_{\epsilon \tarrow \epsilon \tarrow \epsilon} \, x_\epsilon \, y_\epsilon)$.

    \item $(\mname{is-var}_{\epsilon \tarrow o}^{\alpha} \, x_\epsilon \, \And \,
      \mname{is-expr}_{\epsilon \tarrow o}^{\beta} \, y_\epsilon) \Implies
      \mname{is-expr}_{\epsilon \tarrow o}^{\alpha \tarrow \beta} \, 
      (\mname{abs}_{\epsilon \tarrow \epsilon \tarrow \epsilon} \, x_\epsilon \, y_\epsilon)$.

    \item $\mname{is-expr}_{\epsilon \tarrow o} \, x_\epsilon \Implies
      \mname{is-expr}_{\epsilon \tarrow o}^{\epsilon} \,
      (\mname{quo}_{\epsilon \tarrow \epsilon} \, x_\epsilon)$.

    \item $(\mname{is-expr}_{\epsilon \tarrow o}^{\alpha} \, x_\epsilon
      \And \mname{is-expr}_{\epsilon \tarrow o}^{\beta} \, y_\epsilon)
      \Implies \Neg(\mname{is-expr}_{\epsilon \tarrow o} \,
      (\mname{app}_{\epsilon \tarrow \epsilon \tarrow \epsilon} \,
      x_\epsilon \, y_\epsilon))$\\ where $\alpha$ is not of the form
      $\beta \tarrow \gamma$.

    \item $(\Neg(\mname{is-expr}_{\epsilon \tarrow o} \, x_\epsilon)
      \Or \Neg(\mname{is-expr}_{\epsilon \tarrow o} \, y_\epsilon))
      \Implies \Neg(\mname{is-expr}_{\epsilon \tarrow o} \,
      (\mname{app}_{\epsilon \tarrow \epsilon \tarrow \epsilon} \,
      x_\epsilon \, y_\epsilon))$.

    \item $(\Neg(\mname{is-var}_{\epsilon \tarrow o} \, x_\epsilon)
      \Or \Neg(\mname{is-expr}_{\epsilon \tarrow o} \, y_\epsilon))
      \Implies \Neg(\mname{is-expr}_{\epsilon \tarrow o} \,
      (\mname{abs}_{\epsilon \tarrow \epsilon \tarrow \epsilon} \,
      x_\epsilon \, y_\epsilon))$.

    \item $\Neg(\mname{is-expr}_{\epsilon \tarrow o} \, x_\epsilon)
      \Implies \Neg(\mname{is-expr}_{\epsilon \tarrow o} \,
      (\mname{quo}_{\epsilon \tarrow \epsilon} \, x_\epsilon))$.

  \ee 

  \item[] \textbf{B4 (Constructions are Distinct)}

  \be

    \item $\mname{is-var}_{\epsilon \tarrow o} \, x_\epsilon \Implies \\[.5ex] 
      \hspace*{2ex} (\Neg(\mname{is-con}_{\epsilon \tarrow o} \,
      x_\epsilon) \And x_\epsilon \not= \mname{app}_{\epsilon \tarrow \epsilon
        \tarrow \epsilon} \, y_\epsilon \, z_\epsilon \And\\[.5ex] 
      \hspace*{3ex} x_\epsilon \not= \mname{abs}_{\epsilon \tarrow
        \epsilon \tarrow \epsilon} \, y_\epsilon \, z_\epsilon \And
      x_\epsilon \not= \mname{quo}_{\epsilon \tarrow \epsilon} \,
      y_\epsilon)$.

    \item $\mname{is-con}_{\epsilon \tarrow o} \, x_\epsilon \Implies\\[.5ex] 
      \hspace*{2ex} (\Neg(\mname{is-var}_{\epsilon \tarrow o} \,
      x_\epsilon) \And x_\epsilon \not= \mname{app}_{\epsilon \tarrow \epsilon
        \tarrow \epsilon} \, y_\epsilon \, z_\epsilon \And\\[.5ex] 
      \hspace*{3ex} x_\epsilon \not= \mname{abs}_{\epsilon \tarrow
        \epsilon \tarrow \epsilon} \, y_\epsilon \, z_\epsilon \And
      x_\epsilon \not= \mname{quo}_{\epsilon \tarrow \epsilon} \,
      y_\epsilon)$.

    \item $x_\epsilon = \mname{app}_{\epsilon \tarrow \epsilon
        \tarrow \epsilon} \, y_\epsilon \, z_\epsilon \Implies\\[.5ex] 
      \hspace*{2ex} (\Neg(\mname{is-var}_{\epsilon \tarrow o} \,
      x_\epsilon) \And \Neg(\mname{is-con}_{\epsilon \tarrow o} \,
      x_\epsilon) \And\\[.5ex] 
      \hspace*{3ex} x_\epsilon \not= \mname{abs}_{\epsilon \tarrow
        \epsilon \tarrow \epsilon} \, u_\epsilon \, v_\epsilon \And
      x_\epsilon \not= \mname{quo}_{\epsilon \tarrow \epsilon} \,
      u_\epsilon)$.

    \item $x_\epsilon = \mname{abs}_{\epsilon \tarrow \epsilon
        \tarrow \epsilon} \, y_\epsilon \, z_\epsilon \Implies\\[.5ex] 
      \hspace*{2ex} (\Neg(\mname{is-var}_{\epsilon \tarrow o} \,
      x_\epsilon) \And \Neg(\mname{is-con}_{\epsilon \tarrow o} \,
      x_\epsilon) \And\\[.5ex] 
      \hspace*{3ex} x_\epsilon \not= \mname{app}_{\epsilon \tarrow
        \epsilon \tarrow \epsilon} \, u_\epsilon \, v_\epsilon \And
      x_\epsilon \not= \mname{quo}_{\epsilon \tarrow \epsilon} \,
      u_\epsilon)$.

    \item $x_\epsilon = \mname{quo}_{\epsilon \tarrow \epsilon} \,
      y_\epsilon \Implies\\[.5ex] 
      \hspace*{2ex} (\Neg(\mname{is-var}_{\epsilon \tarrow o} \,
      x_\epsilon) \And \Neg(\mname{is-con}_{\epsilon \tarrow o} \,
      x_\epsilon) \And\\[.5ex] 
      \hspace*{3ex} x_\epsilon \not= \mname{app}_{\epsilon \tarrow
        \epsilon \tarrow \epsilon} \, u_\epsilon \, v_\epsilon \And
      x_\epsilon \not= \mname{abs}_{\epsilon \tarrow \epsilon
        \tarrow \epsilon} \, u_\epsilon \, v_\epsilon)$.

    \item $\synbrack{\textbf{x}_\alpha} \not=
      \synbrack{\textbf{y}_\beta}$ {\sglsp} where $\textbf{x}_\alpha$
      and $\textbf{y}_\beta$ are distinct.

    \item $\synbrack{\textbf{c}_\alpha} \not=
      \synbrack{\textbf{d}_\beta}$ {\sglsp} where $\textbf{c}_\alpha$
      and $\textbf{d}_\beta$ are distinct.

    \item $\mname{app}_{\epsilon \tarrow \epsilon \tarrow \epsilon} \,
      x_\epsilon \, y_\epsilon = \mname{app}_{\epsilon \tarrow
        \epsilon \tarrow \epsilon} \, u_\epsilon \, v_\epsilon \Implies
      (x_\epsilon = u_\epsilon \And x_\epsilon = v_\epsilon)$.

    \item $\mname{abs}_{\epsilon \tarrow \epsilon \tarrow \epsilon} \,
      x_\epsilon \, y_\epsilon = \mname{abs}_{\epsilon \tarrow
        \epsilon \tarrow \epsilon} \, u_\epsilon \, v_\epsilon \Implies
      (x_\epsilon = u_\epsilon \And y_\epsilon = v_\epsilon)$.

    \item $\mname{quo}_{\epsilon \tarrow \epsilon} \, x_\epsilon =
      \mname{quo}_{\epsilon \tarrow \epsilon} \, u_\epsilon \Implies
      x_\epsilon = u_\epsilon$.

  \ee

  \item[] \textbf{B5 (Definition of $\sqsubset_{\epsilon \tarrow
      \epsilon \tarrow o}$)}

  \be

    \item $x_\epsilon \sqsubset_{\epsilon \tarrow \epsilon \tarrow o}
      \mname{app}_{\epsilon \tarrow \epsilon \tarrow \epsilon} \,
      x_\epsilon \, y_\epsilon $.

    \item $x_\epsilon \sqsubset_{\epsilon \tarrow \epsilon \tarrow o}
      \mname{app}_{\epsilon \tarrow \epsilon \tarrow \epsilon} \,
      y_\epsilon \, x_\epsilon $.

    \item $x_\epsilon \sqsubset_{\epsilon \tarrow \epsilon \tarrow o}
      \mname{abs}_{\epsilon \tarrow \epsilon \tarrow \epsilon} \,
      x_\epsilon \, y_\epsilon $.

    \item $x_\epsilon \sqsubset_{\epsilon \tarrow \epsilon \tarrow o}
      \mname{abs}_{\epsilon \tarrow \epsilon \tarrow \epsilon} \,
      y_\epsilon \, x_\epsilon $.

    \item $x_\epsilon \sqsubset_{\epsilon \tarrow \epsilon \tarrow o}
      \mname{quo}_{\epsilon \tarrow \epsilon} \, x_\epsilon$.

    \item $(x_\epsilon \sqsubset_{\epsilon \tarrow \epsilon \tarrow o}
      y_\epsilon \And y_\epsilon \sqsubset_{\epsilon \tarrow \epsilon
        \tarrow o} z_\epsilon) \Implies x_\epsilon \sqsubset_{\epsilon
        \tarrow \epsilon \tarrow o} z_\epsilon$.

  \ee 

  \item[] \textbf{B6 (Induction Principle for Constructions)}

  \be

    \item[] $\ForallApp u_\epsilon \mdot (\ForallApp v_\epsilon \mdot
      (v_\epsilon \sqsubset_{\epsilon \tarrow \epsilon \tarrow o}
      u_\epsilon \Implies p_{\epsilon \tarrow o} \, v_\epsilon)
      \Implies p_{\epsilon \tarrow o} \, u_\epsilon) \Implies
      \ForallApp u_\epsilon \mdot p_{\epsilon \tarrow o} \,
      u_\epsilon$.

  \ee

  \item[] \textbf{B7 (Definition of $\mname{is-free-in}_{\epsilon
      \tarrow \epsilon \tarrow o}$)}

  \be

    \item $\mname{is-var}_{\epsilon \tarrow o} \, x_\epsilon \Implies
      \mname{is-free-in}_{\epsilon \tarrow \epsilon \tarrow o} \,
      x_\epsilon \, x_\epsilon$.
 
    \item $(\mname{is-var}_{\epsilon \tarrow o} \, x_\epsilon \And
      \mname{is-var}_{\epsilon \tarrow o} \, y_\epsilon \And
      x_\epsilon \not= y_\epsilon) \Implies
      \Neg(\mname{is-free-in}_{\epsilon \tarrow \epsilon \tarrow o} \,
      x_\epsilon \, y_\epsilon)$.

    \item $(\mname{is-var}_{\epsilon \tarrow o} \, x_\epsilon \And
      \mname{is-con}_{\epsilon \tarrow o} \, y_\epsilon \Implies
      \Neg(\mname{is-free-in}_{\epsilon \tarrow \epsilon \tarrow o} \,
      x_\epsilon \, y_\epsilon)$.

    \item $(\mname{is-var}_{\epsilon \tarrow o} \, x_\epsilon \And
      \mname{is-expr}_{\epsilon \tarrow o} \, (\mname{app}_{\epsilon
        \tarrow \epsilon \tarrow \epsilon} \, y_\epsilon \,
      z_\epsilon)) \Implies \\[.5ex] 
      \hspace*{2ex} \mname{is-free-in}_{\epsilon \tarrow \epsilon
        \tarrow o} \, x_\epsilon \, (\mname{app}_{\epsilon \tarrow
        \epsilon \tarrow \epsilon} \, y_\epsilon \, z_\epsilon) \Iff \\[.5ex] 
      \hspace*{2ex} (\mname{is-free-in}_{\epsilon \tarrow \epsilon
        \tarrow o} \, x_\epsilon \, y_\epsilon \Or
      \mname{is-free-in}_{\epsilon \tarrow \epsilon \tarrow o} \,
      x_\epsilon \, z_\epsilon)$.

    \item $(\mname{is-var}_{\epsilon \tarrow o} \, x_\epsilon \And
      \mname{is-expr}_{\epsilon \tarrow o} \, (\mname{abs}_{\epsilon
        \tarrow \epsilon \tarrow \epsilon} \, y_\epsilon \,
      z_\epsilon) \And x_\epsilon \not= y_\epsilon) \Implies \\[.5ex] 
      \hspace*{2ex} \mname{is-free-in}_{\epsilon \tarrow \epsilon
        \tarrow o} \, x_\epsilon \, (\mname{abs}_{\epsilon \tarrow
        \epsilon \tarrow \epsilon} \, y_\epsilon \, z_\epsilon) \Iff
      \mname{is-free-in}_{\epsilon \tarrow \epsilon \tarrow o} \,
      x_\epsilon \, z_\epsilon$.

    \item $\mname{is-expr}_{\epsilon \tarrow o} \,
      (\mname{abs}_{\epsilon \tarrow \epsilon \tarrow \epsilon} \,
      x_\epsilon \, y_\epsilon) \Implies
      \Neg(\mname{is-free-in}_{\epsilon \tarrow \epsilon \tarrow o} \,
      x_\epsilon \, (\mname{abs}_{\epsilon \tarrow \epsilon \tarrow
        \epsilon} \, x_\epsilon \, y_\epsilon))$.

    \item $(\mname{is-var}_{\epsilon \tarrow o} \, x_\epsilon \And
      \mname{is-expr}_{\epsilon \tarrow o} \, y_\epsilon)
      \Implies \Neg(\mname{is-free-in}_{\epsilon \tarrow \epsilon
        \tarrow o} \, x_\epsilon \, (\mname{quo}_{\epsilon \tarrow
        \epsilon} \, y_\epsilon))$.

    \item $(\Neg(\mname{is-var}_{\epsilon \tarrow o} \, x_\epsilon)
      \Or \Neg(\mname{is-expr}_{\epsilon \tarrow o} \, y_\epsilon))
      \Implies \Neg(\mname{is-free-in}_{\epsilon \tarrow \epsilon
        \tarrow o} \, x_\epsilon \, y_\epsilon)$.

  \ee

\ei

\subsubsection{Axioms for Quotation}

The axioms for quotation are the instances of the quotation properties
expressed below and the instances of the beta-reduction property for
applications of the form $(\LambdaApp \textbf{x}_\alpha \mdot
\synbrack{\textbf{B}_\beta})\, \textbf{A}_\alpha)$.

\smallskip

\bi

  \item[] \textbf{B8 (Properties of Quotation)}

  \be

    \item $\synbrack{\textbf{F}_{\alpha \tarrow \beta} \,
      \textbf{A}_\alpha} = \mname{app}_{\epsilon \tarrow \epsilon
      \tarrow \epsilon} \, \synbrack{\textbf{F}_{\alpha \tarrow
        \beta}} \, \synbrack{\textbf{A}_\alpha}$.

    \item $\synbrack{\LambdaApp \textbf{x}_\alpha \mdot
      \textbf{B}_\beta} = \mname{abs}_{\epsilon \tarrow \epsilon
      \tarrow \epsilon} \, \synbrack{\textbf{x}_\alpha} \,
      \synbrack{\textbf{B}_\beta}$.

    \item $\synbrack{\synbrack{\textbf{A}_\alpha}} =
      \mname{quo}_{\epsilon \tarrow \epsilon} \,
      \synbrack{\textbf{A}_\alpha}$.

  \ee

  \item[] \textbf{B9 (Beta-Reduction for Quotations)}

  \be

    \item[] $(\LambdaApp \textbf{x}_\alpha \mdot
      \synbrack{\textbf{B}_\beta}) \, \textbf{A}_\alpha) =
      \synbrack{\textbf{B}_\beta}$.

  \ee

\ei

\subsubsection{Axioms for Evaluation}

The axioms for evaluation are the instances of the evaluation
properties expressed below and the instances of the beta-reduction
properties for a partial set of applications of the form $(\LambdaApp
\textbf{x}_\alpha \mdot \sembrack{\textbf{B}_\epsilon}_\beta) \,
\textbf{A}_\alpha$ also expressed below.

\bsp
By virtue of the axioms for the properties of evaluation and for
beta-reduction for evaluations, the proof system for {\churchqe}
satisfies Requirements R5 and R4, respectively.
\esp

\smallskip

\bi

  \item[] \textbf{B10 (Properties of Evaluation)}

  \be

    \item $\sembrack{\synbrack{\textbf{x}_\alpha}}_\alpha = \textbf{x}_\alpha$.

    \item $\sembrack{\synbrack{\textbf{c}_\alpha}}_\alpha = \textbf{c}_\alpha$.

    \item $(\mname{is-expr}_{\epsilon \tarrow o}^{\alpha \tarrow \beta}
      \, \textbf{A}_\epsilon \And \mname{is-expr}_{\epsilon \tarrow o}^{\alpha}
      \, \textbf{B}_\epsilon) \Implies\\[.5ex]
      \hspace*{2ex} \sembrack{\mname{app}_{\epsilon \tarrow \epsilon
          \tarrow \epsilon} \, \textbf{A}_\epsilon \,
        \textbf{B}_\epsilon}_{\beta} =
      \sembrack{\textbf{A}_\epsilon}_{\alpha \tarrow \beta} \,
      \sembrack{\textbf{B}_\epsilon}_{\alpha}$.

    \item $(\mname{is-expr}_{\epsilon \tarrow o}^{\beta} \,
      \textbf{A}_\epsilon \And \Neg(\mname{is-free-in}_{\epsilon \tarrow
        \epsilon \tarrow o} \, \synbrack{\textbf{x}_\alpha} \,
      \synbrack{\textbf{A}_\epsilon})) \Implies\\[.5ex]
      \hspace*{2ex} \sembrack{\mname{abs}_{\epsilon \tarrow \epsilon
          \tarrow \epsilon} \, \synbrack{\textbf{x}_\alpha} \,
        \textbf{A}_\epsilon}_{\alpha \tarrow \beta} = \LambdaApp
      \textbf{x}_\alpha \mdot \sembrack{\textbf{A}_\epsilon}_\beta$.

    \item $\mname{is-expr}_{\epsilon \tarrow o}^\epsilon \,
      \textbf{A}_\epsilon \Implies \sembrack{\mname{quo}_{\epsilon
          \tarrow \epsilon} \, \textbf{A}_\epsilon}_\epsilon =
      \textbf{A}_\epsilon$.

  \ee

  \item[] \textbf{B11 (Beta-Reduction for Evaluations)}

  \be

    \item $(\LambdaApp \textbf{x}_\alpha \mdot
      \sembrack{\textbf{B}_\epsilon}_\beta) \, \textbf{x}_\alpha =
      \sembrack{\textbf{B}_\epsilon}_\beta$.
 
    \item $(\mname{is-expr}_{\epsilon \tarrow o}^{\beta} \,
      ((\LambdaApp \textbf{x}_\alpha \mdot \textbf{B}_\epsilon) \,
      \textbf{A}_\alpha) \And \\[.5ex] \hspace*{0.5ex}
      \Neg(\mname{is-free-in}_{\epsilon \tarrow \epsilon \tarrow o}
      \, \synbrack{\textbf{x}_\alpha} \, ((\LambdaApp
      \textbf{x}_\alpha \mdot \textbf{B}_\epsilon) \,
      \textbf{A}_\alpha))) \Implies \\[.5ex]
      \hspace*{2ex} (\LambdaApp \textbf{x}_\alpha \mdot
      \sembrack{\textbf{B}_\epsilon}_\beta) \, \textbf{A}_\alpha =
      \sembrack{(\LambdaApp \textbf{x}_\alpha \mdot
        \textbf{B}_\epsilon) \, \textbf{A}_\alpha}_\beta$.

  \ee

\ei

\subsubsection{Axioms involving \mname{IS-EFFECTIVE-IN}}

The first axiom in this part says that a variable is not effective in
an eval-free expression whenever the variable is not free in the
expression.  The second axiom strengthens Axiom A4.5 by replacing the
side conditions of the form ``$\textbf{B}_\beta$ is eval-free and
$\textbf{x}_\alpha$ is not free in $\textbf{B}_\beta$'' with
conditions of the form
$\Neg\mname{IS-EFFECTIVE-IN}(\textbf{x}_\alpha,\textbf{B}_\beta)$
given in the axiom itself.

\bi

  \item[] \textbf{B12 (``Not Free In'' means ``Not Effective In'')}

  \bi

    \item[] $\Neg\mname{IS-EFFECTIVE-IN}(\textbf{x}_\alpha,\textbf{B}_\beta)$\\
      where $\textbf{B}_\beta$ is eval-free and
      $\textbf{x}_\alpha$ is not free in $\textbf{B}_\beta$.

  \ei

  \item[] \textbf{B13 (Axiom A4.5 using ``Not Effective In'')}

  \bi

    \item[] $(\Neg \mname{IS-EFFECTIVE-IN}(\textbf{y}_\beta,\textbf{A}_\alpha)\Or 
      \Neg \mname{IS-EFFECTIVE-IN}(\textbf{x}_\alpha,\textbf{B}_\gamma)) \Implies {}\\ 
      \hspace*{2ex}(\LambdaApp \textbf{x}_\alpha \mdot 
      \LambdaApp \textbf{y}_\beta \mdot \textbf{B}_\gamma) \, \textbf{A}_\alpha =
      \LambdaApp \textbf{y}_\beta \mdot 
      ((\LambdaApp \textbf{x}_\alpha \mdot \textbf{B}_\gamma) \, \textbf{A}_\alpha)$\\
      where $\textbf{x}_\alpha$ and $\textbf{y}_\beta$ are distinct.

   \ei

\ei

\subsection{Proofs}

Let $T=(L,\Gamma)$ be a theory of {\churchqe}, $\textbf{A}_o$ be a
formula in $L$, and $\sH$ be a set of eval-free formulas in $L$.  A
\emph{proof of $\textbf{A}_o$ in $T$} is a finite sequence of formulas
in $L$ ending with $\textbf{A}_o$ such that every formula in the
sequence is one of the axioms of {\churchqe}, a member of $\Gamma$, or
is inferred from previous formulas in the sequence by the Rule R.
$\textbf{A}_o$ is a \emph{theorem of $T$} if there is a proof of
$\textbf{A}_o$ in $T$.  A \emph{proof in {\churchqe}} is a proof in
the theory $T_{\rm logic} = (L_{\cal C},\emptyset)$, and a \emph{theorem
  of {\churchqe}} is a theorem of $T_{\rm logic}$.

\emph{$\textbf{A}_o$ is provable from $\sH$ in $T$}, written
$\proves{T,\sH}{\textbf{A}_o}$, is defined by the following
statements:

\be

  \item If $\textbf{A}_o \in \sH$, then $\proves{T,\sH}{\textbf{A}_o}$.

  \item If $\textbf{A}_o$ is a theorem of $T$, then
    $\proves{T,\sH}{\textbf{A}_o}$.

  \item \textbf{Rule ${\rm \bf R}'$}.  If
    $\proves{T,\sH}{\textbf{A}_\alpha = \textbf{B}_\alpha}$;
    $\proves{T,\sH}{\textbf{C}_o}$; and $\textbf{D}_o$ is obtained
    from $\textbf{C}_o$ by replacing one occurrence of
    $\textbf{A}_\alpha$ in $\textbf{C}_o$ by an occurrence of
    $\textbf{B}_\alpha$, then $\proves{T,\sH}{\textbf{D}_o}$ provided
    that the occurrence of $\textbf{A}_\alpha$ in $\textbf{C}_o$ is
    not within a quotation, not the first argument of a function
    abstraction, not the second argument of an evaluation, and in
    a function application $\LambdaApp \textbf{x}_\beta \mdot
    \textbf{E}_\gamma$ only if
    \[\proves{T,\sH}{\Neg\mname{IS-EFFECTIVE-IN}(\textbf{x}_\beta,\textbf{A}_\alpha
      = \textbf{B}_\alpha)}\] or 
    \[\proves{T,\sH}{\Neg\mname{IS-EFFECTIVE-IN}(\textbf{x}_\beta,\textbf{H}_o)}\]
    for all $\textbf{H}_o \in \sH$.

\ee
\noindent
Notice that the variables in the members of $\Gamma$ are treated as if
they were universally quantified, while the variables in the members
of $\sH$ are treated as if they were constants.
$\proves{T,\emptyset}{\textbf{A}_o}$ and $\proves{T_{\rm
    logic},\emptyset}{\textbf{A}_o}$ are abbreviated as
$\proves{T}{\textbf{A}_o}$ and $\proves{}{\textbf{A}_o}$,
respectively.  Note that $\proves{T}{\textbf{A}_o}$ iff $\textbf{A}_o$
is a theorem of $T$ since Rule ${\rm R}'$ reduces to Rule R when $\sH
= \emptyset$.  $T$ is \emph{consistent} if not $\proves{T}{F_o}$

\section{Proof-Theoretic Results}\label{sec:pt-results}

In this section, let $T$ be a theory of {\churchqe} and $\sH$ be a set
of formulas of $T$.

\subsection{Equality}

\begin{lem}[Reflexivity of Equality]\label{lem:eq-reflex}
$\proves{}{\textbf{A}_\alpha = \textbf{A}_\alpha}$.
\end{lem}

\begin{proof}
  \begin{align} \setcounter{equation}{0}
  & \vdash (\LambdaApp \textbf{x}_\alpha \mdot \textbf{x}_\alpha) \,
      \textbf{A}_\alpha = \textbf{A}_\alpha\\
  & \vdash \textbf{A}_\alpha = \textbf{A}_\alpha
  \end{align}
(1) is an instance of Axiom A4.2 and (2) follows from (1) and (1) by
  Rule R.
\end{proof}

\begin{lem}[Equality Rules]\label{lem:eq-rules}
\be

  \item[]

  \item If $\proves{T,\sH}{\textbf{A}_\alpha = \textbf{B}_\alpha}$, then
    $\proves{T,\sH}{\textbf{B}_\alpha = \textbf{A}_\alpha}$.

  \item If $\proves{T,\sH}{\textbf{A}_\alpha = \textbf{B}_\alpha}$ and
    $\proves{T,\sH}{\textbf{B}_\alpha = \textbf{C}_\alpha}$, then
    $\proves{T,\sH}{\textbf{A}_\alpha = \textbf{C}_\alpha}$.

  \item If $\proves{T,\sH}{\textbf{A}_{\alpha \tarrow \beta} =
    \textbf{B}_{\alpha \tarrow \beta}}$ and
    $\proves{T,\sH}{\textbf{C}_\alpha = \textbf{D}_\alpha}$, then\\
    \mbox{$\proves{T,\sH}{\textbf{A}_{\alpha \tarrow \beta} \,
    \textbf{C}_\alpha = \textbf{B}_{\alpha \tarrow \beta} \,
    \textbf{D}_\alpha}$.}

  \item If $\proves{T,\sH}{\textbf{A}_\epsilon =
    \textbf{B}_\epsilon}$, then $\proves{T,\sH}{\sembrack{\textbf{A}_\epsilon}_\alpha =
    \sembrack{\textbf{B}_\epsilon}_\alpha}$.

  \item If $\proves{T,\sH}{\textbf{A}_o}$ and
    $\proves{T,\sH}{\textbf{A}_o = \textbf{B}_o}$, then
    $\proves{T,\sH}{\textbf{B}_o}$.

\ee
\end{lem}

\begin{proof}
By Lemma~\ref{lem:eq-reflex} and Rule ${\rm R}'$.
\end{proof}

\subsection{Substitution}

It will be convenient to define in the metalogic of {\churchqe} a
substitution operator named $\mname{SUB}$.  Roughly speaking,
$\mname{SUB}(\textbf{A}_\alpha,\textbf{x}_\alpha,\textbf{B}_\beta)$
denotes the expression obtained by replacing each free occurrence of
$\textbf{x}_\alpha$ in $\textbf{B}_\beta$ with $\textbf{A}_\alpha$
except that the substitution is curtained (1) within a function
abstraction when a variable capture will occur and (2) within an
evaluation when $\textbf{A}_\alpha$ is not $\textbf{x}_\alpha$.

The \emph{substitution of an expression $\textbf{A}_\alpha$ for a
  variable $\textbf{x}_\alpha$ in an expression $\textbf{B}_\beta$},
written
$\mname{SUB}(\textbf{A}_\alpha,\textbf{x}_\alpha,\textbf{B}_\beta)$,
is defined recursively as follows:

\be

  \item
    $\mname{SUB}(\textbf{A}_\alpha,\textbf{x}_\alpha,\textbf{y}_\beta)
    = \textbf{y}_\beta$ {\sglsp} where $\textbf{x}_\alpha$ and
    $\textbf{y}_\beta$ are distinct.

  \item
    $\mname{SUB}(\textbf{A}_\alpha,\textbf{x}_\alpha,\textbf{x}_\alpha)
    = \textbf{A}_\alpha$.

  \item
    $\mname{SUB}(\textbf{A}_\alpha,\textbf{x}_\alpha,\textbf{c}_\beta)
    = \textbf{c}_\beta$.

  \item
    $\mname{SUB}(\textbf{A}_\alpha,\textbf{x}_\alpha,\textbf{B}_{\beta
    \tarrow \gamma} \, \textbf{C}_\beta) =
    \mname{SUB}(\textbf{A}_\alpha,\textbf{x}_\alpha,\textbf{B}_{\beta
      \tarrow \gamma}) \,
    (\mname{SUB}(\textbf{A}_\alpha,\textbf{x}_\alpha,\textbf{C}_\beta))$.

  \item
    $\mname{SUB}(\textbf{A}_\alpha,\textbf{x}_\alpha,\LambdaApp
    \textbf{y}_\beta \mdot \textbf{B}_\gamma) = \LambdaApp
    \textbf{y}_\beta \mdot
    \mname{SUB}(\textbf{A}_\alpha,\textbf{x}_\alpha,\textbf{B}_\gamma)$
          {\sglsp} where $\textbf{x}_\alpha$ and $\textbf{y}_\beta$
          are distinct and either (1) $\textbf{A}_\alpha$ is eval-free
          and $\textbf{y}_\beta$ is not free in $\textbf{A}_\alpha$ or
          (2) $\textbf{B}_\gamma$ is eval-free and $\textbf{x}_\alpha$
          is not free in $\textbf{B}_\gamma$.

  \item
    $\mname{SUB}(\textbf{A}_\alpha,\textbf{x}_\alpha,\LambdaApp
    \textbf{y}_\beta \mdot \textbf{B}_\gamma) = (\LambdaApp
    \textbf{x}_\alpha \mdot \LambdaApp \textbf{y}_\beta \mdot
    \textbf{B}_\gamma) \, \textbf{A}_\alpha$ {\sglsp} where
    $\textbf{x}_\alpha$ and $\textbf{y}_\beta$ are distinct and not
    (1) $\textbf{A}_\alpha$ is eval-free and $\textbf{y}_\beta$ is not
    free in $\textbf{A}_\alpha$ and not (2) $\textbf{B}_\gamma$ is
    eval-free and $\textbf{x}_\alpha$ is not free in
    $\textbf{B}_\gamma$.

  \item
    $\mname{SUB}(\textbf{A}_\alpha,\textbf{x}_\alpha,\LambdaApp
    \textbf{x}_\alpha \mdot \textbf{B}_\beta) = \LambdaApp
    \textbf{x}_\alpha \mdot \textbf{B}_\beta$.

  \item 
    $\mname{SUB}(\textbf{A}_\alpha,\textbf{x}_\alpha,\synbrack{\textbf{B}_\beta})
    = \synbrack{\textbf{B}_\beta}$.

  \item
    $\mname{SUB}(\textbf{x}_\alpha,\textbf{x}_\alpha,\sembrack{\textbf{B}_\epsilon}_\beta)
    = \sembrack{\textbf{B}_\epsilon}_\beta$.

  \item 
    $\mname{SUB}(\textbf{A}_\alpha,\textbf{x}_\alpha,\sembrack{\textbf{B}_\epsilon}_\beta)
    = (\LambdaApp \textbf{x}_\alpha \mdot
    \sembrack{\textbf{B}_\epsilon}_\beta) \, \textbf{A}_\alpha$
             {\sglsp} where $\textbf{A}_\alpha$ is not
             $\textbf{x}_\alpha$.

\ee
Notice that in clauses 6 and 10, the substitution is not performed and
the intent of the substitution is recorded as an appropriate lambda
application.

The following theorem shows that beta-reduction can be performed (in
part) by substitution via \mname{SUB}.

\begin{thm}[Beta-Reduction by Substitution]
\[\proves{}{(\LambdaApp \textbf{x}_\alpha \mdot \textbf{B}_\beta) \,
    \textbf{A}_\alpha = \mname{SUB}(\textbf{A}_\alpha,
    \textbf{x}_\alpha, \textbf{B}_\beta)}.\]
\end{thm}

\begin{proof}
The proof is by induction on the structure of $\textbf{B}_\beta$.
There are six cases corresponding to the six formation rules for
expressions.

\bi

  \item[] \textbf{Case 1}: $\textbf{B}_\beta$ is a variable.

  \bi

    \item[] \textbf{Subcase 1a}: $\textbf{B}_\beta$ is
      $\textbf{y}_\beta$ and $\textbf{x}_\alpha$ and
      $\textbf{y}_\beta$ are distinct.  The theorem follows
      immediately from Axiom A4.1 and the definition of
      $\mname{SUB}$.

    \item[] \textbf{Subcase 1b}: $\textbf{B}_\beta$ is
      $\textbf{x}_\alpha$. The theorem follows immediately from Axiom
      A4.2 and the definition of $\mname{SUB}$.

  \ei

  \item[] \textbf{Case 2}: $\textbf{B}_\beta$ is a constant. The
    theorem follows immediately from Axiom A4.3 and the definition of
    $\mname{SUB}$.

  \item[] \textbf{Case 3}: $\textbf{B}_\beta$ is a function
    application $\textbf{B}_{\beta \tarrow \gamma} \,
    \textbf{C}_\beta$.
  \begin{align} \setcounter{equation}{0}
  \vdash {} &(\LambdaApp \textbf{x}_\alpha \mdot 
  (\textbf{B}_{\beta \tarrow \gamma} \, 
  \textbf{C}_\beta)) \, \textbf{A}_\alpha =\nonumber\\
  &((\LambdaApp \textbf{x}_\alpha \mdot \textbf{B}_{\beta \tarrow \gamma}) \, 
  \textbf{A}_\alpha) \, 
  ((\LambdaApp \textbf{x}_\alpha \mdot \textbf{C}_\beta) \, \textbf{A}_\alpha)\\
  \vdash {} &(\LambdaApp \textbf{x}_\alpha \mdot \textbf{B}_{\beta \tarrow \gamma}) \,
  \textbf{A}_\alpha = 
  \mname{SUB}(\textbf{A}_\alpha, \textbf{x}_\alpha, \textbf{B}_{\beta \tarrow \gamma})\\
  \vdash {} &(\LambdaApp \textbf{x}_\alpha \mdot \textbf{C}_\beta) \,
  \textbf{A}_\alpha = 
  \mname{SUB}(\textbf{A}_\alpha, \textbf{x}_\alpha, \textbf{C}_\beta)\\
  \vdash {} &(\LambdaApp \textbf{x}_\alpha \mdot 
  (\textbf{B}_{\beta \tarrow \gamma} \, 
  \textbf{C}_\beta)) \, \textbf{A}_\alpha =\nonumber\\
  &\mname{SUB}(\textbf{A}_\alpha, \textbf{x}_\alpha, \textbf{B}_{\beta \tarrow \gamma})
  (\mname{SUB}(\textbf{A}_\alpha, \textbf{x}_\alpha, \textbf{C}_\beta))\\
  \vdash {} &(\LambdaApp \textbf{x}_\alpha \mdot 
  (\textbf{B}_{\beta \tarrow \gamma} \, 
  \textbf{C}_\beta)) \, \textbf{A}_\alpha =
  \mname{SUB}(\textbf{A}_\alpha,\textbf{x}_\alpha,
  \textbf{B}_{\beta \tarrow \gamma} \, \textbf{C}_\beta)
  \end{align}
  (1) is by Axiom A4.4; (2) and (3) are by the induction hypothesis;
  (4) follows from (1), (2), and (3) by the Equality Rules; and (5),
  the theorem, follows from (4) by the definition of $\mname{SUB}$.

  \item[] \textbf{Case 4}: $\textbf{B}_\beta$ is a function abstraction.  

  \bi

    \item[] \textbf{Subcase 4a}: $\textbf{B}_\beta$ is $\LambdaApp
      \textbf{y}_\beta \mdot \textbf{B}_\gamma$, $\textbf{x}_\alpha$
      and $\textbf{y}_\beta$ are distinct, and either
      (1)~$\textbf{A}_\alpha$ is eval-free and $\textbf{y}_\beta$ is
      not free in $\textbf{A}_\alpha$ or (2) $\textbf{B}_\gamma$ is
      eval-free and $\textbf{x}_\alpha$ is not free in
      $\textbf{B}_\gamma$.
    \begin{align} \setcounter{equation}{0}
    & 
    \vdash (\LambdaApp \textbf{x}_\alpha \mdot \LambdaApp \textbf{y}_\beta \mdot 
    \textbf{B}_\gamma) \, \textbf{A}_\alpha =
    \LambdaApp \textbf{y}_\beta \mdot 
    ((\LambdaApp \textbf{x}_\alpha \mdot \textbf{B}_\gamma) \, \textbf{A}_\alpha)\\
    &
    \vdash (\LambdaApp \textbf{x}_\alpha \mdot \textbf{B}_\gamma) \, 
    \textbf{A}_\alpha =
    \mname{SUB}(\textbf{A}_\alpha,\textbf{x}_\alpha,\textbf{B}_\gamma)\\
    & 
    \vdash (\LambdaApp \textbf{x}_\alpha \mdot \LambdaApp \textbf{y}_\beta \mdot 
    \textbf{B}_\gamma) \, \textbf{A}_\alpha =
    \LambdaApp \textbf{y}_\beta \mdot
    \mname{SUB}(\textbf{A}_\alpha,\textbf{x}_\alpha,\textbf{B}_\gamma)\\
    &
    \vdash (\LambdaApp \textbf{x}_\alpha \mdot \LambdaApp \textbf{y}_\beta \mdot 
    \textbf{B}_\gamma) \, \textbf{A}_\alpha =
    \mname{SUB}(\textbf{A}_\alpha,\textbf{x}_\alpha,
    \LambdaApp \textbf{y}_\beta \mdot \textbf{B}_\gamma)
    \end{align}
    (1) is by Axiom A4.5; (2) is by the induction hypothesis; (3)
    follows from (1) and (2) by the Equality Rules; and (4), the
    theorem, follows from (3) by the definition of $\mname{SUB}$.      

    \item[] \textbf{Subcase 4b}: $\textbf{B}_\beta$ is $\LambdaApp
      \textbf{y}_\beta \mdot \textbf{B}_\gamma$, $\textbf{x}_\alpha$
      and $\textbf{y}_\beta$ are distinct, not (1) $\textbf{A}_\alpha$
      is eval-free and $\textbf{y}_\beta$ is not free in
      $\textbf{A}_\alpha$, and not (2) $\textbf{B}_\gamma$ is
      eval-free and $\textbf{x}_\alpha$ is not free in
      $\textbf{B}_\gamma$.  The theorem follows immediately from the
      definition of $\mname{SUB}$.

    \item[] \textbf{Subcase 4c}: $\textbf{B}_\beta$ is $\LambdaApp
      \textbf{x}_\alpha \mdot \textbf{B}_\beta$.  The theorem follows
      immediately from Axiom~A4.6 and the definition of $\mname{SUB}$.

  \ei

  \item[] \textbf{Case 5}: $\textbf{B}_\beta$ is a quotation. The
    theorem follows immediately from Axiom~B9 and the definition of
    $\mname{SUB}$.

  \item[] \textbf{Case 6}: $\textbf{B}_\beta$ is an evaluation.  

  \bi

    \item[] \textbf{Subcase 6a}: $\textbf{A}_\alpha$ is
      $\textbf{x}_\alpha$.  The theorem follows immediately from Axiom
      B11.1 and the definition of $\mname{SUB}$.

    \item[] \textbf{Subcase 6b}: $\textbf{A}_\alpha$ is not
      $\textbf{x}_\alpha$.  The theorem follows immediately from the
      definition of $\mname{SUB}$.

  \ei

\ei
\end{proof}

\begin{thm}[Universal Instantiation]
If $\proves{T,\sH}{\ForallApp \textbf{x}_\alpha \mdot \textbf{B}_o}$,
then\\
$\proves{T,\sH}{\mname{SUB}(\textbf{A}_\alpha,\textbf{x}_\alpha,\textbf{B}_o)}$.
\end{thm}

\begin{proof}
Similar to the proof of 5215 in~\cite[p.~221]{Andrews02}.
\end{proof}

\begin{lem}\label{lem:sub-eqn}
If $\proves{}{\textbf{B}_\beta = \textbf{C}_\beta}$, then
$\proves{}{\mname{SUB}(\textbf{A}_\alpha, \textbf{x}_\alpha,
  \textbf{B}_\beta = \textbf{C}_\beta)}$.
\end{lem}

\begin{proof}
Similar to the proof of 5204 in~\cite[p.~217]{Andrews02}.
\end{proof}

\begin{lem}\label{lem:eta}
$\proves{}{f_{\alpha \tarrow \beta} = \LambdaApp \textbf{y}_\alpha
    \mdot f_{\alpha \tarrow \beta} \, \textbf{y}_\alpha}$.
\end{lem}

\begin{proof}
Similar to the proof of 5205 in~\cite[p.~217]{Andrews02}.
\end{proof}

\begin{lem}\label{lem:alpha-eq}
If $\textbf{B}_\beta$ is eval-free and $\textbf{y}_\alpha$ is not free
in $\textbf{B}_\beta$, then\\ $\proves{}{\LambdaApp \textbf{x}_\alpha
  \mdot \textbf{B}_\beta = \LambdaApp \textbf{y}_\alpha \mdot
  \mname{SUB}(\textbf{y}_\alpha, \textbf{x}_\alpha,
  \textbf{B}_\beta)}$.
\end{lem}

\begin{proof}
\begin{align} \setcounter{equation}{0}
& 
\vdash f_{\alpha \tarrow \beta} = \LambdaApp \textbf{y}_\alpha \mdot 
f_{\alpha \tarrow \beta} \, \textbf{y}_\alpha\\
&
\vdash \mname{SUB}(\LambdaApp \textbf{x}_\alpha \mdot \textbf{B}_\beta, {\sglsp}
f_{\alpha \tarrow \beta}, {\sglsp}
f_{\alpha \tarrow \beta} = \LambdaApp \textbf{y}_\alpha \mdot 
f_{\alpha \tarrow \beta} \, \textbf{y}_\alpha)\\
&
\vdash \LambdaApp \textbf{x}_\alpha \mdot \textbf{B}_\beta = 
\LambdaApp \textbf{y}_\alpha \mdot 
((\LambdaApp \textbf{x}_\alpha \mdot \textbf{B}_\beta) \,
\textbf{y}_\alpha)\\
&
\vdash \LambdaApp \textbf{x}_\alpha \mdot \textbf{B}_\beta = 
\LambdaApp \textbf{y}_\alpha \mdot \mname{SUB}(\textbf{y}_\alpha,
\textbf{x}_\alpha, \textbf{B}_\beta)
\end{align}
(1) is Lemma~\ref{lem:eta}; (2) follows from (1) by
Lemma~\ref{lem:sub-eqn}; (3) follows from (2) and the hypothesis by
the definition of \mname{SUB}; and (4) follows from (3) by
Beta-Reduction by Substitution and the Equality Rules.
\end{proof}

\begin{cor}[Alpha-Equivalence] 
If $\textbf{B}_\beta$ is eval-free, $\textbf{y}_\alpha$ is not be free
in $\textbf{B}_\beta$, and $\textbf{y}_\alpha$ is free for
$\textbf{x}_\alpha$ in $\textbf{B}_\beta$, then $\LambdaApp
\textbf{x}_\alpha \mdot \textbf{B}_\beta$ and $\LambdaApp
\textbf{y}_\alpha \mdot \mname{SUB}(\textbf{y}_\alpha,
\textbf{x}_\alpha, \textbf{B}_\beta)$ are alpha-equivalent.
\end{cor}

\begin{proof}
By Lemma~\ref{lem:alpha-eq} and the definition of \mname{SUB}.
\end{proof}

\begin{rem}[Nominal Data Types]\em
Since alpha-conversion is not universally valid in {\churchqe}, it is
not clear whether techniques for managing variable naming and binding
--- such as \emph{higher-order abstract
  syntax}~\cite{Miller00,PfenningElliot88} and \emph{nominal
  techniques}~\cite{GabbayPitts02,Pitts03} --- are applicable to
{\churchqe}.  However, the paper~\cite{NanevskiPfenning05} does
combine quotation/evaluation techniques with nominal techniques.
\end{rem}

\begin{lem}\label{lem:univ-inst}
$\proves{}{\ForallApp \textbf{x}_\alpha \mdot \textbf{B}_o \Implies
    \mname{SUB}(\textbf{A}_\alpha, \textbf{x}_\alpha, \textbf{B}_o)}$.
\end{lem}

\begin{proof}
Similar to the proof of 5226 in~\cite[p.~224]{Andrews02}.
\end{proof}

\begin{thm}[Universal Generalization]
If $\proves{T,\sH}{\textbf{A}_o}$, then\\ $\proves{T,\sH}{\ForallApp
  \textbf{x}_\alpha \mdot \textbf{A}_o}$ provided $\proves{T,\sH}{\Neg
  \mname{IS-EFFECTIVE-IN}(\textbf{x}_\alpha,\textbf{H}_o)}$ for all
$\textbf{H}_o \in \sH$.
\end{thm}

\begin{proof}
Similar to the proof of 5220 in~\cite[p.~222]{Andrews02}.  The notion
of ``is effective in'' is used in place of the notion of ``is free
in''.
\end{proof}

\begin{thm}[Existential Generalization]
If $\proves{T,\sH}{\mname{SUB}(\textbf{A}_\alpha, \textbf{x}_\alpha,
  \textbf{B}_o)}$,\\ then $\proves{T,\sH}{\ForsomeApp \textbf{x}_\alpha
  \mdot \textbf{B}_o}$.
\end{thm}

\begin{proof}
Similar to the proof of 5242 in~\cite[p.~229]{Andrews02}.  
\end{proof}

\subsection{Propositional Reasoning}

\begin{thm}[Modus Ponens]
If $\proves{T,\sH}{\textbf{A}_o}$ and $\proves{T,\sH}{\textbf{A}_o
  \Implies \textbf{B}_o}$, then $\proves{T,\sH}{\textbf{B}_o}$.
\end{thm}

\begin{proof}
Similar to the proof of 5224 in~\cite[p.~226]{Andrews02}.
\end{proof}

\begin{thm}[Tautology Theorem]
If $\textbf{A}_o$ is a substitution instance of a a tautology, then
$\proves{}{\textbf{A}_o}$.
\end{thm}

\begin{proof}
Similar to the proof of 5234 in~\cite[p.~227]{Andrews02}.
\end{proof}

\begin{lem}\label{lem:forall-or}
If $\proves{T,\sH}{\Neg \mname{IS-EFFECTIVE-IN}(\textbf{x}_\alpha,
  \textbf{A}_o)}$, then \[\proves{T,\sH}{\ForallApp \textbf{x}_\alpha
  \mdot (\textbf{A}_o \Or \textbf{B}_o) \Implies (\textbf{A}_o \Or
  \ForallApp \textbf{x}_\alpha \mdot \textbf{B}_o)}.\]
\end{lem}

\begin{proof}
Similar to the proof of 5235 in \cite[p.~227]{Andrews02} except Axiom
B13 is needed in the last step of the proof.
\end{proof}

\begin{lem}\label{lem:not-effective-syn}
If $\proves{T,\sH}{\Neg \mname{IS-EFFECTIVE-IN}(\textbf{x}_\alpha,
  \textbf{A}_o)}$, then \[\proves{T,\sH}{(\ForallApp \textbf{x}_\alpha
  \mdot \textbf{A}_o) \Iff \textbf{A}_o}.\]
\end{lem}

\begin{proof}
\begin{align} \setcounter{equation}{0}
& \vdash
  \Neg \textbf{A}_o \Or \textbf{A}_o\\
& \vdash
  \ForallApp \textbf{x}_\alpha \mdot (\Neg \textbf{A}_o \Or \textbf{A}_o)\\
T,\sH & \vdash
  \Neg \textbf{A}_o \Or \ForallApp \textbf{x}_\alpha \mdot \textbf{A}_o\\
T,\sH & \vdash
  \textbf{A}_o \Implies \ForallApp \textbf{x}_\alpha \mdot \textbf{A}_o\\
T,\sH & \vdash
  (\ForallApp \textbf{x}_\alpha \mdot \textbf{A}_o) \Implies 
  (\LambdaApp \textbf{x}_\alpha \mdot \textbf{A}_o) \, \textbf{x}_\alpha\\
T,\sH & \vdash
  (\ForallApp \textbf{x}_\alpha \mdot \textbf{A}_o) \Implies \textbf{A}_o\\
T,\sH & \vdash
  (\ForallApp \textbf{x}_\alpha \mdot \textbf{A}_o) \Iff \textbf{A}_o
\end{align}
(1) is an instance of the Tautology Theorem; (2) follows from (1) by
Universal Generalization; (3) follows from (2), the hypothesis, and
Lemma~\ref{lem:forall-or} by Modus Ponens; (4) follows from (3) by the
Tautology Theorem; (5) follows from Lemma~\ref{lem:univ-inst} and
Beta-Reduction by Substitution by the Equality Rules; (6)~follows from
Beta-Reduction by Substitution and the Equality Rules if
$\textbf{A}_o$ is eval-free and by Axiom B11.1 and the Equality Rules
if $\textbf{A}_o$ is not eval-free; and (7) follows from (4) and (6)
by the Tautology Theorem.
\end{proof}

\begin{rem}\em
Lemma~\ref{lem:not-effective-syn} shows that a variable that is not
effective in an expression has the same behavior with respect to
universal quantification as a variable that is not free in an
eval-free expression.
\end{rem}

\begin{thm}[Deduction Theorem]
If $\proves{T,\sH \cup \set{\textbf{H}_o}}{\textbf{A}_o}$,
then\\ $\proves{T,\sH}{\textbf{H}_o \Implies \textbf{A}_o}$.
\end{thm}

\begin{proof}
Similar to the proof of 5240 in~\cite[p.~228]{Andrews02}.  The notion
of ``is effective in'' is used in place of the notion of ``is free
in''.
\end{proof}

\begin{lem} \label{lem:exist-rule}
If $\proves{T,\sH \cup \set{\textbf{A}_o}}{\textbf{B}_o}$, then
$\proves{T,\sH\cup \set{\ForsomeApp \textbf{x}_\alpha \mdot
    \textbf{A}_o}}{\textbf{B}_o}$ provided $\proves{T,\sH}{\Neg
  \mname{IS-EFFECTIVE-IN}(\textbf{x}_\alpha,\textbf{B}_o)}$ and
$\proves{T,\sH}{\Neg
  \mname{IS-EFFECTIVE-IN}(\textbf{x}_\alpha,\textbf{H}_o)}$ for all
$\textbf{H}_o \in \sH$.
\end{lem}

\begin{proof}
\begin{align} \setcounter{equation}{0}
& T,\sH \cup \set{\textbf{A}_o} \vdash
  \textbf{B}_o\\
& T,\sH \cup \set{\Neg\textbf{B}_o} \vdash
  \Neg\textbf{A}_o\\
& T,\sH \cup \set{\Neg\textbf{B}_o} \vdash
  \ForallApp \textbf{x}_\alpha \mdot \Neg\textbf{A}_o\\
& T,\sH \cup \set{\ForsomeApp \textbf{x}_\alpha \mdot \textbf{A}_o} \vdash
  \textbf{B}_o
\end{align}
(1) is the hypothesis; (2) follows from (1) by the Deduction Theorem
and propositional logic; (3) follows from (2) by Universal
Generalization and the condition placed on $\textbf{x}_\alpha$; and
(4) follows from (3) by the Deduction Theorem, the definition of
$\Forsome$, and propositional logic.
\end{proof}

\begin{thm}[Weakening]
Let $\sH'$ be a set of formulas of $T$ such that\\ $\sH \subseteq \sH'$.
If $\proves{T,\sH}{\textbf{A}_o}$, then
$\proves{T,\sH'}{\textbf{A}_o}$.
\end{thm}

\begin{proof}
Follows immediately from the definition of
$\proves{T,\sH}{\textbf{A}_o}$.
\end{proof}

\subsection{Quotations}

\begin{thm}[Syntactic Law of Quotation]\label{thm:syn-quotation}
$\proves{}{\synbrack{\textbf{A}_\alpha} = \sE(\textbf{A}_\alpha)}$.
\end{thm}

\begin{proof}
Let $\textbf{A}_\alpha$ be eval-free.  Our proof is by induction on
the structure of $\textbf{A}_\alpha$.  There are five cases
corresponding to the five formation rules for eval-free expressions.

\bi

  \item[] \textbf{Case 1}: $\textbf{A}_\alpha$ is a variable
    $\textbf{x}_\alpha$.  $\proves{}{\synbrack{\textbf{x}_\alpha} =
    \synbrack{\textbf{x}_\alpha}}$ by Lemma~\ref{lem:eq-reflex}.
    So $\proves{}{\synbrack{\textbf{x}_\alpha} =
      \sE(\textbf{x}_\alpha)}$ since $\sE(\textbf{x}_\alpha) =
    \synbrack{\textbf{x}_\alpha}$.

  \item[] \textbf{Case 2}: $\textbf{A}_\alpha$ is a constant
    $\textbf{c}_\alpha$.  Similar to Case 1.

  \item[] \textbf{Case 3}: $\textbf{A}_\alpha$ is an function
    application $\textbf{F}_{\beta \tarrow \gamma} \, \textbf{B}_\beta$.
    \begin{align} \setcounter{equation}{0}
    & \vdash \synbrack{\textbf{F}_{\beta \tarrow \gamma} \, \textbf{B}_\beta} = 
      \mname{app}_{\epsilon \tarrow \epsilon \tarrow \epsilon} \, 
      \synbrack{\textbf{F}_{\beta \tarrow \gamma}} \, \synbrack{\textbf{B}_\beta}\\ 
    & \vdash \synbrack{\textbf{F}_{\beta \tarrow \gamma}} =
      \sE(\textbf{F}_{\beta \tarrow \gamma})\\
    & \vdash \synbrack{\textbf{B}_\beta} = \sE(\textbf{B}_\beta)\\
    & \vdash \synbrack{\textbf{F}_{\beta \tarrow \gamma} \, \textbf{B}_\beta} = 
      \mname{app}_{\epsilon \tarrow \epsilon \tarrow \epsilon} \, 
      \sE(\textbf{F}_{\beta \tarrow \gamma}) \, \sE(\textbf{B}_\beta)\\
    & \vdash \synbrack{\textbf{F}_{\beta \tarrow \gamma} \, \textbf{B}_\beta} = 
      \sE(\textbf{F}_{\beta \tarrow \gamma} \, \textbf{B}_\beta)
    \end{align}
    (1) is an instance of Axiom B8.1; (2) and (3) are by the induction
    hypothesis; (4) follows from (1), (2), and (3) by Rule R used
    twice; and (5) follows from (4) since $\sE(\textbf{F}_{\beta
      \tarrow \gamma} \, \textbf{B}_\beta) = \mname{app}_{\epsilon
      \tarrow \epsilon \tarrow \epsilon} \, \sE(\textbf{F}_{\beta
      \tarrow \gamma}) \, \sE(\textbf{B}_\beta)$.

  \item[] \textbf{Case 4}: $\textbf{A}_\alpha$ is an function
    abstraction $(\LambdaApp \textbf{x}_\beta \mdot
    \textbf{B}_\gamma)$.
    \begin{align} \setcounter{equation}{0}
    & \vdash \synbrack{\LambdaApp \textbf{x}_\beta \mdot \textbf{B}_\gamma} = 
      \mname{abs}_{\epsilon \tarrow \epsilon \tarrow \epsilon} \, 
      \synbrack{\textbf{x}_\beta} \, \synbrack{\textbf{B}_\gamma}\\
    & \vdash \synbrack{\textbf{B}_\gamma} = \sE(\textbf{B}_\gamma)\\
    & \vdash \synbrack{\LambdaApp \textbf{x}_\beta \mdot \textbf{B}_\gamma} =
      \mname{abs}_{\epsilon \tarrow \epsilon \tarrow \epsilon} \, 
      \synbrack{\textbf{x}_\beta} \, \sE(\textbf{B}_\gamma)\\
    & \vdash \synbrack{\LambdaApp \textbf{x}_\beta \mdot \textbf{B}_\gamma} =
      \sE(\LambdaApp \textbf{x}_\beta \mdot \textbf{B}_\gamma)
    \end{align}
    (1) is an instance of Axiom B8.2; (2) is by the induction
    hypothesis; (3)~follows from (2) and (1) by Rule R; and (4)
    follows from (3) since $\sE(\LambdaApp \textbf{x}_\beta \mdot
    \textbf{B}_\gamma) = \mname{abs}_{\epsilon \tarrow \epsilon
      \tarrow \epsilon} \, \sE(\textbf{x}_\beta) \,
    \sE(\textbf{B}_\gamma) = \mname{abs}_{\epsilon \tarrow \epsilon
      \tarrow \epsilon} \, \synbrack{\textbf{x}_\beta} \,
    \sE(\textbf{B}_\gamma)$.

  \item[] \textbf{Case 5}: $\textbf{A}_\alpha$ is a quotation
    $\synbrack{\textbf{B}_\beta}$. 
    \begin{align} \setcounter{equation}{0}
    & \vdash \synbrack{\synbrack{\textbf{B}_\beta}} = 
      \mname{quo}_{\epsilon \tarrow \epsilon} \, \synbrack{\textbf{B}_\beta}\\
    & \vdash \synbrack{\textbf{B}_\beta} = \sE(\textbf{B}_\beta)\\
    & \vdash \synbrack{\synbrack{\textbf{B}_\beta}} = 
      \mname{quo}_{\epsilon \tarrow \epsilon} \, \sE(\textbf{B}_\beta)\\
    & \vdash \synbrack{\synbrack{\textbf{B}_\beta}} = 
      \sE(\synbrack{\textbf{B}_\beta})
    \end{align}
    (1) is an instance of Axiom B8.3; (2) is by the induction
    hypothesis; (3) follows from (2) and (1) by Rule R; and (4)
    follows from (3) since $\sE(\synbrack{\textbf{B}_\beta}) =
    \mname{quo}_{\epsilon \tarrow \epsilon} \, \sE(\textbf{B}_\beta)$.

\ei
\end{proof}

\subsection{Constructions}

\begin{lem}\label{lem:syn-is-expr}
$\proves{}{\mname{is-expr}_{\epsilon \tarrow o}^{\alpha} \,
    \synbrack{\textbf{A}_\alpha}}$.
\end{lem}

\begin{proof}
By induction on the structure of $\textbf{A}_\alpha$ using Axioms
B1.1--4, B2.1--4, and B3.1--9, Universal Generalization, Universal
Instantiation, the Syntactic Law of Quotation, and propositional
logic.
\end{proof} 

\begin{lem}\label{lem:syn-sqsubset}
$\proves{}{\synbrack{\textbf{A}_\alpha} \sqsubset_{\epsilon \tarrow
      \epsilon \tarrow o} \, \synbrack{\textbf{B}_\beta}}$ iff
  $\textbf{A}_\alpha$ is a proper subexpression of $\textbf{B}_\beta$.
\end{lem}

\begin{proof}
By induction on the structure of $\textbf{B}_\beta$ using Axioms
B5.1--6, Universal Generalization, Universal Instantiation, the
Syntactic Law of Quotation, and propositional logic.
\end{proof} 

\begin{lem}\label{lem:syn-is-free-in}
$\proves{}{\mname{is-free-in}_{\epsilon \tarrow \epsilon \tarrow o} \,
    \synbrack{\textbf{x}_\alpha} \, \synbrack{\textbf{B}_\beta}}$
  {\sglsp} iff {\sglsp} $\textbf{x}_\alpha$ is free in
  $\textbf{B}_\beta$.
\end{lem}

\begin{proof}
By induction on the structure of $\textbf{B}_\beta$ using Axioms
B1.1--4, B2.1--4, B7.1--8, Lemma~\ref{lem:syn-is-expr}, Universal
Generalization, Universal Instantiation, the Syntactic Law of
Quotation, and propositional logic.
\end{proof}

\subsection{Evaluations}

\begin{thm}[Syntactic Law of Disquotation]\label{thm:syn-disquotation}
$\proves{}{\sembrack{\synbrack{\textbf{A}_\alpha}}_\alpha =
    \textbf{A}_\alpha}$.
\end{thm}

\begin{proof}
Let $\textbf{A}_\alpha$ be eval-free.  Our proof is by induction on
the structure of $\textbf{A}_\alpha$.  There are five cases
corresponding to the five formation rules for eval-free expressions.

\bi

  \item[] \textbf{Case 1}: $\textbf{A}_\alpha$ is a variable
    $\textbf{x}_\alpha$. $\proves{}{\sembrack{\synbrack{\textbf{x}_\alpha}}_\alpha
    = \textbf{x}_\alpha}$ by Axiom B10.1.

  \item[] \textbf{Case 2}: $\textbf{A}_\alpha$ is a constant
    $\textbf{c}_\alpha$. $\proves{}{\sembrack{\synbrack{\textbf{c}_\alpha}}_\alpha
    = \textbf{c}_\alpha}$ by Axiom B10.2.

  \item[] \textbf{Case 3}: $\textbf{A}_\alpha$ is a function
    application $\textbf{F}_{\beta \tarrow \gamma} \, \textbf{B}_\beta$.
    \begin{align} \setcounter{equation}{0}
    & \vdash \mname{is-expr}_{\epsilon \tarrow o}^{\beta \tarrow \gamma} \, 
      \synbrack{\textbf{F}_{\beta \tarrow \gamma}}\\
    & \vdash \mname{is-expr}_{\epsilon \tarrow o}^{\beta} \, 
      \synbrack{\textbf{B}_\beta}\\
    & \vdash \sembrack{\mname{app}_{\epsilon \tarrow \epsilon \tarrow \epsilon} \, 
      \synbrack{\textbf{F}_{\beta \tarrow \gamma}} \,
      \synbrack{\textbf{B}_\beta}}_\gamma =
      \sembrack{\synbrack{\textbf{F}_{\beta \tarrow \gamma}}}_{\beta \tarrow \gamma} \,
      \sembrack{\synbrack{\textbf{B}_\beta}}_{\beta}\\
    & \vdash \sembrack{\synbrack{\textbf{F}_{\beta \tarrow \gamma}}}_{\beta \tarrow \gamma}
      = \textbf{F}_{\beta \tarrow \gamma}\\
    & \vdash \sembrack{\synbrack{\textbf{B}_\beta}}_{\beta} = \textbf{B}_\beta\\
    & \vdash \synbrack{\textbf{F}_{\beta \tarrow \gamma} \, \textbf{B}_\beta} =
      \mname{app}_{\epsilon \tarrow \epsilon \tarrow \epsilon} \, 
      \synbrack{\textbf{F}_{\beta \tarrow \gamma}} \,
      \synbrack{\textbf{B}_\beta}\\
    & \vdash \sembrack{\synbrack{\textbf{F}_{\beta \tarrow \gamma} \, 
      \textbf{B}_\beta}}_\gamma = 
      \textbf{F}_{\beta \tarrow \gamma} \, \textbf{B}_\beta
    \end{align}
    (1) and (2) are by Lemma~\ref{lem:syn-is-expr}; (3) follows from
    (1) and (2) by Axiom B10.3; (4) and (5) are by the induction
    hypothesis; (6) is an instance of Axiom B8.1; and (7) follows from
    (3)--(6) by the Equality Rules.

  \item[] \textbf{Case 4}: $\textbf{A}_\alpha$ is an function
    abstraction $\LambdaApp \textbf{x}_\beta \mdot \textbf{B}_\gamma$.
    \begin{align} \setcounter{equation}{0}
    & \vdash \mname{is-expr}_{\epsilon \tarrow o}^{\gamma} \, 
      \synbrack{\textbf{B}_\gamma}\\
    & \vdash \Neg(\mname{is-free-in}_{\epsilon \tarrow \epsilon \tarrow o} \, 
      \synbrack{\textbf{x}_\beta} \,
      \synbrack{\synbrack{\textbf{B}_\gamma}})\\
    & \vdash \sembrack{\mname{abs}_{\epsilon \tarrow \epsilon \tarrow \epsilon} \, 
      \synbrack{\textbf{x}_\beta} \,
      \synbrack{\textbf{B}_\gamma}}_{\beta \tarrow \gamma} = 
      \LambdaApp \textbf{x}_\beta \mdot 
      \sembrack{\synbrack{\textbf{B}_\gamma}}_\gamma\\
    & \vdash \sembrack{\synbrack{\textbf{B}_\gamma}}_\gamma =
      \textbf{B}_\gamma\\
    & \vdash \synbrack{\LambdaApp \textbf{x}_\beta \mdot \textbf{B}_\gamma} =
      \mname{abs}_{\epsilon \tarrow \epsilon \tarrow \epsilon} \, 
      \synbrack{\textbf{x}_\beta} \, \synbrack{\textbf{B}_\gamma}\\
    & \vdash \sembrack{\synbrack{\LambdaApp \textbf{x}_\beta \mdot 
      \textbf{B}_\gamma}}_{\beta \tarrow \gamma} = 
      \LambdaApp \textbf{x}_\beta \mdot \textbf{B}_\gamma
    \end{align}
    (1) is by Lemma~\ref{lem:syn-is-expr}; (2) is by
    Lemma~\ref{lem:syn-is-free-in}; (3) follows from (1) and (2) by
    Axiom B10.4; (4) is by the induction hypothesis; (5) is an instance
    of Axiom B8.2; and (6) follows from (3)--(5) by the Equality
    Rules.

  \item[] \textbf{Case 5}: $\textbf{A}_\alpha$ is a quotation
    $\synbrack{\textbf{B}_\beta}$.
    \begin{align} \setcounter{equation}{0} 
    & \vdash \mname{is-expr}_{\epsilon \tarrow o}^{\beta} \, 
      \synbrack{\textbf{B}_\beta}\\
    & \vdash \sembrack{\mname{quo}_{\epsilon \tarrow \epsilon} \,
      \synbrack{\textbf{B}_\beta}}_\epsilon = \synbrack{\textbf{B}_\beta}\\
    & \vdash \synbrack{\synbrack{\textbf{B}_\beta}} =
      \mname{quo}_{\epsilon \tarrow \epsilon} \, \synbrack{\textbf{B}_\beta}\\
    & \vdash \sembrack{\synbrack{\synbrack{\textbf{B}_\beta}}}_\epsilon =
      \synbrack{\textbf{B}_\beta}   
    \end{align}
    (1) is by Lemma~\ref{lem:syn-is-expr}; (2) follows from (1) by
    Axiom B10.5; (3) is an instance of Axiom B8.3; and (4) follows
    from (2) and (3) by the Equality Rules.

\ei
\end{proof}

\bigskip

The next lemma, which illustrates the application of Axiom B11.2, will
be employed in section~\ref{sec:revisit}.

\begin{lem}\label{lem:axiomB11.2-app}
\be

  \item[]

  \item If $\textbf{A}_\alpha$ is eval-free and $\textbf{x}_\epsilon$
    is not free in $\textbf{A}_\alpha$, then \[\proves{}{(\LambdaApp
      \textbf{x}_\epsilon \mdot \sembrack{\textbf{x}_\epsilon}_\alpha)
      \, \synbrack{\textbf{A}_\alpha} = \textbf{A}_\alpha}.\]

  \item If $\proves{T,\sH}{\mname{is-expr}_{\epsilon \tarrow
      o}^{\alpha} \, \textbf{B}_\epsilon}$ and
    $\proves{T,\sH}{\Neg\mname{is-free-in}_{\epsilon \tarrow \epsilon
      \tarrow o} \, \synbrack{\textbf{x}_\epsilon} \,
    \textbf{B}_\epsilon}$, then \[\proves{T,\sH}{(\LambdaApp
    \textbf{x}_\epsilon \mdot \sembrack{\textbf{x}_\epsilon}_\alpha)
    \, \textbf{B}_\epsilon = \sembrack{\textbf{B}_\epsilon}_\alpha}.\]

  \item If $\proves{T,\sH}{\mname{is-expr}_{\epsilon \tarrow
      o}^{\alpha} \, \textbf{y}_\epsilon}$,
    $\proves{T,\sH}{\Neg\mname{is-free-in}_{\epsilon \tarrow \epsilon
      \tarrow o} \, \synbrack{\textbf{x}_\epsilon} \,
    \textbf{y}_\epsilon}$, and $\textbf{x}_\epsilon$ and
    $\textbf{y}_\epsilon$ are different variables,
    then \[\proves{T,\sH}{(\LambdaApp \textbf{x}_\epsilon \mdot
      \sembrack{\textbf{y}_\epsilon}_\alpha) \, \textbf{B}_\epsilon =
      \sembrack{\textbf{y}_\epsilon}_\alpha}.\]

  \item \bsp Let $\textbf{B}_\epsilon$ be eval-free;
    $\proves{T,\sH}{\Neg\mname{is-free-in}_{\epsilon \tarrow \epsilon
      \tarrow o} \, \synbrack{\textbf{x}_\epsilon} \,
    \synbrack{\textbf{B}_\epsilon}}$;
    $\proves{T,\sH}{\mname{is-expr}_{\epsilon \tarrow o}^{\alpha} \,
    \textbf{B}_\epsilon}$; and
    $\proves{T,\sH}{\Neg\mname{is-free-in}_{\epsilon \tarrow \epsilon
      \tarrow o} \, \synbrack{\textbf{x}_\epsilon} \,
    \textbf{B}_\epsilon}$.  Also assume that there is some variable
    $\textbf{y}_\epsilon$ different from $\textbf{x}_\epsilon$ such
    that \[\proves{T,\sH}{\Neg
      \mname{IS-EFFECTIVE-IN}(\textbf{y}_\epsilon,\textbf{H}_o)}\] for
    all $\textbf{H}_o \in \sH$. Then \[\proves{T,\sH}{\Neg
      \mname{IS-EFFECTIVE-IN}( \textbf{x}_\epsilon,
      \sembrack{\textbf{B}_\epsilon}_\alpha)}.\]\esp

\ee
\end{lem}

\begin{proof}

\medskip

\noindent \textbf{Part 1}
\begin{align} \setcounter{equation}{0}
& \vdash (\LambdaApp \textbf{x}_\epsilon \mdot \textbf{x}_\epsilon) \, 
  \synbrack{\textbf{A}_\alpha} = 
  \synbrack{\textbf{A}_\alpha}\\
& \vdash \mname{is-expr}_{\epsilon \tarrow o}^{\alpha} \,
  ((\LambdaApp \textbf{x}_\epsilon \mdot \textbf{x}_\epsilon) \, 
  \synbrack{\textbf{A}_\alpha})\\
& \vdash \Neg\mname{is-free-in}_{\epsilon \tarrow \epsilon \tarrow o} \,
  \synbrack{\textbf{x}_\epsilon} \, 
  ((\LambdaApp \textbf{x}_\epsilon \mdot \textbf{x}_\epsilon) \, 
  \synbrack{\textbf{A}_\alpha})\\
& \vdash (\LambdaApp \textbf{x}_\epsilon \mdot \sembrack{\textbf{x}_\epsilon}_\alpha) \, 
  \synbrack{\textbf{A}_\alpha} =
  \sembrack{(\LambdaApp \textbf{x}_\epsilon \mdot \textbf{x}_\epsilon) \, 
  \synbrack{\textbf{A}_\alpha}}_\alpha\\
& \vdash (\LambdaApp \textbf{x}_\epsilon \mdot 
  \sembrack{\textbf{x}_\epsilon}_\alpha) \,
  \synbrack{\textbf{A}_\alpha} = \textbf{A}_\alpha
\end{align}
(1) is an instance of Axiom A4.2; (2) follows from (1) by
Lemma~\ref{lem:syn-is-expr} and the Equality Rules; (3) follows from
(1) by Lemma~\ref{lem:syn-is-free-in}, the hypothesis, and the
Equality Rules; (4) follows from (2), (3), and Axiom B11.2 by Modus
Ponens; (5) follows from (1) and (4) by the Equality Rules and the
Syntactic Law of Disquotation.

\medskip

\noindent \textbf{Part 2}
\begin{align} \setcounter{equation}{0}
& \vdash (\LambdaApp \textbf{x}_\epsilon \mdot \textbf{x}_\epsilon) \, 
  \textbf{B}_\epsilon = 
  \textbf{B}_\epsilon\\
T,\sH & \vdash \mname{is-expr}_{\epsilon \tarrow o}^{\alpha} \,
  ((\LambdaApp \textbf{x}_\epsilon \mdot \textbf{x}_\epsilon) \, 
  \textbf{B}_\epsilon)\\
T,\sH & \vdash \Neg\mname{is-free-in}_{\epsilon \tarrow \epsilon \tarrow o} \,
  \synbrack{\textbf{x}_\epsilon} \, 
  ((\LambdaApp \textbf{x}_\epsilon \mdot \textbf{x}_\epsilon) \, 
  \textbf{B}_\epsilon)\\
T,\sH & \vdash (\LambdaApp \textbf{x}_\epsilon \mdot 
  \sembrack{\textbf{x}_\epsilon}_\alpha) \, 
  \textbf{B}_\epsilon =
  \sembrack{(\LambdaApp \textbf{x}_\epsilon \mdot \textbf{x}_\epsilon) \, 
  \textbf{B}_\epsilon}_\alpha\\
T,\sH & \vdash (\LambdaApp \textbf{x}_\epsilon \mdot 
  \sembrack{\textbf{x}_\epsilon}_\alpha) \,
  \textbf{B}_\epsilon = \sembrack{\textbf{B}_\epsilon}_\alpha
\end{align}
(1) is an instance of Axiom A4.2; (2) and (3) follow from (1) and the
two hypotheses by the Equality Rules; (4) follows from (2), (3), and
Axiom B11.2 by Modus Ponens; (5) follows from (1) and (4) by the
Equality Rules.

\medskip

\noindent \textbf{Part 3} {\sglsp} Similar to the proof of part 2.

\medskip

\noindent \textbf{Part 4} {\sglsp} By the definitions of
$\mname{IS-EFFECTIVE-IN}$ and $\Forsome$ and propositional logic, we
must show
\[\proves{T,\sH}{\ForallApp \textbf{y}_\epsilon \mdot ((\LambdaApp
  \textbf{x}_\epsilon \mdot \sembrack{\textbf{B}_\epsilon}_\alpha) \,
  \textbf{y}_\epsilon} = \sembrack{\textbf{B}_\epsilon}_\alpha).\]
\begin{align} \setcounter{equation}{0}
& T,\sH \vdash (\LambdaApp \textbf{x}_\epsilon \mdot \textbf{B}_\epsilon) \, 
  \textbf{y}_\epsilon = 
  \textbf{B}_\epsilon\\
& T,\sH \vdash \mname{is-expr}_{\epsilon \tarrow o}^{\alpha} \,
  ((\LambdaApp \textbf{x}_\epsilon \mdot \textbf{B}_\epsilon) \, 
  \textbf{y}_\epsilon)\\
& T,\sH \vdash \Neg\mname{is-free-in}_{\epsilon \tarrow \epsilon \tarrow o} \,
  \synbrack{\textbf{x}_\epsilon} \, 
  ((\LambdaApp \textbf{x}_\epsilon \mdot \textbf{B}_\epsilon) \, 
  \textbf{y}_\epsilon)\\
& T,\sH \vdash (\LambdaApp \textbf{x}_\epsilon \mdot 
  \sembrack{\textbf{B}_\epsilon}_\alpha) \, 
  \textbf{y}_\epsilon =
  \sembrack{(\LambdaApp \textbf{x}_\epsilon \mdot \textbf{B}_\epsilon) \, 
  \textbf{y}_\epsilon}_\alpha\\
& T,\sH \vdash (\LambdaApp \textbf{x}_\epsilon \mdot 
  \sembrack{\textbf{B}_\epsilon}_\alpha) \,
  \textbf{y}_\epsilon = \sembrack{\textbf{B}_\epsilon}_\alpha\\
& T,\sH \vdash \ForallApp \textbf{y}_\epsilon \mdot 
  ((\LambdaApp \textbf{x}_\epsilon \mdot 
  \sembrack{\textbf{B}_\epsilon}_\alpha) \,
  \textbf{y}_\epsilon = \sembrack{\textbf{B}_\epsilon}_\alpha)
\end{align}
(1) follows from the second hypothesis by Beta-Reduction by
Substitution and Lemma~\ref{lem:syn-is-free-in}; (2) and (3) follow
from (1) and the third and fourth hypotheses by the Equality Rules;
(4) follows from (2), (3), and Axiom B11.2 by Modus Ponens; (5)
follows from (1) and (4) by the Equality Rules; and (6) follows from
(5) and the hypothesis about $\textbf{y}_\epsilon$ by Universal
Generalization.
\end{proof}

\section{Soundness}\label{sec:soundness}

The proof system for {\churchqe} is \emph{sound} if
$\proves{T}{\textbf{A}_o}$ implies $T \vDash \textbf{A}_o$ whenever
$T$ is a theory of {\churchqe} and $\textbf{A}_o$ is a formula of $T$.
We will prove that the proof system for {\churchqe} is sound by
showing that its axioms are valid in all general models for
{\churchqe} and its single rule of inference preserves validity in
all general models for {\churchqe}.  

\subsection{Lemmas Concerning Semantics}

\begin{lem}\label{lem:val-a}
Let $\sM = (\set{D_\alpha \;|\; \alpha \in \sT}, I)$ be a general
model for {\churchqe} and $\phi \in \mname{assign}(\sM)$.  

\be

  \item $V^{\cal M}_{\phi}(\textbf{A}_\alpha = \textbf{B}_\alpha) =
    \TRUE$ {\sglsp} iff {\sglsp} $V^{\cal M}_{\phi}(\textbf{A}_\alpha)
    = V^{\cal M}_{\phi}(\textbf{B}_\alpha)$.

  \item $V^{\cal M}_{\phi}((\LambdaApp \textbf{x}_\alpha \mdot
    \textbf{B}_\beta) \, \textbf{A}_\alpha) = V^{\cal M}_{\phi[{\bf
          x}_\alpha \mapsto V^{\cal M}_{\phi}({\bf
          A}_\alpha)]}(\textbf{B}_{\beta})$.

  \item $V^{\cal M}_{\phi}(\LambdaApp \textbf{x}_\alpha \mdot
    \textbf{B}_\beta) = V^{\cal M}_{\phi[{\bf x}_\alpha \mapsto
        d]}(\LambdaApp \textbf{x}_\alpha \mdot \textbf{B}_{\beta})$
    for all $d \in D_\alpha$.

  \item $V^{\cal M}_{\phi}(\ForallApp \textbf{x}_\alpha \mdot
    \textbf{A}_o) = \TRUE$ {\sglsp} iff {\sglsp} $V^{\cal
      M}_{\phi[{\bf x}_\alpha \mapsto d]}(\textbf{A}_o) = \TRUE$ for
    all $d \in D_\alpha$.

\ee

\end{lem}

\begin{proof}
Part 1 is by the semantics of equality; part 2 is by the semantics of
function application and abstraction; part 3 is by the semantics of
function abstraction; and part 4 is by parts 1 and 3 and the
definition of $\Forall$.
\end{proof}

\begin{lem}\label{lem:val-b}
Let $\sM$ be a general model for {\churchqe}, $\textbf{A}_\alpha$ be
eval-free, and $\phi, \psi \in \mname{assign}(\sM)$ agree on all free
variables of $\textbf{A}_\alpha$.  Then $V^{\cal
  M}_{\phi}(\textbf{A}_\alpha) = V^{\cal
  M}_{\psi}(\textbf{A}_\alpha)$.
\end{lem} 

\begin{proof}
By induction on the structure of $\textbf{A}_\alpha$.
\end{proof}

\begin{lem}\label{lem:not-effective-sem}
Let $\sM = (\set{D_\alpha \;|\; \alpha \in \sT}, I)$ be a general
model for {\churchqe}.  $\sM \vDash \Neg
\mname{IS-EFFECTIVE-IN}(\textbf{x}_\alpha, \textbf{B}_\beta)$ {\sglsp}
iff {\sglsp} $V^{\cal M}_{\phi[{\bf x}_\alpha \mapsto
    d_1]}(\textbf{B}_\beta) = V^{\cal M}_{\phi[{\bf x}_\alpha \mapsto
    d_2]}(\textbf{B}_\beta)$ for all $\phi \in \mname{assign}(\sM)$
and all $d_1,d_2 \in D_\alpha$ with $d_1 \not= d_2$.
\end{lem}

\begin{proof}
Assume $\textbf{x}_\alpha$ and $\textbf{y}_\alpha$ are different variables.
\begin{align} \setcounter{equation}{0}
& \sM \vDash \Neg \mname{IS-EFFECTIVE-IN}(\textbf{x}_\alpha, \textbf{B}_\beta)\\
& \sM \vDash \Neg (\ForsomeApp \textbf{y}_\alpha \mdot 
  ((\LambdaApp \textbf{x}_\alpha \mdot \textbf{B}_\beta) \, \textbf{y}_\alpha \not=
  \textbf{B}_\beta))\\
& \sM \vDash \ForallApp \textbf{y}_\alpha \mdot 
  ((\LambdaApp \textbf{x}_\alpha \mdot \textbf{B}_\beta) \, \textbf{y}_\alpha =
  \textbf{B}_\beta)\\
&  V^{\cal M}_{\phi}(\ForallApp \textbf{y}_\alpha \mdot 
  ((\LambdaApp \textbf{x}_\alpha \mdot \textbf{B}_\beta) \, \textbf{y}_\alpha =
  \textbf{B}_\beta)) = \TRUE {\dblsp}
  \mbox{for all } \phi \in \mname{assign}(\sM)\\
& V^{\cal M}_{\phi}((\LambdaApp \textbf{x}_\alpha \mdot \textbf{B}_\beta) \, 
  \textbf{y}_\alpha) =
  V^{\cal M}_{\phi}(\textbf{B}_\beta) {\dblsp}
  \mbox{for all } \phi \in \mname{assign}(\sM)\\
& V^{\cal M}_{\phi[{\bf x}_\alpha \mapsto V^{\cal M}_{\phi}({\bf y}_\alpha)]}
  (\textbf{B}_\beta) =
  V^{\cal M}_{\phi}(\textbf{B}_\beta) {\dblsp}
  \mbox{for all } \phi \in \mname{assign}(\sM)
\end{align}
(1) is the left side of the iff statement of the lemma; (2) holds iff
(1) holds by the definition of \mname{IS-EFFECTIVE-IN}; (3) holds iff
(2) holds by the definition of $\Forsome$; (4) holds iff (3) holds by
the definition of a valid formula; (5) holds iff (4) holds by parts 1
and 4 of Lemma~\ref{lem:val-a}; and (6) holds iff (5) holds by part 2
of Lemma~\ref{lem:val-a}.

Now let $\psi \in \mname{assign}(\sM)$ and $d_1,d_2 \in D_\alpha$ with
$d_1 \not= d_1$.
\begin{align} 
& V^{\cal M}_{\psi[{\bf x}_\alpha \mapsto d_1]}(\textbf{B}_\beta)\\
& = V^{\cal M}_{\psi[{\bf x}_\alpha \mapsto d_1]
  [{\bf x}_\alpha \mapsto V^{\cal M}_{\psi}({\bf y}_\alpha)]}(\textbf{B}_\beta)\\
& = V^{\cal M}_{\psi[{\bf x}_\alpha \mapsto d_2]
  [{\bf x}_\alpha \mapsto V^{\cal M}_{\psi}({\bf y}_\alpha)]}(\textbf{B}_\beta)\\
& = V^{\cal M}_{\psi[{\bf x}_\alpha \mapsto d_2]}(\textbf{B}_\beta)
\end{align}
(8) and (10) are by (6) and (9) holds since \[\psi[{\bf x}_\alpha
  \mapsto d] [{\bf x}_\alpha \mapsto V^{\cal M}_{\psi}({\bf
    y}_\alpha)] = \psi[{\bf x}_\alpha \mapsto V^{\cal M}_{\psi}({\bf
    y}_\alpha)]\] for any $d \in D_\alpha$.
Hence (6) implies 
\begin{align} 
& V^{\cal M}_{\phi[{\bf x}_\alpha \mapsto d_1]}(\textbf{B}_\beta)
  = V^{\cal M}_{\phi[{\bf x}_\alpha \mapsto d_2]}(\textbf{B}_\beta)
\end{align}
for all $\phi \in \mname{assign}(\sM)$ and $d_1,d_2 \in D_\alpha$ with
$d_1 \not= d_1$, and obviously (11) implies (6).  Therefore, (1) holds
iff (11) holds.
\end{proof}

\begin{lem}\label{lem:sem-is-expr}
Let $\sM$ be a general model for {\churchqe} and $\phi \in
\mname{assign}(\sM)$.

\be

  \item $V^{\cal M}_{\phi}(\mname{is-expr}_{\epsilon \tarrow
    o}^{\alpha} \, \sE(\textbf{A}_\alpha)) = \TRUE$.

  \item $V^{\cal M}_{\phi}(\mname{is-expr}_{\epsilon \tarrow
    o}^{\alpha} \, \synbrack{\textbf{A}_\alpha}) = \TRUE$.

  \item $V^{\cal M}_{\phi}(\mname{is-expr}_{\epsilon \tarrow
    o}^{\alpha} \, \textbf{B}_\epsilon) = \TRUE$ implies $V^{\cal
    M}_{\phi}(\textbf{B}_\epsilon) = V^{\cal
    M}_{\phi}(\sE(\textbf{A}_\alpha))$ for some (eval-free)
    $\textbf{A}_\alpha$.

\ee
\end{lem}

\begin{proof}

\medskip

\bsp
\noindent \textbf{Part 1} {\sglsp} $V^{\cal
  M}_{\phi}(\sE(\textbf{A}_\alpha)) = \sE(\textbf{A}_\alpha)$ by
Proposition~\ref{prop:val-const}.  This implies $V^{\cal
  M}_{\phi}(\mname{is-expr}_{\epsilon \tarrow o}^{\alpha} \,
\sE(\textbf{A}_\alpha)) = \TRUE$ by the semantics of
$\mname{is-expr}_{\epsilon \tarrow o}^{\alpha}$.
\esp

\medskip

\noindent \textbf{Part 2} {\sglsp} $V^{\cal
  M}_{\phi}(\synbrack{\textbf{A}_\alpha}) = \sE(\textbf{A}_\alpha)$ by
condition 5 of the definition of a general model.  This implies
$V^{\cal M}_{\phi}(\mname{is-expr}_{\epsilon \tarrow o}^{\alpha} \,
\synbrack{\textbf{A}_\alpha}) = \TRUE$ by the semantics of
$\mname{is-expr}_{\epsilon \tarrow o}^{\alpha}$.

\medskip

\noindent \textbf{Part 3} {\sglsp} By the hypothesis and the semantics
of $\mname{is-expr}_{\epsilon \tarrow o}^{\alpha}$, $V^{\cal
  M}_{\phi}(\textbf{B}_\epsilon) = \sE(\textbf{A}_\alpha)$ for some
$\textbf{A}_\alpha$.  This implies the conclusion of the lemma by
Proposition~\ref{prop:val-const}.
\end{proof} 

\begin{lem}\label{lem:sem-is-free-in}
Let $\sM$ be a general model for {\churchqe}, $\phi \in
\mname{assign}(\sM)$.  Then $V^{\cal
  M}_{\phi}(\mname{is-free-in}_{\epsilon \tarrow \epsilon \tarrow o}
\, \synbrack{\textbf{x}_\alpha} \, \synbrack{\textbf{B}_\beta}) =
\TRUE$ {\sglsp} iff {\sglsp} $\textbf{x}_\alpha$ is free in
$\textbf{B}_\beta$.
\end{lem}

\begin{proof}
By induction on the structure of $\textbf{B}_\beta$.
\end{proof}

\begin{lem}\label{lem:alt-disquotation}
Let $\sM$ be a general model for {\churchqe} and $\phi \in
\mname{assign}(\sM)$.  Then $V^{\cal
  M}_{\phi}(\sembrack{\sE(\textbf{A}_\alpha)}_\alpha) = V^{\cal
  M}_{\phi}(\textbf{A}_\alpha)$.
\end{lem}

\begin{proof}
\begin{align} \setcounter{equation}{0}
&
V^{\cal M}_{\phi}(\sembrack{\sE(\textbf{A}_\alpha)}_\alpha) \\
&=
V^{\cal M}_{\phi}(\sE^{-1}(V^{\cal M}_{\phi}(\sE(\textbf{A}_\alpha)))) \\
&=
V^{\cal M}_{\phi}(\sE^{-1}(\sE(\textbf{A}_\alpha))) \\
&=
V^{\cal M}_{\phi}(\textbf{A}_\alpha)
\end{align}
(2) is by part 1 of Lemma~\ref{lem:sem-is-expr} and condition 6 of
the definition of a general model; (3) is by
Proposition~\ref{prop:val-const}; and (4) is immediate.
\end{proof}

\subsection{Soundness of Axioms and Rule of Inference}

\begin{lem}[Axioms are Valid]\label{lem:axioms}
Each axiom of the proof system for {\churchqe} is valid in
{\churchqe}.
\end{lem}

\begin{proof}
Let $\sM = (\set{D_\alpha \;|\; \alpha \in \sT}, I)$ be a general
model for {\churchqe} and $\phi \in \mname{assign}(\sM)$.  We must
prove that each of the 64 axioms is valid in $\sM$.

\medskip

\noindent \textbf{Axiom A1} {\sglsp} The proof is the same as the
proof of 5402 for Axiom 1 in~\cite[p.~241]{Andrews02}.

\medskip

\noindent \textbf{Axiom A2} {\sglsp} The proof is the same as the
proof of 5402 for Axiom 2 in~\cite[p.~242]{Andrews02}.

\medskip

\noindent \textbf{Axiom A3} {\sglsp} The proof is the same as the
proof of 5402 for Axiom 3 in~\cite[p.~242]{Andrews02}.

\medskip

\noindent \textbf{Axiom Group A4}

\bi

  \item[] \textbf{Axiom A4.1} {\sglsp} Let $\textbf{x}_\alpha$ and
    $\textbf{y}_\beta$ be distinct.  We must show
    \[V^{\cal M}_{\phi}((\LambdaApp \textbf{x}_\alpha \mdot \textbf{y}_\beta) \,
      \textbf{A}_\alpha) = V^{\cal M}_{\phi}(\textbf{y}_\beta)\] to
      prove Axiom A4.1 is valid in $\sM$.
    \begin{align} \setcounter{equation}{0}
    &
    V^{\cal M}_{\phi}((\LambdaApp \textbf{x}_\alpha \mdot \textbf{y}_\beta) \,
      \textbf{A}_\alpha)\\
    &
    = V^{\cal M}_{\phi[{\bf x}_\alpha \mapsto V^{\cal M}_{\phi}({\bf  A}_\alpha)]}
    (\textbf{y}_\beta)\\
    &
    = V^{\cal M}_{\phi}(\textbf{y}_\beta)
    \end{align}
    (2) is by part 2 of Lemma~\ref{lem:val-a} and (3) is by
    Lemma~\ref{lem:val-b}.

  \item[] \textbf{Axiom A4.2} {\sglsp} We must show
    \[V^{\cal M}_{\phi}((\LambdaApp \textbf{x}_\alpha \mdot \textbf{x}_\alpha) \,
      \textbf{A}_\alpha) = V^{\cal M}_{\phi}(\textbf{A}_\alpha)\] to
      prove Axiom A4.2 is valid in $\sM$.
    \begin{align} \setcounter{equation}{0}
    &
    V^{\cal M}_{\phi}((\LambdaApp \textbf{x}_\alpha \mdot \textbf{x}_\alpha) \,
      \textbf{A}_\alpha)\\
    &
    = V^{\cal M}_{\phi[{\bf x}_\alpha \mapsto V^{\cal M}_{\phi}({\bf  A}_\alpha)]}
    (\textbf{x}_\alpha)\\
    &
    = V^{\cal M}_{\phi}(\textbf{A}_\alpha)
    \end{align}
    (2) is by part 2 of Lemma~\ref{lem:val-a} and (3) is by
    the semantics of variables.

  \item[] \textbf{Axiom A4.3} {\sglsp} Similar to Axiom A4.1.

  \item[] \textbf{Axiom A4.4} {\sglsp} We must show
    \begin{align*}
    &
    V^{\cal M}_{\phi}((\LambdaApp \textbf{x}_\alpha \mdot 
    (\textbf{B}_{\beta \tarrow \gamma} \, 
    \textbf{C}_\beta)) \, \textbf{A}_\alpha)\\
    &
    = V^{\cal M}_{\phi}(((\LambdaApp \textbf{x}_\alpha \mdot
    \textbf{B}_{\beta \tarrow \gamma}) \, 
    \textbf{A}_\alpha) \, 
    ((\LambdaApp \textbf{x}_\alpha \mdot \textbf{C}_\beta) \, \textbf{A}_\alpha))
    \end{align*}
    to prove Axiom A4.4 is valid in $\sM$.
    \begin{align} \setcounter{equation}{0}
    &
    V^{\cal M}_{\phi}((\LambdaApp \textbf{x}_\alpha \mdot 
    (\textbf{B}_{\beta \tarrow \gamma} \, \textbf{C}_\beta)) \, \textbf{A}_\alpha)\\
    &
    = V^{\cal M}_{\phi[{\bf x}_\alpha \mapsto V^{\cal M}_{\phi}({\bf  A}_\alpha)]}
    (\textbf{B}_{\beta \tarrow \gamma} \, \textbf{C}_\beta)\\
    &
    = V^{\cal M}_{\phi[{\bf x}_\alpha \mapsto V^{\cal M}_{\phi}({\bf  A}_\alpha)]}
    (\textbf{B}_{\beta \tarrow \gamma})
    (V^{\cal M}_{\phi[{\bf x}_\alpha \mapsto V^{\cal M}_{\phi}({\bf  A}_\alpha)]}
    (\textbf{C}_\beta))\\
    &
    = V^{\cal M}_{\phi}
    ((\LambdaApp \textbf{x}_\alpha \mdot \textbf{B}_{\beta \tarrow \gamma}) \, 
    \textbf{A}_\alpha)
    (V^{\cal M}_{\phi}
    ((\LambdaApp \textbf{x}_\alpha \mdot \textbf{C}_\beta) \,
    \textbf{A}_\alpha))\\
    &
    = V^{\cal M}_{\phi}
    (((\LambdaApp \textbf{x}_\alpha \mdot \textbf{B}_{\beta \tarrow \gamma}) \, 
    \textbf{A}_\alpha) \,
    ((\LambdaApp \textbf{x}_\alpha \mdot \textbf{C}_\beta) \,
    \textbf{A}_\alpha))
    \end{align}
    (2) and (4) are by part 2 of Lemma~\ref{lem:val-a} and (3) and (5)
    are by the semantics of function application.

  \item[] \textbf{Axiom A4.5} {\sglsp} Let $\textbf{x}_\alpha$ and
    $\textbf{y}_\beta$ be distinct and either (a) $\textbf{A}_\alpha$ is
    eval-free and $\textbf{y}_\beta$ is not free in
    $\textbf{A}_\alpha$ or (b) $\textbf{B}_\gamma$ is eval-free and
    $\textbf{x}_\alpha$ is not free in $\textbf{B}_\gamma$.  We must
    show
    \[V^{\cal M}_{\phi}((\LambdaApp \textbf{x}_\alpha \mdot \LambdaApp
      \textbf{y}_\beta \mdot \textbf{B}_\gamma) \,
      \textbf{A}_\alpha)(d) = V^{\cal M}_{\phi}(\LambdaApp
      \textbf{y}_\beta \mdot ((\LambdaApp \textbf{x}_\alpha \mdot
      \textbf{B}_\gamma) \, \textbf{A}_\alpha))(d),\] where $d \in
      D_\beta$, to prove Axiom 4.5 is valid in $\sM$.
    \begin{align} \setcounter{equation}{0}
    &
    V^{\cal M}_{\phi}((\LambdaApp \textbf{x}_\alpha \mdot 
    \LambdaApp \textbf{y}_\beta \mdot \textbf{B}_\gamma) \, \textbf{A}_\alpha)(d)\\
    &
    = V^{\cal M}_{\phi[{\bf x}_\alpha \mapsto V^{\cal M}_{\phi}({\bf  A}_\alpha)]}
    (\LambdaApp \textbf{y}_\beta \mdot \textbf{B}_\gamma)(d)\\
    &
    = V^{\cal M}_{\phi[{\bf x}_\alpha \mapsto V^{\cal M}_{\phi}({\bf  A}_\alpha)]
    [{\bf y}_\beta \mapsto d]}
    (\textbf{B}_\gamma)\\
    &
    = V^{\cal M}_{\phi[{\bf y}_\beta \mapsto d]
    [{\bf x}_\alpha \mapsto V^{\cal M}_{\phi}({\bf  A}_\alpha)]}
    (\textbf{B}_\gamma)\\
    &
    = V^{\cal M}_{\phi[{\bf y}_\beta \mapsto d]
    [{\bf x}_\alpha \mapsto V^{\cal M}_{\phi[{\bf y}_\beta \mapsto d]}({\bf  A}_\alpha)]}
    (\textbf{B}_\gamma)\\
    &
    = V^{\cal M}_{\phi[{\bf y}_\beta \mapsto d]}
    ((\LambdaApp \textbf{x}_\alpha \mdot \textbf{B}_\gamma) \, \textbf{A}_\alpha)\\
    &
    = V^{\cal M}_{\phi}
    (\LambdaApp \textbf{y}_\beta \mdot 
    ((\LambdaApp \textbf{x}_\alpha \mdot \textbf{B}_\gamma) \, \textbf{A}_\alpha))(d)
    \end{align}
    (2) and (6) are by part 2 of Lemma~\ref{lem:val-a}; (3) and (7)
    are by the semantics of function abstraction; (4) is by
    $\textbf{x}_\alpha$ and $\textbf{y}_\beta$ being distinct; and (5)
    is by the hypothesis and Lemma~\ref{lem:val-b}.

  \item[] \textbf{Axiom A4.6} {\sglsp} We must show
    \[V^{\cal M}_{\phi}((\LambdaApp \textbf{x}_\alpha \mdot \LambdaApp
      \textbf{x}_\alpha \mdot \textbf{B}_\beta) \, \textbf{A}_\alpha)
      = V^{\cal M}_{\phi}(\LambdaApp \textbf{x}_\alpha \mdot
      \textbf{B}_\beta)\] to prove Axiom A4.6 is valid in $\sM$.
    \begin{align} \setcounter{equation}{0}
    &
    V^{\cal M}_{\phi}((\LambdaApp \textbf{x}_\alpha \mdot 
    \LambdaApp \textbf{x}_\alpha \mdot 
    \textbf{B}_\beta) \, \textbf{A}_\alpha)\\
    &
    = V^{\cal M}_{\phi[{\bf x}_\alpha \mapsto V^{\cal M}_{\phi}({\bf  A}_\alpha)]}
    (\LambdaApp \textbf{x}_\alpha \mdot \textbf{B}_\beta)\\
    & 
    = V^{\cal M}_{\phi}(\LambdaApp \textbf{x}_\alpha \mdot \textbf{B}_\beta)
    \end{align}
    (2) and (3) are by parts 2 and 3 of Lemma~\ref{lem:val-a},
    respectively.

\ei

\medskip

\noindent \textbf{Axiom Group B1} {\sglsp} By clauses 2 and 3 of the
definition of an interpretation.

\medskip

\noindent \textbf{Axiom Group B2} {\sglsp} By clauses 4 and 5 of the
definition of an interpretation.

\medskip

\noindent \textbf{Axiom Group B3} {\sglsp} By clauses 9 and 10 of the
definition of an interpretation.

\medskip

\noindent \textbf{Axiom Group B4} {\sglsp} By clauses 2--8 of the
definition of an interpretation.

\medskip

\noindent \textbf{Axiom Group B5} {\sglsp} By clause 11 of the
definition of an interpretation.

\medskip

\noindent \textbf{Axiom B6} {\sglsp} By the definition of the
domain of constructions in a frame.

\medskip

\noindent \textbf{Axiom Group B7} {\sglsp} By clause 12 of the
definition of an interpretation.

\medskip

\noindent \textbf{Axiom Group B8} {\sglsp}  

\bi

  \item[] \textbf{Axiom B8.1} {\sglsp} We must show
    \[V^{\cal M}_{\phi}(\synbrack{\textbf{F}_{\alpha \tarrow \beta} \,
     \textbf{A}_\alpha}) = 
     V^{\cal M}_{\phi}(\mname{app}_{\epsilon \tarrow \epsilon \tarrow \epsilon} \, 
     \synbrack{\textbf{F}_{\alpha \tarrow \beta}} \, 
     \synbrack{\textbf{A}_\alpha})\] to prove Axiom B8.1 is valid in $\sM$.
    \begin{align} \setcounter{equation}{0}
    &
    V^{\cal M}_{\phi}(\synbrack{\textbf{F}_{\alpha \tarrow \beta} \,
    \textbf{A}_\alpha})\\
    &
    = \sE(\synbrack{\textbf{F}_{\alpha \tarrow \beta} \, \textbf{A}_\alpha})\\
    &
    = \mname{app}_{\epsilon \tarrow \epsilon \tarrow \epsilon} \,
    \sE(\textbf{F}_{\alpha \tarrow \beta}) \, 
    \sE(\textbf{A}_\alpha)\\
    &
    = \mname{app}_{\epsilon \tarrow \epsilon \tarrow \epsilon} \,
    V^{\cal M}_{\phi}(\synbrack{\textbf{F}_{\alpha \tarrow \beta}}) \, 
    V^{\cal M}_{\phi}(\synbrack{\textbf{A}_\alpha})\\
    &
    = V^{\cal M}_{\phi}(\mname{app}_{\epsilon \tarrow \epsilon \tarrow \epsilon} \,
    \synbrack{\textbf{F}_{\alpha \tarrow \beta}} \,
    \synbrack{\textbf{A}_\alpha})
    \end{align}
    (2) and (4) are by condition 5 of the definition of a general
    model; (3) is by the definition of $\sE$; and (5) is by the
    semantics of $\mname{app}_{\epsilon \tarrow \epsilon \tarrow
      \epsilon}$ and clause~6 of the definition of an interpretation.

  \item[] \textbf{Axiom B8.2} {\sglsp} Similar to Axiom B8.1.

  \item[] \textbf{Axiom B8.3} {\sglsp} Similar to Axiom B8.1.

\ei

\medskip

\noindent \textbf{Axiom B9} {\sglsp} We must show
    \[V^{\cal M}_{\phi}((\LambdaApp \textbf{x}_\alpha \mdot 
      \synbrack{\textbf{B}_\beta}) \,
      \textbf{A}_\alpha) = V^{\cal M}_{\phi}(\synbrack{\textbf{B}_\beta})\] to
      prove Axiom B9 is valid in $\sM$.
    \begin{align} \setcounter{equation}{0}
    &
    V^{\cal M}_{\phi}((\LambdaApp \textbf{x}_\alpha \mdot 
      \synbrack{\textbf{B}_\beta}) \, \textbf{A}_\alpha)\\
    &
    = V^{\cal M}_{\phi[{\bf x}_\alpha \mapsto V^{\cal M}_{\phi}({\bf  A}_\alpha)]}
    (\synbrack{\textbf{B}_\beta})\\
    &
    = V^{\cal M}_{\phi}(\synbrack{\textbf{B}_\beta})
    \end{align}
    (2) is by part 2 of Lemma~\ref{lem:val-a} and (3) is by
    Lemma~\ref{lem:val-b}.

\medskip

\noindent \textbf{Axiom Group B10}

\bi

  \item[] \textbf{Axiom B10.1} {\sglsp} We must show
    \[V^{\cal M}_{\phi}(\sembrack{\synbrack{\textbf{x}_\alpha}}_\alpha) = 
     V^{\cal M}_{\phi}(\textbf{x}_\alpha)\] to prove Axiom B10.1 is
     valid in $\sM$.  This equation follows from the Law of
     Disquotation and part 1 of Lemma~\ref{lem:val-a}.

  \item[] \textbf{Axiom B10.2} {\sglsp} Same proof as for Axiom B10.1.

  \item[] \textbf{Axiom B10.3} {\sglsp} Let (a) $V^{\cal
    M}_{\phi}(\mname{is-expr}_{\epsilon \tarrow o}^{\alpha \tarrow
    \beta} \, \textbf{A}_\epsilon) = \TRUE$ and\\ (b)~$V^{\cal
    M}_{\phi}(\mname{is-expr}_{\epsilon \tarrow o}^{\alpha} \,
    \textbf{B}_\epsilon) = \TRUE$.  We must show
    \[V^{\cal M}_{\phi}(\sembrack{\mname{app}_{\epsilon \tarrow \epsilon
          \tarrow \epsilon} \, \textbf{A}_\epsilon \,
      \textbf{B}_\epsilon}_{\beta}) = V^{\cal
      M}_{\phi}(\sembrack{\textbf{A}_\epsilon}_{\alpha \tarrow \beta}
    \, \sembrack{\textbf{B}_\epsilon}_{\alpha})\] to prove Axiom B10.3
    is valid in $\sM$.  By part 3 of Lemma~\ref{lem:sem-is-expr}, (a) and
    (b) imply (c) $V^{\cal M}_{\phi}(\textbf{A}_\epsilon) = V^{\cal
      M}_{\phi}(\sE(\textbf{C}_{\alpha \tarrow \beta}))$ for some
    $\textbf{C}_{\alpha \tarrow \beta}$ and (d)~$V^{\cal
      M}_{\phi}(\textbf{B}_\epsilon) = V^{\cal
      M}_{\phi}(\sE(\textbf{D}_\alpha))$ for some $\textbf{D}_\alpha$,
    respectively.
    \begin{align} \setcounter{equation}{0}
    &
    V^{\cal M}_{\phi}(\sembrack{\mname{app}_{\epsilon \tarrow \epsilon \tarrow \epsilon} \, 
    \textbf{A}_\epsilon \, \textbf{B}_\epsilon}_{\beta})\\
    &
    = V^{\cal M}_{\phi}(\sembrack{\mname{app}_{\epsilon \tarrow \epsilon \tarrow \epsilon} \, 
    \sE(\textbf{C}_{\alpha \tarrow \beta}) \, \sE(\textbf{D}_\alpha)}_{\beta})\\
    &
    = V^{\cal M}_{\phi}(\sembrack{\sE(\textbf{C}_{\alpha \tarrow \beta} \, \textbf{D}_\alpha)}_{\beta})\\
    &
    = V^{\cal M}_{\phi}(\textbf{C}_{\alpha \tarrow \beta} \, \textbf{D}_\alpha)\\
    &
    = V^{\cal M}_{\phi}(\textbf{C}_{\alpha \tarrow \beta})(V^{\cal M}_{\phi}(\textbf{D}_\alpha))\\
    &
    = V^{\cal M}_{\phi}(\sembrack{\sE(\textbf{C}_{\alpha \tarrow \beta})}_{\alpha \tarrow \beta})
    (V^{\cal M}_{\phi}(\sembrack{\sE(\textbf{D}_\alpha)}_\alpha))\\
    &
    = V^{\cal M}_{\phi}(\sembrack{\textbf{A}_\epsilon}_{\alpha \tarrow \beta})
    (V^{\cal M}_{\phi}(\sembrack{\textbf{B}_\epsilon}_\alpha))\\
    &
    = V^{\cal M}_{\phi}(\sembrack{\textbf{A}_\epsilon}_{\alpha \tarrow \beta} \,
    \sembrack{\textbf{B}_\epsilon}_\alpha)
    \end{align}
    (2) and (7) are by (c) and (d); (3) is by the definition of $\sE$;
    (4) and (6) are by Lemma~\ref{lem:alt-disquotation}; and (5)
    and (8) are by the semantics of function application.

  \item[] \bsp \textbf{Axiom B10.4} {\sglsp} Let (a) $V^{\cal
    M}_{\phi}(\mname{is-expr}_{\epsilon \tarrow o}^{\beta} \,
    \textbf{A}_\epsilon) = \TRUE$ and\\ (b) $V^{\cal
      M}_{\phi}(\Neg(\mname{is-free-in}_{\epsilon \tarrow \epsilon
      \tarrow o} \, \synbrack{\textbf{x}_\alpha} \,
    \synbrack{\textbf{A}_\epsilon})) = \TRUE$.  We must show
    \[V^{\cal M}_{\phi}(\sembrack{\mname{abs}_{\epsilon \tarrow \epsilon
          \tarrow \epsilon} \, \synbrack{\textbf{x}_\alpha} \,
      \textbf{A}_\epsilon}_{\alpha \tarrow \beta})(d) = V^{\cal
      M}_{\phi}(\LambdaApp \textbf{x}_\alpha \mdot
    \sembrack{\textbf{A}_\epsilon}_\beta)(d),\] where $d \in
    D_\alpha$, to prove Axiom B10.4 is valid in $\sM$. By part 3 of
    Lemma~\ref{lem:sem-is-expr}, (a) implies (c) $V^{\cal
      M}_{\phi}(\textbf{A}_\epsilon) = V^{\cal
      M}_{\phi}(\sE(\textbf{B}_\beta))$ for some $\textbf{B}_\beta$.
    By Lemma~\ref{lem:sem-is-free-in}, (b) implies (d)
    $\textbf{x}_\alpha$ is not free in $\textbf{A}_\epsilon$.
    \begin{align} \setcounter{equation}{0}
    &
    V^{\cal M}_{\phi}(
    \sembrack{\mname{abs}_{\epsilon \tarrow \epsilon \tarrow \epsilon} \, 
    \synbrack{\textbf{x}_\alpha} \, 
    \textbf{A}_\epsilon}_{\alpha \tarrow \beta})(d)\\
    &
    = V^{\cal M}_{\phi}(
    \sembrack{\mname{abs}_{\epsilon \tarrow \epsilon \tarrow \epsilon} \, 
    \synbrack{\textbf{x}_\alpha} \, 
    \sE(\textbf{B}_\beta)}_{\alpha \tarrow \beta})(d)\\
    &
    = V^{\cal M}_{\phi}(\sembrack{\sE(\LambdaApp \textbf{x}_\alpha \mdot 
    \textbf{B}_\beta)}_{\alpha \tarrow \beta})(d)\\
    &
    = V^{\cal M}_{\phi}(\LambdaApp \textbf{x}_\alpha \mdot \textbf{B}_\beta)(d)\\
    &
    = V^{\cal M}_{\phi[{\bf x}_\alpha \mapsto d]}(\textbf{B}_\beta)\\
    &
    = V^{\cal M}_{\phi[{\bf x}_\alpha \mapsto d]}
    (\sembrack{\sE(\textbf{B}_\beta)}_\beta)\\
    &
    = V^{\cal M}_{\phi[{\bf x}_\alpha \mapsto d]}
    (\sembrack{\textbf{A}_\epsilon}_\beta)\\
    &
    = V^{\cal M}_{\phi}(\LambdaApp \textbf{x}_\alpha \mdot 
    \sembrack{\textbf{A}_\epsilon}_\beta)(d)
    \end{align}
    (2) is by (c); (3) is by the definition of $\sE$; (4) and (6) are
    by Lemma~\ref{lem:alt-disquotation}; (5) and (8) are by the
    semantics of function abstraction; and (7) is by (c), (d),
    Lemma~\ref{lem:val-b}, and the fact that constructions are closed
    expressions.\esp

  \item[] \textbf{Axiom B10.5} {\sglsp} Let (a) $V^{\cal
    M}_{\phi}(\mname{is-expr}_{\epsilon \tarrow o}^\epsilon \,
    \textbf{A}_\epsilon) = \TRUE$.  We must show
    \[V^{\cal M}_{\phi}(\sembrack{\mname{quo}_{\epsilon \tarrow \epsilon} \,
      \textbf{A}_\epsilon}_\epsilon) = V^{\cal
      M}_{\phi}(\textbf{A}_\epsilon)\] to prove Axiom B10.5 is valid in
    $\sM$.  By part 3 of Lemma~\ref{lem:sem-is-expr}, (a) implies (b)
    $V^{\cal M}_{\phi}(\textbf{A}_\epsilon) = V^{\cal
      M}_{\phi}(\sE(\textbf{B}_\epsilon))$ for some
    $\textbf{B}_\epsilon$.  
    \begin{align} \setcounter{equation}{0}
    &
    V^{\cal M}_{\phi}(\sembrack{\mname{quo}_{\epsilon \tarrow \epsilon} \,
      \textbf{A}_\epsilon}_\epsilon)\\
    &
    = V^{\cal M}_{\phi}(\sembrack{\mname{quo}_{\epsilon \tarrow \epsilon} \,
      \sE(\textbf{B}_\epsilon)}_\epsilon)\\
    &
    = V^{\cal M}_{\phi}(\sembrack{\sE(\synbrack{\textbf{B}_\epsilon})}_\epsilon)\\
    &
    = V^{\cal M}_{\phi}(\synbrack{\textbf{B}_\epsilon})\\
    &
    = V^{\cal M}_{\phi}(\sE(\textbf{B}_\epsilon))\\
    &
    = V^{\cal M}_{\phi}(\textbf{A}_\epsilon)
    \end{align}
    (2) and (6) are by (b); (3) is by the definition of $\sE$; (4) is
    by Lemma~\ref{lem:alt-disquotation}; and (5) is by Law of
    Quotation.

\ei

\medskip

\noindent \textbf{Axiom Group B11}

\bi

  \item[] \textbf{Axiom B11.1} {\sglsp} \item[] We must show
    \[V^{\cal M}_{\phi}((\LambdaApp \textbf{x}_\alpha \mdot 
    \sembrack{\textbf{B}_\epsilon}_\beta) \, \textbf{x}_\alpha) =
    V^{\cal M}_{\phi}(\sembrack{\textbf{B}_\epsilon}_\beta)\] to prove
    Axiom B11.1 is valid in $\sM$.
    \begin{align} \setcounter{equation}{0}
    &
    V^{\cal M}_{\phi}((\LambdaApp \textbf{x}_\alpha \mdot 
    \sembrack{\textbf{B}_\epsilon}_\beta) \, 
    \textbf{x}_\alpha)\\
    &
    = V^{\cal M}_{\phi[{\bf x}_\alpha \mapsto V^{\cal M}_{\phi}({\bf  x}_\alpha)]}
    (\sembrack{\textbf{B}_\epsilon}_\beta)\\
    & 
    = V^{\cal M}_{\phi}(\sembrack{\textbf{B}_\epsilon}_\beta)
    \end{align}
    (2) is by part 2 of Lemma~\ref{lem:val-a} and (3) is by
    \[\phi[{\bf x}_\alpha \mapsto V^{\cal M}_{\phi}({\bf x}_\alpha)]
    = \phi[{\bf x}_\alpha \mapsto \phi({\bf x}_\alpha)]
    = \phi.\]

  \item[] \textbf{Axiom B11.2} {\sglsp} Let (a) $V^{\cal
    M}_{\phi}(\mname{is-expr}_{\epsilon \tarrow o}^{\beta} \,
    ((\LambdaApp \textbf{x}_\alpha \mdot \textbf{B}_\epsilon) \,
    \textbf{A}_\alpha)) = \TRUE$ and\\ (b) $V^{\cal
      M}_{\phi}(\Neg(\mname{is-free-in}_{\epsilon \tarrow \epsilon
      \tarrow o} \, \synbrack{\textbf{x}_\alpha} \, ((\LambdaApp
    \textbf{x}_\alpha \mdot \textbf{B}_\epsilon) \,
    \textbf{A}_\alpha))) = \TRUE$.  We must show \[V^{\cal
      M}_{\phi}((\LambdaApp \textbf{x}_\alpha \mdot
    \sembrack{\textbf{B}_\epsilon}_\beta) \, \textbf{A}_\alpha) =
    V^{\cal M}_{\phi}(\sembrack{(\LambdaApp \textbf{x}_\alpha \mdot
      \textbf{B}_\epsilon) \, \textbf{A}_\alpha}_\beta)\] to prove
    Axiom B11.2 is valid in $\sM$.  By part 3 of
    Lemma~\ref{lem:sem-is-expr}, (a) implies (c) $V^{\cal
      M}_{\phi}((\LambdaApp \textbf{x}_\alpha \mdot
    \textbf{B}_\epsilon) \, \textbf{A}_\alpha) = V^{\cal
      M}_{\phi}(\sE(\textbf{C}_\beta))$ for some (eval-free)
    $\textbf{C}_\beta$. By (a) and part 2 of Lemma~\ref{lem:val-a},
    (d) $V^{\cal M}_{\phi[{\bf x}_\alpha \mapsto V^{\cal
          M}_{\phi}({\bf A}_\alpha)]}(\mname{is-expr}_{\epsilon
      \tarrow o}^{\beta} \, \textbf{B}_\epsilon) = \TRUE$. By (b),
    (c), Lemma~\ref{lem:sem-is-free-in}, and the Law of Quotation, (e)
    $\textbf{x}_\alpha$ is not free in $\textbf{C}_\beta$.
    \begin{align} \setcounter{equation}{0}
    &
    V^{\cal M}_{\phi}((\LambdaApp \textbf{x}_\alpha \mdot
    \sembrack{\textbf{B}_\epsilon}_\beta) \, \textbf{A}_\alpha)\\
    &
    = V^{\cal M}_{\phi[{\bf x}_\alpha \mapsto V^{\cal M}_{\phi}({\bf  A}_\alpha)]}
    (\sembrack{\textbf{B}_\epsilon}_\beta)\\
    &
    = V^{\cal M}_{\phi[{\bf x}_\alpha \mapsto V^{\cal M}_{\phi}({\bf  A}_\alpha)]}
    ({\sE}^{-1}
    (V^{\cal M}_{\phi[{\bf x}_\alpha \mapsto V^{\cal M}_{\phi}({\bf  A}_\alpha)]}
    (\textbf{B}_\epsilon)))\\
    &
    = V^{\cal M}_{\phi[{\bf x}_\alpha \mapsto V^{\cal M}_{\phi}({\bf  A}_\alpha)]}
    ({\sE}^{-1}
    (V^{\cal M}_{\phi}
    ((\LambdaApp \textbf{x}_\alpha \mdot
    \textbf{B}_\epsilon) \, \textbf{A}_\alpha)))\\
    &
    = V^{\cal M}_{\phi[{\bf x}_\alpha \mapsto V^{\cal M}_{\phi}({\bf  A}_\alpha)]}
    ({\sE}^{-1}
    (V^{\cal M}_{\phi}
    (\sE(\textbf{C}_\beta))))\\
    &
    = V^{\cal M}_{\phi[{\bf x}_\alpha \mapsto V^{\cal M}_{\phi}({\bf  A}_\alpha)]}
    ({\sE}^{-1}
    (\sE(\textbf{C}_\beta)))\\
    &
    = V^{\cal M}_{\phi[{\bf x}_\alpha \mapsto V^{\cal M}_{\phi}({\bf  A}_\alpha)]}
    (\textbf{C}_\beta)\\
    &
    = V^{\cal M}_{\phi}(\textbf{C}_\beta)\\
    &
    = V^{\cal M}_{\phi}(\sembrack{\sE(\textbf{C}_\beta)}_\beta)\\
    &
    = V^{\cal M}_{\phi}(\sembrack{(\LambdaApp \textbf{x}_\alpha \mdot
    \textbf{B}_\epsilon) \, \textbf{A}_\alpha}_\beta)
    \end{align}
    (2) and (4) are by part 2 of Lemma~\ref{lem:val-a}; (3) is by (d)
    and condition 6 of the definition of a general model; (5) is by
    (c); (6) is by Proposition~\ref{prop:val-const}; (7) is immediate;
    (8) is by (e) and Lemma~\ref{lem:sem-is-free-in}; (9) is by
    Lemma~\ref{lem:alt-disquotation}; and (10) is by (c).

\ei

\noindent \textbf{Axiom B12} {\sglsp} By Lemmas~\ref{lem:val-b}
and~\ref{lem:not-effective-sem}.

\medskip

\noindent \textbf{Axiom B13} {\sglsp} By Lemma ~\ref{lem:not-effective-sem}
and the proof for Axiom A4.5.
\end{proof}

\begin{lem}[Rule R Preserves Validity]\label{lem:rule}
Rule R preserves validity in all general models for {\churchqe}.
\end{lem}

\begin{proof}
Let $\sM$ be a general model for {\churchqe}.  Suppose $\textbf{C}_o$
and $\textbf{C}'_o$ are formulas such that $\textbf{C}'_o$ is the
result of replacing one occurrence of $\textbf{A}_\alpha$ in
$\textbf{C}_o$ by an occurrence of $\textbf{B}_\alpha$, provided that
the occurrence of $\textbf{A}_\alpha$ in $\textbf{C}_o$ is not within
a quotation, not the first argument of a function abstraction, and not
the second argument of an evaluation.  Then it easily follows that
$V^{\cal M}_{\phi}(\textbf{A}_\alpha) = V^{\cal
  M}_{\phi}(\textbf{B}_\alpha)$ for all $\phi \in \mname{assign}(\sM)$
implies $V^{\cal M}_{\phi}(\textbf{C}_o) = V^{\cal
  M}_{\phi}(\textbf{C}'_o)$ for all $\phi \in \mname{assign}(\sM)$ by
induction on the structure of $\textbf{C}_o$ (condition 7 of the
definition of a general model is needed for the argument).  $\sM
\vDash \textbf{A}_\alpha = \textbf{B}_\alpha$ implies $V^{\cal
  M}_{\phi}(\textbf{A}_\alpha) = V^{\cal M}_{\phi}(\textbf{B}_\alpha)$
for all $\phi \in \mname{assign}(\sM)$ by part 1 of
Lemma~\ref{lem:val-a}, and hence $\sM \vDash \textbf{A}_\alpha =
\textbf{B}_\alpha$ and $\sM \vDash \textbf{C}_o$ imply $V^{\cal
  M}_{\phi}(\textbf{C}_o) = V^{\cal M}_{\phi}(\textbf{C}'_o) = \TRUE$
for all $\phi \in \mname{assign}(\sM)$.  Therefore, $\sM \vDash
\textbf{A}_\alpha = \textbf{B}_\alpha$ and $\sM \vDash \textbf{C}_o$
imply $\sM \vDash \textbf{C}'_o$, and so Rule R preserves validity in
$\sM$.
\end{proof}

\subsection{Soundness and Consistency Theorems}

\begin{thm}[Soundness Theorem]

Let $T$ be a theory of {\churchqe}, $\textbf{A}_o$ be
a formula of $T$, and $\sH$ be a set of formulas of $T$.
\be

  \item $\proves{T}{\textbf{A}_o}$ implies $T \vDash \textbf{A}_o$
    (i.e., the proof system for {\churchqe} is sound).

  \item $\proves{T,\sH}{\textbf{A}_o}$ implies $T,\sH \vDash
    \textbf{A}_o$.

\ee
\end{thm}

\begin{proof}

\medskip

\noindent 
\textbf{Part 1} {\sglsp} Assume $\proves{T}{\textbf{A}_o}$ and
$\sM$ is a general model for $T$.  We must show that $\sM \vDash
\textbf{A}_o$. By assumption, each member of $\Gamma$ is valid in
$\sM$.  By Lemma~\ref{lem:axioms}, each axiom of {\churchqe} is valid
in $\sM$.  And by Lemma~\ref{lem:rule}, Rule R preserves validity in
$\sM$.  Therefore, $\proves{T}{\textbf{A}_o}$ implies $\sM \vDash
\textbf{A}_o$.

\medskip

\noindent 
\textbf{Part 2} {\sglsp} Assume $\proves{T,\sH}{\textbf{A}_o}$, $\sM$
is a general model for $T$, $\phi \in \mname{assign}(\sM)$, and
$V^{\cal M}_{\phi}(\textbf{H}_o) = \TRUE$ for all $\textbf{H}_o \in
\sH$.  We need to show $V^{\cal M}_{\phi}(\textbf{A}_o) = \TRUE$.
There is a finite subset
$\set{\textbf{H}^{1}_o,\ldots,\textbf{H}^{n}_o}$ of $\sH$ such that
$\proves{T,\set{\textbf{H}^{1}_o,\ldots,\textbf{H}^{n}_o}}{\textbf{A}_o}$.
Then
\[\proves{T}{\textbf{H}^{1}_o \Implies (\ldots(\textbf{H}^{n}_o \Implies 
\textbf{A}_o)\ldots)}\] by the Deduction Theorem, and so by part 1 of
the theorem, \[\sM \vDash \textbf{H}^{1}_o \Implies
(\ldots(\textbf{H}^{n}_o \Implies \textbf{A}_o)\ldots).\] By
hypothesis, $V^{\cal M}_{\phi}(\textbf{H}^{i}_o) = \TRUE$ for each $i$
with $1 \le i \le n$, and so $V^{\cal M}_{\phi}(\textbf{A}_o) =
\TRUE$.
\end{proof}

\begin{cor}
The proof system for {\churchqe} satisfies Requirement R1.
\end{cor}

\begin{thm}[Consistency Theorem]
Let $T$ be a theory of {\churchqe}.  If $T$ has a general model, then
$T$ is consistent.
\end{thm}

\begin{proof}
Assume $\sM$ is a general model for $T$ and $T$ is inconsistent, i.e.,
$\proves{T}{F_o}$.  By the Soundness Theorem, $T \vDash F_o$ and
hence $\sM \vDash F_o$.  This means $V^{\cal M}_{\phi}(F_o) = \TRUE$
and hence $V^{\cal M}_{\phi}(F_o) \not= \FALSE$ (for any assignment
$\phi$), which contracts the definition of a general model.
\end{proof}

\section{Completeness}\label{sec:completeness}

We will now show that the proof system for {\churchqe} is complete
with respect to the (general models) semantics for {\churchqe} for
eval-free formulas.  More precisely, we will show that $T \vDash
\textbf{A}_o$ implies $\proves{T}{\textbf{A}_o}$ whenever $T$ is an
eval-free theory of {\churchqe} and $\textbf{A}_o$ is an eval-free
formula of $T$.  We will also show that the proof system for
{\churchqe} is not complete with respect to the semantics for
{\churchqe} for non-eval-free formulas.

\subsection{Eval-Free Completeness}

Let {\churcheps} be the logic obtained from {\churchqe} as follows:

\be

  \item Replace the set $\sC$ of constants by the expanded set $\sC
    \cup \sC'$ where:
  \[\sC' = \set{d^{\synbrack{{\bf
        x}_\alpha}}_{\epsilon} \;|\; \textbf{x}_\alpha \in \sV} \cup
  \set{d^{\synbrack{{\bf c}_\alpha}}_{\epsilon} \;|\;
    \textbf{c}_\alpha \in \sC}.\]

  \item Remove the quotation operator $\synbrack{\cdot}$ from the set
    of expression constructors so that there are no (primitive)
    quotations in {\churcheps}, but let $\synbrack{\textbf{A}_\alpha}$
    be an abbreviation defined by:

    \be

      \item $\synbrack{\textbf{x}_\alpha}$ stands for
        $d^{\synbrack{{\bf x}_\alpha}}_{\epsilon}$.

      \item $\synbrack{\textbf{c}_\alpha}$ stands for
        $d^{\synbrack{{\bf c}_\alpha}}_{\epsilon}$.

      \item $\synbrack{\textbf{F}_{\alpha \tarrow \beta} \,
        \textbf{A}_\alpha}$ stands for $\mname{app}_{\epsilon \tarrow
        \epsilon \tarrow \epsilon} \, \synbrack{\textbf{F}_{\alpha
          \tarrow \beta}} \, \synbrack{\textbf{A}_\alpha}$.

      \item $\synbrack{\LambdaApp \textbf{x}_\alpha \mdot
        \textbf{B}_\beta}$ stands for $\mname{abs}_{\epsilon \tarrow
        \epsilon \tarrow \epsilon} \, \synbrack{\textbf{x}_\alpha} \,
        \synbrack{\textbf{B}_\beta}$.

      \item $\synbrack{\synbrack{\textbf{A}_\alpha}}$ stands for
        $\mname{quo}_{\epsilon \tarrow \epsilon} \,
        \synbrack{\textbf{A}_\alpha}$.

    \ee

  \item Remove the evaluation operator $\sembrack{\cdot}_{\cdot}$ from
    the set of expression constructors so that there are no
    evaluations in {\churcheps}.

  \item A general model for {\churcheps} is the same as a general
    model for {\churchqe} except that (1) an interpretation maps the
    constants of the form $d^{\synbrack{{\bf x}_\alpha}}_{\epsilon}$
    and $d^{\synbrack{{\bf c}_\alpha}}_{\epsilon}$ to the
    constructions $\synbrack{\textbf{x}_\alpha}$ and
    $\synbrack{\textbf{c}_\alpha}$, respectively, and (2) a valuation
    function is not applicable to quotations and evaluations.

  \item The proof system for {\churcheps} is the same as the proof
    system for {\churchqe} except the axioms involving quotation,
    evaluation, and $\mname{IS-EFFECTIVE-IN}$ are removed: B8.1--3,
    B9, B10.1--5, B11.1--2, B12, and B13.

\ee
{\churcheps} is essentially the same as {\qzero} with the
inductive type $\epsilon$ added to it.

\begin{lem}\label{lem:completeness}
Let $T = (L_{\cal D}, \Gamma)$ be an eval-free theory of {\churchqe},
$\textbf{A}_o$ be an eval-free formula of $T$, and $T' = (L_{{\cal D}
  \cup {\cal C}'}, \Gamma)$.

\be

  \item If $\textbf{A}_o$ is valid in $T$ in {\churchqe}, then
    $\textbf{A}_o$ is valid in $T'$ in {\churcheps}.

  \item If $\textbf{A}_o$ is valid in $T'$ in {\churcheps}, then
    $\textbf{A}_o$ is a theorem of $T'$ in {\churcheps}.

  \item If $\textbf{A}_o$ is a theorem of $T'$ in {\churcheps}, then
    $\textbf{A}_o$ is a theorem of $T$ in {\churchqe}.

\ee
\end{lem}

\begin{proof}
Let (a) $T = (L_{\cal D}, \Gamma)$ be an eval-free theory of
{\churchqe}, (b) $\textbf{A}_o$ be an eval-free formula of $T$, and
$T' = (L_{{\cal D} \cup {\cal C}'}, \Gamma)$.  (a) implies $T'$ is a
theory of {\churcheps}, and (b) implies $\textbf{A}_o$ is a formula of
$T'$.

\medskip

\noindent 
\textbf{Part 1} {\sglsp} Let (c) $T \vDash \textbf{A}_o$ in
{\churchqe}.  Let $\sM$ be a model for $T'$ in {\churcheps}.  Then (c)
and clauses 2 and 4 of the definition of {\churcheps} implies $\sM
\vDash \textbf{A}_o$.  Therefore $T' \vDash \textbf{A}_o$ in
{\churcheps}.

\medskip

\noindent 
\textbf{Part 2} {\sglsp} Let $T' \vDash \textbf{A}_o$ in {\churcheps}.
Then $T' \vdash \textbf{A}_o$ in {\churcheps} by a proof that is
essentially the same as the proof of the completeness of the proof
system for {\qzero} (see 5502 in~\cite[p.~253]{Andrews02}).

\medskip

\noindent 
\textbf{Part 3} {\sglsp} Let $P$ be a proof of $\textbf{A}_o$ in $T'$
in {\churcheps}.  Define $\textbf{A}'_o$ to be result of replacing
each $\synbrack{\textbf{B}_\beta}$ in $\textbf{A}_o$, where
$\textbf{B}_\beta \not\in \sV \cup \sC$, with $\sE(\textbf{B}_\beta)$,
and define $P'$ to be the result of replacing each $d^{\synbrack{{\bf
      x}_\alpha}}_{\epsilon}$ in $P'$ with
$\synbrack{\textbf{x}_\alpha}$ and each $d^{\synbrack{{\bf
      c}_\alpha}}_{\epsilon}$ in $P'$ with
$\synbrack{\textbf{c}_\alpha}$.  Since the axioms of the proof system
for {\churcheps} are a subset of the axioms of the proof system for
{\churchqe} and both proof systems share the same rule of inference,
(d) $P'$ is a proof of $\textbf{A}'_o$ in $T$ in {\churchqe}.  By the
Syntactic Law of Quotation and Rule R, (e) $T \vdash \textbf{A}_o =
\textbf{A}'_o$ in {\churchqe}.  Therefore, $T \vdash \textbf{A}_o$ in
       {\churchqe} follows from (d) and (e) by the Equality Rules.

\end{proof}

\begin{thm}[Completeness for Eval-Free Formulas]
\bsp
Let $T$ be an eval-free theory of {\churchqe} and $\textbf{A}_o$ be an
eval-free formula of $T$.  If $T \vDash \textbf{A}_o$, then
$\proves{T}{\textbf{A}_o}$.
\esp
\end{thm}

\begin{proof}
Let $T = (L_{\cal D}, \Gamma)$ be an eval-free theory of {\churchqe},
$\textbf{A}_o$ be an eval-free formula of $T$, and $T' = (L_{{\cal D}
  \cup {\cal C}'}, \Gamma)$.  Assume $T \vDash \textbf{A}_o$ in
{\churchqe}.  This implies $T' \vDash \textbf{A}_o$ in {\churcheps},
which implies $\proves{T'}{\textbf{A}_o}$ in {\churcheps}, which
implies $\proves{T}{\textbf{A}_o}$ in {\churchqe} by parts 1, 2, and
3, respectively, of Lemma~\ref{lem:completeness}.
\end{proof}

\begin{cor}
The proof system for {\churchqe} satisfies Requirement R2.
\end{cor}

\subsection{Non-Eval-Free Incompleteness}

We will show that the example mentioned in the paragraph about the
Double Substitution problem in section~\ref{sec:introduction} is valid
in {\churchqe} but not provable in the proof system for {\churchqe}.

\begin{prop}\label{prop:double-sub}
$\vDash (\LambdaApp x_\epsilon \mdot \sembrack{x_\epsilon}_\epsilon) \,
  \synbrack{x_\epsilon} = \synbrack{x_\epsilon}$.
\end{prop}

\begin{proof}
Let $\sM$ be a general model for {\churchqe} and $\phi \in
\mname{assign}(\sM)$.  By part 1 of Lemma~\ref{lem:val-a}, we need to
show that \[V^{\cal M}_{\phi}((\LambdaApp x_\epsilon \mdot
\sembrack{x_\epsilon}_\epsilon) \, \synbrack{x_\epsilon}) = V^{\cal
  M}_{\phi}(\synbrack{x_\epsilon})\] to prove the proposition.
\begin{align} \setcounter{equation}{0}
&
V^{\cal M}_{\phi}((\LambdaApp x_\epsilon \mdot \sembrack{x_\epsilon}_\epsilon) \,
\synbrack{x_\epsilon})\\
&
= V^{\cal M}_{\phi[x_\epsilon \mapsto V^{\cal M}_{\phi}(\synbrack{x_\epsilon})]}
(\sembrack{x_\epsilon}_\epsilon)\\
&
= V^{\cal M}_{\phi[x_\epsilon \mapsto V^{\cal M}_{\phi}(\synbrack{x_\epsilon})]}
(\sE^{-1}(V^{\cal M}_{\phi[x_\epsilon \mapsto V^{\cal M}_{\phi}(\synbrack{x_\epsilon})]}
(x_\epsilon)))\\
&
= V^{\cal M}_{\phi[x_\epsilon \mapsto V^{\cal M}_{\phi}(\synbrack{x_\epsilon})]}
(\sE^{-1}(V^{\cal M}_{\phi}(\synbrack{x_\epsilon}))))\\
&
= V^{\cal M}_{\phi[x_\epsilon \mapsto V^{\cal M}_{\phi}(\synbrack{x_\epsilon})]}
(\sE^{-1}(\sE(x_\epsilon)))\\
&
= V^{\cal M}_{\phi[x_\epsilon \mapsto V^{\cal M}_{\phi}(\synbrack{x_\epsilon})]}
(x_\epsilon)\\
&
= V^{\cal M}_{\phi}(\synbrack{x_\epsilon})
\end{align}
(2) by part 2 of Lemma~\ref{lem:val-b}; (3) is by condition 6 of the
definition of a general model since \[V^{\cal M}_{\phi[x_\epsilon
    \mapsto V^{\cal M}_{\phi}(\synbrack{x_\epsilon})]}
(\mname{is-expr}_{\epsilon \tarrow o}^{\epsilon}(x_\epsilon)) =
V^{\cal M}_{\phi[x_\epsilon \mapsto V^{\cal
      M}_{\phi}(\synbrack{x_\epsilon})]}(\mname{is-expr}_{\epsilon
  \tarrow o}^{\epsilon}(\synbrack{x_\epsilon})) = \TRUE\] by the
semantics of variables and part 2 of Lemma~\ref{lem:sem-is-expr}; (4)
and (7) are by the semantics of variables; (5) is by condition 5 of
the definition of a general model; and (6) is immediate.
\end{proof}

\begin{prop} $\not\vdash (\LambdaApp x_\epsilon \mdot
  \sembrack{x_\epsilon}_\epsilon) \, \synbrack{x_\epsilon} =
  \synbrack{x_\epsilon}$.
\end{prop}

\begin{proof}
It is necessary to use Axiom B11.2 in order to prove $(\LambdaApp
x_\epsilon \mdot \sembrack{x_\epsilon}_\epsilon) \,
\synbrack{x_\epsilon} = \synbrack{x_\epsilon}$, but Axiom B11.2
requires that
\[\proves{}{\Neg\mname{is-free-in}_{\epsilon \tarrow \epsilon \tarrow o} \,
\synbrack{x_\epsilon} \, ((\LambdaApp x_\epsilon \mdot x_\epsilon) \,
\synbrack{x_\epsilon}}).\] However, $\proves{}{(\LambdaApp x_\epsilon
  \mdot x_\epsilon) \, \synbrack{x_\epsilon} = \synbrack{x_\epsilon}}$
holds by Axiom A4.2 and\\ $\proves{}{\mname{is-free-in}_{\epsilon
    \tarrow \epsilon \tarrow o} \, \synbrack{x_\epsilon} \,
  \synbrack{x_\epsilon}}$ holds by Lemma~\ref{lem:syn-is-free-in},
which implies \[\proves{}{\mname{is-free-in}_{\epsilon \tarrow
    \epsilon \tarrow o} \, \synbrack{x_\epsilon} \, ((\LambdaApp
  x_\epsilon \mdot x_\epsilon) \, \synbrack{x_\epsilon}})\] by the
Equality Rules.  Therefore, Axiom B11.2 is not applicable
to\\ $(\LambdaApp x_\epsilon \mdot \sembrack{x_\epsilon}_\epsilon) \,
\synbrack{x_\epsilon}$ by the Consistency Theorem, and thus
$(\LambdaApp x_\epsilon \mdot \sembrack{x_\epsilon}_\epsilon) \,
\synbrack{x_\epsilon} = \synbrack{x_\epsilon}$ cannot be proved in the
proof system for {\churchqe}.
\end{proof}

\begin{thm}[Incompleteness Theorem]
The proof system for {\churchqe} is incomplete.
\end{thm}

\begin{proof}
This theorem follows directly from the previous two propositions.
\end{proof}

\section{Examples Revisited}\label{sec:revisit}

We prove in this section within the proof system for {\churchqe} the
results that were stated in section~\ref{sec:examples}.  These proofs
show the efficacy of the proof system for {\churchqe} for reasoning
about syntax, instantiating formulas containing evaluations, and proving
schemas and meaning formulas.

\subsection{Reasoning about Syntax}\label{subsec:revisit-syntax}

Let $T = (L_{\cal D},\Gamma)$ be a theory of {\churchqe} such that
$\mname{make-implication}_{\epsilon \tarrow \epsilon \tarrow \epsilon}
\in \sD$ and $\Gamma$ contains the definition
\begin{align*}
&
\mname{make-implication}_{\epsilon \tarrow \epsilon \tarrow \epsilon}\\
&
= \LambdaApp x_\epsilon \mdot \LambdaApp y_\epsilon \mdot
(\mname{app}_{\epsilon \tarrow \epsilon \tarrow \epsilon} \,
(\mname{app}_{\epsilon \tarrow \epsilon \tarrow \epsilon} \,
\synbrack{\Implies_{o \tarrow o \tarrow o}} \, x_\epsilon) \, y_\epsilon).
\end{align*}

\begin{prop}\label{prop:make-impl-a}
$\proves{T}{\mname{make-implication}_{\epsilon \tarrow \epsilon
      \tarrow \epsilon} \, \synbrack{\textbf{A}_o} \,
    \synbrack{\textbf{B}_o} = \synbrack{\textbf{A}_o \Implies
      \textbf{B}_o}}$.
\end{prop}

\begin{proof}
\begin{align} \setcounter{equation}{0}
T \vdash {} 
&\mname{make-implication}_{\epsilon \tarrow \epsilon \tarrow \epsilon} = {}\nonumber\\
&\LambdaApp x_\epsilon \mdot \LambdaApp y_\epsilon \mdot
(\mname{app}_{\epsilon \tarrow \epsilon \tarrow \epsilon} \,
(\mname{app}_{\epsilon \tarrow \epsilon \tarrow \epsilon} \,
\synbrack{\Implies_{o \tarrow o \tarrow o}} \, x_\epsilon) \, y_\epsilon)\\
T \vdash {} 
&(\LambdaApp x_\epsilon \mdot \LambdaApp y_\epsilon \mdot
(\mname{app}_{\epsilon \tarrow \epsilon \tarrow \epsilon} \,
(\mname{app}_{\epsilon \tarrow \epsilon \tarrow \epsilon} \,
\synbrack{\Implies_{o \tarrow o \tarrow o}} \, x_\epsilon) \, y_\epsilon)) \, 
\synbrack{\textbf{A}_o} \, \synbrack{\textbf{B}_o}\nonumber\\
&= \mname{app}_{\epsilon \tarrow \epsilon \tarrow \epsilon} \,
(\mname{app}_{\epsilon \tarrow \epsilon \tarrow \epsilon} \,
\synbrack{\Implies_{o \tarrow o \tarrow o}} \, 
\synbrack{\textbf{A}_o}) \, \synbrack{\textbf{B}_o}\\
T \vdash {} 
&\sE(\textbf{A}_o \Implies \textbf{B}_o) =  
\mname{app}_{\epsilon \tarrow \epsilon \tarrow \epsilon} \,
(\mname{app}_{\epsilon \tarrow \epsilon \tarrow \epsilon} \,
\synbrack{\Implies_{o \tarrow o \tarrow o}} \, 
\synbrack{\textbf{A}_o}) \, \synbrack{\textbf{B}_o}\\
T \vdash {}
& \synbrack{\textbf{A}_o \Implies \textbf{B}_o} = 
\sE(\textbf{A}_o \Implies \textbf{B}_o)\\
T \vdash {} 
&\mname{make-implication}_{\epsilon \tarrow \epsilon \tarrow \epsilon} \, 
\synbrack{\textbf{A}_o} \, \synbrack{\textbf{B}_o} =
\synbrack{\textbf{A}_o \Implies \textbf{B}_o}
\end{align}
(1) follows from $T$ containing the definition of
$\mname{make-implication}_{\epsilon \tarrow \epsilon \tarrow
  \epsilon}$; (2) is by Beta-Reduction by Substitution; (3) is by the
definition of $\sE$; (4) is by the Syntactic Law of Quotation; and (5)
follows from (1), (2), (3), and (4) by the Equality Rules.
\end{proof}

\begin{prop}\label{prop:make-impl-b}
$\proves{T}{\sembrack{\mname{make-implication}_{\epsilon \tarrow
        \epsilon \tarrow \epsilon} \, \synbrack{\textbf{A}_o} \,
      \synbrack{\textbf{B}_o}}_o = \textbf{A}_o \Implies \textbf{B}_o}$.
\end{prop}

\begin{proof}
\begin{align} \setcounter{equation}{0}
&
T \vdash \mname{make-implication}_{\epsilon \tarrow \epsilon \tarrow \epsilon} \, 
\synbrack{\textbf{A}_o} \, \synbrack{\textbf{B}_o} = 
\synbrack{\textbf{A}_o \Implies \textbf{B}_o}\\
&
T \vdash \sembrack{
\mname{make-implication}_{\epsilon \tarrow \epsilon \tarrow \epsilon} \, 
\synbrack{\textbf{A}_o} \, \synbrack{\textbf{B}_o}}_o = 
\sembrack{\synbrack{\textbf{A}_o \Implies \textbf{B}_o}}_o\\
&
T \vdash \sembrack{
\mname{make-implication}_{\epsilon \tarrow \epsilon \tarrow \epsilon} \, 
\synbrack{\textbf{A}_o} \, \synbrack{\textbf{B}_o}}_o = 
\textbf{A}_o \Implies \textbf{B}_o
\end{align}
(1) is Proposition~\ref{prop:make-impl-a}; (2) follows from (1) by the
Equality Rules; and (3) follows from (2) and the Syntactic Law of
Disquotation by the Equality Rules.

\end{proof}

\subsection{Schemas}\label{subsec:revisit-schemas}

\begin{thm}[Law of Excluded Middle]\label{thm:lem}
\[\proves{}{\ForallApp x_\epsilon \mdot \mname{is-expr}_{\epsilon
      \tarrow o}^{o} \, x_\epsilon \Implies (\sembrack{x_\epsilon}_o
    \Or \Neg \sembrack{x_\epsilon}_o)}.\]
\end{thm}

\begin{proof}
\begin{align} \setcounter{equation}{0}
&
\vdash \mname{is-expr}_{\epsilon \tarrow o}^{o} \, x_\epsilon \Implies 
(\sembrack{x_\epsilon}_o \Or \Neg \sembrack{x_\epsilon}_o)\\
&
\vdash \ForallApp x_\epsilon \mdot 
\mname{is-expr}_{\epsilon  \tarrow o}^{o} \, x_\epsilon \Implies 
(\sembrack{x_\epsilon}_o \Or \Neg \sembrack{x_\epsilon}_o)
\end{align}
(1) is by the Tautology Theorem and (2) follows from (1) by Universal
Generalization.
\end{proof}

\begin{prop}\label{prop:lem-instance}
If $\textbf{A}_o$ is eval-free and $x_\epsilon$ is not free in
$\textbf{A}_o$, then $\textbf{A}_o \Or \Neg\textbf{A}_o$ can be
derived from the Law of Excluded Middle within the proof system for
{\churchqe}.
\end{prop}

\begin{proof}
\begin{align} \setcounter{equation}{0}
& \vdash \ForallApp x_\epsilon \mdot \mname{is-expr}_{\epsilon \tarrow o}^{o} \, 
  x_\epsilon \Implies (\sembrack{x_\epsilon}_o \Or \Neg \sembrack{x_\epsilon}_o)\\
& \vdash \mname{is-expr}_{\epsilon \tarrow o}^{o} \, \synbrack{\textbf{A}_o} \Implies 
  ((\LambdaApp x_\epsilon \mdot \sembrack{x_\epsilon}_o) \, 
  \synbrack{\textbf{A}_o} \Or 
  \Neg((\LambdaApp x_\epsilon \mdot \sembrack{x_\epsilon}_o) \, 
  \synbrack{\textbf{A}_o}))\\
& \vdash (\LambdaApp x_\epsilon \mdot \sembrack{x_\epsilon}_o) \, 
  \synbrack{\textbf{A}_o} 
  \Or \Neg((\LambdaApp x_\epsilon \mdot \sembrack{x_\epsilon}_o) \, 
  \synbrack{\textbf{A}_o})\\
& \vdash (\LambdaApp x_\epsilon \mdot \sembrack{x_\epsilon}_o) \, \synbrack{\textbf{A}_o} =
  \textbf{A}_o\\
& \vdash \textbf{A}_o \Or \Neg\textbf{A}_o
\end{align}
(1) is the Law of Excluded Middle; (2) follows from (1) by Universal
Instantiation; (3) follows from (2) by Lemma~\ref{lem:syn-is-expr} and
Modus Ponens; (4) follows from part 1 of Lemma~\ref{lem:axiomB11.2-app}
and the hypothesis; and (5) follows from (3) and (4) by the Equality
Rules.
\end{proof}

\subsection{Meaning Formulas}\label{subsec:revisit-meaning}

We show now that the meaning formula for $\mname{poly-diff}_{\epsilon
  \tarrow \epsilon \tarrow \epsilon}$ and applications of it are
provable in the proof system for {\churchqe}.  These results
illustrate the power of {\churchqe}'s facility for reasoning about
the interplay of syntax and semantics.

\begin{thm}[Derivatives of Polynomial Functions]\label{lem:derivatives}
\be

  \item[]

  \item $\proves{T_{\mathbb R}}{\mname{is-diff}_{(\iota \tarrow \iota)
      \tarrow o} \, (\LambdaApp \textbf{x}_\iota \mdot
    \textbf{x}_\iota) \And \mname{deriv}_{(\iota \tarrow \iota)
      \tarrow (\iota \tarrow \iota)} (\LambdaApp \textbf{x}_\iota
    \mdot \textbf{x}_\iota) = \LambdaApp \textbf{x}_\iota \mdot
    1_\iota}$.

  \item $\proves{T_{\mathbb R}}{\mname{is-diff}_{(\iota \tarrow \iota)
      \tarrow o} \, (\LambdaApp \textbf{x}_\iota \mdot
    \textbf{y}_\iota) \And \mname{deriv}_{(\iota \tarrow \iota)
      \tarrow (\iota \tarrow \iota)} (\LambdaApp \textbf{x}_\iota
    \mdot \textbf{y}_\iota) = \LambdaApp \textbf{x}_\iota \mdot
    0_\iota}$\\ where $\textbf{x}_\iota$ and $\textbf{y}_\iota$ are
    distinct.

  \item $\proves{T_{\mathbb R}}{\mname{is-diff}_{(\iota \tarrow \iota)
      \tarrow o} \, (\LambdaApp \textbf{x}_\iota \mdot
    \textbf{c}_\iota) \And \mname{deriv}_{(\iota \tarrow \iota)
      \tarrow (\iota \tarrow \iota)} (\LambdaApp \textbf{x}_\iota
    \mdot \textbf{c}_\iota) = \LambdaApp \textbf{x}_\iota \mdot
    0_\iota}$.

  \item $\proves{T_{\mathbb R}}{\mname{is-diff}_{(\iota \tarrow \iota) \tarrow o} \,
    \textbf{F}_{\iota \tarrow \iota} \Implies {}\\
    \hspace*{7ex} (\mname{is-diff}_{(\iota \tarrow \iota) \tarrow o} \,
    (-_{(\iota \tarrow \iota) \tarrow (\iota \tarrow \iota)} \,
    \textbf{F}_{\iota \tarrow \iota}) \And {}\\
    \hspace*{8ex} \mname{deriv}_{(\iota \tarrow \iota) \tarrow (\iota
      \tarrow \iota)} \, (-_{(\iota \tarrow \iota) \tarrow (\iota \tarrow
      \iota)} \, \textbf{F}_{\iota \tarrow \iota}) = {}\\
    \hspace*{8ex}{-_{(\iota \tarrow
      \iota) \tarrow (\iota \tarrow \iota)} \, (\mname{deriv}_{(\iota
      \tarrow \iota) \tarrow (\iota \tarrow \iota)} \,
    \textbf{F}_{\iota \tarrow \iota})})}$.

  \item $\proves{T_{\mathbb R}}{\mname{is-diff}_{(\iota \tarrow \iota) \tarrow o} \,
    \textbf{F}_{\iota \tarrow \iota} \And \mname{is-diff}_{(\iota
      \tarrow \iota) \tarrow o} \, \textbf{G}_{\iota \tarrow
      \iota}\Implies {}\\ 
    \hspace*{7ex}(\mname{is-diff}_{(\iota \tarrow \iota) \tarrow o} \,
    (\textbf{F}_{\iota \tarrow \iota} +_{(\iota \tarrow \iota) \tarrow
      (\iota \tarrow \iota) \tarrow (\iota \tarrow \iota)}
    \textbf{G}_{\iota \tarrow \iota}) \And {}\\
    \hspace*{8ex}\mname{deriv}_{(\iota \tarrow \iota) \tarrow (\iota
      \tarrow \iota)} \, (\textbf{F}_{\iota \tarrow \iota} +_{(\iota \tarrow \iota) \tarrow
      (\iota \tarrow \iota) \tarrow (\iota \tarrow \iota)}
    \textbf{G}_{\iota \tarrow \iota}) = {}\\
    \hspace*{8ex}(\mname{deriv}_{(\iota \tarrow \iota) \tarrow (\iota
      \tarrow \iota)} \, \textbf{F}_{\iota \tarrow \iota})\\
    \hspace*{9ex}+_{(\iota \tarrow \iota) \tarrow (\iota \tarrow
      \iota) \tarrow (\iota \tarrow \iota)}\\ 
    \hspace*{8ex}(\mname{deriv}_{(\iota \tarrow \iota) \tarrow (\iota
      \tarrow \iota)} \, \textbf{G}_{\iota \tarrow \iota}))}$.

  \item $\proves{T_{\mathbb R}}{\mname{is-diff}_{(\iota \tarrow \iota) \tarrow o} \,
    \textbf{F}_{\iota \tarrow \iota} \And \mname{is-diff}_{(\iota
      \tarrow \iota) \tarrow o} \, \textbf{G}_{\iota \tarrow
      \iota}\Implies {}\\ 
    \hspace*{7ex}(\mname{is-diff}_{(\iota \tarrow \iota) \tarrow o} \,
    (\textbf{F}_{\iota \tarrow \iota} *_{(\iota \tarrow \iota) \tarrow
      (\iota \tarrow \iota) \tarrow (\iota \tarrow \iota)}
    \textbf{G}_{\iota \tarrow \iota}) \And {}\\
    \hspace*{8ex}\mname{deriv}_{(\iota \tarrow \iota) \tarrow (\iota
      \tarrow \iota)} \, (\textbf{F}_{\iota \tarrow \iota} *_{(\iota \tarrow \iota) \tarrow
      (\iota \tarrow \iota) \tarrow (\iota \tarrow \iota)}
    \textbf{G}_{\iota \tarrow \iota}) = {}\\
    \hspace*{8ex}((\mname{deriv}_{(\iota \tarrow \iota) \tarrow
      (\iota \tarrow \iota)} \, \textbf{F}_{\iota \tarrow \iota})
    *_{(\iota \tarrow \iota) \tarrow (\iota \tarrow \iota) \tarrow
      (\iota \tarrow \iota)} \textbf{G}_{\iota \tarrow \iota})\\
    \hspace*{9ex}+_{(\iota \tarrow \iota) \tarrow (\iota \tarrow \iota) \tarrow
      (\iota \tarrow \iota)} {}\\
    \hspace*{8ex} (\textbf{F}_{\iota \tarrow \iota} *_{(\iota \tarrow
      \iota) \tarrow (\iota \tarrow \iota) \tarrow (\iota \tarrow
      \iota)} (\mname{deriv}_{(\iota \tarrow \iota) \tarrow (\iota
      \tarrow \iota)} \, \textbf{G}_{\iota \tarrow \iota})))}$.

\ee
\end{thm}

\begin{proof}
This theorem is proved in the standard way from the definition of
$\mname{deriv}_{(\iota \tarrow \iota) \tarrow (\iota \tarrow \iota)}$.
See~\cite{Spivak08} for details.
\end{proof} 

\begin{lem}\label{lem:poly}
\be

  \item[]

  \item $\proves{T_{\mathbb R}}{\mname{is-poly}_{\epsilon \tarrow o}
    \, \textbf{B}_\epsilon \Implies \mname{is-expr}_{\epsilon \tarrow
      o}^{\iota} \, \textbf{B}_\epsilon}$.

  \item $\proves{T_{\mathbb R}}{\mname{is-poly}_{\epsilon \tarrow o}
    \, \textbf{B}_\epsilon \Implies (\mname{is-free-in}_{\epsilon
      \tarrow \epsilon \tarrow o} \, \textbf{A}_\epsilon \,
    \textbf{B}_\epsilon} \Implies (\textbf{A}_\epsilon =
    \synbrack{x_\iota} \Or \textbf{A}_\epsilon = \synbrack{y_\iota}))$.

  \item $\proves{T_{\mathbb R}}{\mname{is-poly}_{\epsilon \tarrow o}
    \, \textbf{B}_\epsilon \Implies \mname{is-poly}_{\epsilon \tarrow
      o} \, (\mname{poly-diff}_{\epsilon \tarrow \epsilon \tarrow
      \epsilon} \, \textbf{B}_\epsilon \,
    \synbrack{\textbf{x}_\iota})}$.

  \item $\proves{T_{\mathbb R}}{\mname{is-poly}_{\epsilon \tarrow o}
    \, \textbf{B}_\epsilon \Implies (\LambdaApp u_\epsilon \mdot
    \sembrack{u_\epsilon}_\iota) \, \textbf{B}_\epsilon =
    \sembrack{\textbf{B}_\epsilon}_\iota}$.

  \item $\proves{T_{\mathbb R}}{\mname{is-poly}_{\epsilon \tarrow o}
    \, u_\epsilon \Implies (\LambdaApp z_\epsilon \mdot
    \sembrack{u_\epsilon}_\iota) \, \textbf{B}_\epsilon =
    \sembrack{u_\epsilon}_\iota}$.

  \item $\proves{T_{\mathbb R}}{\mname{is-poly}_{\epsilon \tarrow o}
    \, \textbf{B}_\epsilon \Implies {}\\
    \hspace*{7ex}(\LambdaApp u_\epsilon \mdot
    \sembrack{\mname{poly-diff}_{\epsilon \tarrow \epsilon \tarrow
        \epsilon} \, u_\epsilon \,
      \synbrack{\textbf{x}_\iota}}_\iota) \, \textbf{B}_\epsilon =
    \sembrack{\mname{poly-diff}_{\epsilon \tarrow \epsilon \tarrow
        \epsilon} \, \textbf{B}_\epsilon \,
      \synbrack{\textbf{x}_\iota}}_\iota}$.

  \item $\proves{T_{\mathbb R}}{\mname{is-poly}_{\epsilon \tarrow o}
    \, u_\epsilon \Implies {}\\
    \hspace*{7ex}(\LambdaApp z_\epsilon \mdot
    \sembrack{\mname{poly-diff}_{\epsilon \tarrow \epsilon \tarrow
        \epsilon} \, u_\epsilon \,
      \synbrack{\textbf{x}_\iota}}_\iota) \, \textbf{B}_\epsilon =
    \sembrack{\mname{poly-diff}_{\epsilon \tarrow \epsilon \tarrow
        \epsilon} \, u_\epsilon \,
      \synbrack{\textbf{x}_\iota}}_\iota}$.

\ee
\end{lem}

\bsp
\begin{proof}

\medskip

\noindent \textbf{Parts 1--3} {\sglsp} Follow straightforwardly from the definitions of
$\mname{is-poly}_{\epsilon \tarrow o}$ and
$\mname{poly-diff}_{\epsilon \tarrow \epsilon \tarrow \epsilon}$ and
the axioms for $\mname{is-var}_{\epsilon \tarrow o}^{\iota}$,
$\mname{is-con}_{\epsilon \tarrow o}^{\iota}$,
$\mname{is-expr}_{\epsilon \tarrow o}^{\iota}$, $\sqsubset_{\epsilon
  \tarrow \epsilon \tarrow o}$, and $\mname{is-free-in}_{\epsilon
  \tarrow \epsilon \tarrow o}$ (Axioms B1.1--4, B2.1--4, B3.1--9,
B5.1--6, and B7.1--8) by induction using the Induction Principle for
Constructions (Axiom B6).

\medskip

\noindent \textbf{Part 4} {\sglsp} Follows from parts 1 and 2
of the lemma and part 2 of Lemma~\ref{lem:axiomB11.2-app}.

\medskip

\noindent \textbf{Part 5} {\sglsp} Follows from parts 1 and 2 of the
lemma and part 3 of Lemma~\ref{lem:axiomB11.2-app}.

\medskip

\noindent \textbf{Part 6} {\sglsp} Let $\textbf{A}_o$ be
$\mname{is-poly}_{\epsilon \tarrow o} \, \textbf{B}_\epsilon$.
\begin{align} \setcounter{equation}{0}
& T_{\mathbb R} \vdash
  (\LambdaApp u_\epsilon \mdot
  \mname{poly-diff}_{\epsilon \tarrow \epsilon \tarrow \epsilon} \, 
  u_\epsilon \, \synbrack{\textbf{x}_\iota}) \, \textbf{B}_\epsilon =
  \mname{poly-diff}_{\epsilon \tarrow \epsilon \tarrow \epsilon} \, 
  \textbf{B}_\epsilon \, \synbrack{\textbf{x}_\iota}\\
& T_{\mathbb R}, \set{\textbf{A}_o} \vdash
  \mname{is-expr}_{\epsilon \tarrow o}^{\iota} \,
  ((\LambdaApp u_\epsilon \mdot
  \mname{poly-diff}_{\epsilon \tarrow \epsilon \tarrow \epsilon} \, 
  u_\epsilon \, \synbrack{\textbf{x}_\iota}) \, \textbf{B}_\epsilon)\\
& T_{\mathbb R}, \set{\textbf{A}_o} \vdash
  \Neg (\mname{is-free-in}_{\epsilon \tarrow \epsilon \tarrow o} \, 
  \synbrack{u_\epsilon} \,
  ((\LambdaApp u_\epsilon \mdot
  \mname{poly-diff}_{\epsilon \tarrow \epsilon \tarrow \epsilon} \, 
  u_\epsilon \, \synbrack{\textbf{x}_\iota}) \, \textbf{B}_\epsilon))\\
& T_{\mathbb R}, \set{\textbf{A}_o} \vdash
  (\LambdaApp u_\epsilon \mdot
  \sembrack{\mname{poly-diff}_{\epsilon \tarrow \epsilon \tarrow \epsilon} \, 
  u_\epsilon \, \synbrack{\textbf{x}_\iota}}_\iota) \, 
  \textbf{B}_\epsilon = {}\nonumber\\
& \hspace*{5.5ex}\sembrack{(\LambdaApp u_\epsilon \mdot
  \mname{poly-diff}_{\epsilon \tarrow \epsilon \tarrow \epsilon} \, 
  u_\epsilon \, \synbrack{\textbf{x}_\iota}) \, \textbf{B}_\epsilon}_\iota\\
& T_{\mathbb R}, \set{\textbf{A}_o} \vdash
  (\LambdaApp u_\epsilon \mdot
  \sembrack{\mname{poly-diff}_{\epsilon \tarrow \epsilon \tarrow \epsilon} \, 
  u_\epsilon \, \synbrack{\textbf{x}_\iota}}_\iota) \, 
  \textbf{B}_\epsilon = {}\nonumber\\
& \hspace*{5.5ex}
  \sembrack{\mname{poly-diff}_{\epsilon \tarrow \epsilon \tarrow \epsilon} \, 
  \textbf{B}_\epsilon \, \synbrack{\textbf{x}_\iota}}_\iota\\
& T_{\mathbb R} \vdash
  \textbf{A}_o \Implies {}\nonumber\\
& \hspace*{5.5ex}(\LambdaApp u_\epsilon \mdot
  \sembrack{\mname{poly-diff}_{\epsilon \tarrow \epsilon \tarrow \epsilon} \, 
  u_\epsilon \, \synbrack{\textbf{x}_\iota}}_\iota) \, 
  \textbf{B}_\epsilon = 
  \sembrack{\mname{poly-diff}_{\epsilon \tarrow \epsilon \tarrow \epsilon} \, 
  \textbf{B}_\epsilon \, \synbrack{\textbf{x}_\iota}}_\iota
\end{align}
(1) is by Beta-Reduction by Substitution; (2) and (3) follow from (1),
the hypothesis $\textbf{A}_o$, and parts 1--3 of the lemma by the
Equality Rules and propositional logic; (4) follows from (2), (3), and
Axiom B11.2 by Modus Ponens; (5) follows from (1) and (4) by the
Equality Rules; and (6), the lemma to be proved, follows from (5) by
the Deduction Theorem.

\medskip

\noindent \textbf{Part 7} {\sglsp} Similar to the proof of part 6.
\end{proof}
\esp

\begin{thm}[Meaning Formula for $\mname{poly-diff}_{\epsilon \tarrow \epsilon
  \tarrow \epsilon}$]\label{thm:meaning-form}
\begin{align*}
& \proves{T_{\mathbb R}}
  {\ForallApp u_\epsilon \mdot 
  (\mname{is-poly}_{\epsilon \tarrow o} \, u_\epsilon \Implies {} \\
& \hspace*{7ex}
  \mname{deriv}_{(\iota \tarrow \iota) \tarrow (\iota \tarrow \iota)}
  (\LambdaApp \textbf{x}_\iota \mdot \sembrack{u_\epsilon}_\iota) = 
  \LambdaApp \textbf{x}_\iota \mdot
  \sembrack{\mname{poly-diff}_{\epsilon \tarrow \epsilon \tarrow \epsilon} \, 
  u_\epsilon \, \synbrack{\textbf{x}_\iota}}_\iota)}.
\end{align*}
where $\textbf{x}_\iota$ is either $x_\iota$ or $y_\iota$.
\end{thm}

\begin{proof}
Without loss of generality, we may assume that $\textbf{x}_\iota$ is
$x_\iota$; the proof is exactly the same when $\textbf{x}_\iota$ is
$y_\iota$.  For this proof we make the following notational
definitions:

\be

  \item $\textbf{C}_o$ {\sglsp} is
    \begin{align*}
    & \mname{is-diff}_{(\iota \tarrow \iota) \tarrow o} \, 
      (\LambdaApp x_\iota \mdot \sembrack{u_\epsilon}_\iota) \And {}\\
    & \mname{deriv}_{(\iota \tarrow \iota) \tarrow (\iota \tarrow \iota)}
      (\LambdaApp x_\iota \mdot \sembrack{u_\epsilon}_\iota) = 
      \LambdaApp x_\iota \mdot
      \sembrack{\mname{poly-diff}_{\epsilon \tarrow \epsilon \tarrow \epsilon} \, 
      u_\epsilon \, \synbrack{x_\iota}}_\iota.
    \end{align*}

  \item $\textbf{P}_{\epsilon \tarrow o}$ {\sglsp} is {\sglsp}
    $\LambdaApp u_\epsilon \mdot (\mname{is-poly}_{\epsilon \tarrow o}
    \, u_\epsilon \Implies \textbf{C}_o)$.

  \item $\textbf{D}_o$ {\sglsp} is {\sglsp} $\ForallApp z_\epsilon
    \mdot (z_\epsilon \sqsubset_{\epsilon \tarrow \epsilon \tarrow o}
    u_\epsilon \Implies \textbf{P}_{\epsilon \tarrow o} \, z_\epsilon)$.

  \item $\textbf{E}^{1}_o$ {\sglsp} is {\sglsp}
    $\mname{is-poly}_{\epsilon \tarrow o} \, v_\epsilon \And
    u_\epsilon = \synbrack{-_{\iota \tarrow \iota} \,
      \commabrack{v_\epsilon}}$.

  \item $\textbf{E}^{2}_o$ {\sglsp} is {\sglsp}
    $\mname{is-poly}_{\epsilon \tarrow o} \, v_\epsilon \And
    \mname{is-poly}_{\epsilon \tarrow o} \, w_\epsilon \And
    u_\epsilon = \synbrack{\commabrack{v_\epsilon} +_{\iota \tarrow
        \iota \tarrow \iota} \commabrack{w_\epsilon}}$.

  \item $\textbf{E}^{3}_o$ {\sglsp} is {\sglsp}
    $\mname{is-poly}_{\epsilon \tarrow o} \, v_\epsilon \And
    \mname{is-poly}_{\epsilon \tarrow o} \, w_\epsilon \And
    u_\epsilon = \synbrack{\commabrack{v_\epsilon} *_{\iota \tarrow
        \iota \tarrow \iota} \commabrack{w_\epsilon}}$.

\ee

We will prove the following statement, a stronger result
that immediately implies the theorem: \[\proves{T_{\mathbb
    R}}{\ForallApp u_\epsilon \mdot \textbf{P}_{\epsilon \tarrow o} \,
  u_\epsilon}.\]
Our proof is given by the following derivation:
\begin{align} \setcounter{equation}{0}
& T_{\mathbb R} \vdash 
  \Neg \mname{IS-EFFECTIVE-IN}(u_\iota, \textbf{P}_{\epsilon \tarrow o})\\
& T_{\mathbb R} \vdash 
  \Neg \mname{IS-EFFECTIVE-IN}(z_\iota, \textbf{P}_{\epsilon \tarrow o})\\
& T_{\mathbb R} \vdash
  \ForallApp u_\epsilon \mdot (\textbf{D}_o \Implies 
  \textbf{P}_{\epsilon \tarrow o} \, u_\epsilon) \Implies
  \ForallApp u_\epsilon \mdot \textbf{P}_{\epsilon \tarrow o} \, u_\epsilon\\
& T_{\mathbb R}, \set{\textbf{D}_o, 
  u_\epsilon = \synbrack{x_\iota}} \vdash
  \textbf{C}_o\\
& T_{\mathbb R}, \set{\textbf{D}_o, 
  u_\epsilon = \synbrack{y_\iota}} \vdash
  \textbf{C}_o\\
& T_{\mathbb R}, \set{\textbf{D}_o, 
  u_\epsilon = \synbrack{0_\iota}} \vdash
  \textbf{C}_o\\
& T_{\mathbb R}, \set{\textbf{D}_o, 
  u_\epsilon = \synbrack{1_\iota}} \vdash
  \textbf{C}_o\\
& T_{\mathbb R}, \set{\textbf{D}_o,
  \ForsomeApp v_\epsilon \mdot \textbf{E}^{1}_o} 
  \vdash
  \textbf{C}_o\\
& T_{\mathbb R}, \set{\textbf{D}_o,
  \ForsomeApp v_\epsilon \mdot \ForsomeApp w_\epsilon \mdot \textbf{E}^{2}_o}
  \vdash
  \textbf{C}_o\\
& T_{\mathbb R}, \set{\textbf{D}_o, 
  \ForsomeApp v_\epsilon \mdot \ForsomeApp w_\epsilon \mdot \textbf{E}^{3}_o}
  \vdash
  \textbf{C}_o\\
& T_{\mathbb R}, \set{\textbf{D}_o,
  \mname{is-poly}_{\epsilon \tarrow o} \, u_\epsilon} \vdash
  \textbf{C}_o\\
& T_{\mathbb R} \vdash
  \textbf{D}_o \Implies 
  (\mname{is-poly}_{\epsilon \tarrow o} \, u_\epsilon \Implies \textbf{C}_o)\\
& T_{\mathbb R} \vdash
  \textbf{D}_o \Implies \textbf{P}_{\epsilon \tarrow o} \, u_\epsilon\\
& T_{\mathbb R} \vdash
  \ForallApp u_\epsilon \mdot (\textbf{D}_o \Implies 
  \textbf{P}_{\epsilon \tarrow o} \, u_\epsilon)\\
& T_{\mathbb R} \vdash
  \ForallApp u_\epsilon \mdot \textbf{P}_{\epsilon \tarrow o} \, u_\epsilon
\end{align}
(1) follows from the definition of $\mname{IS-EFFECTIVE-IN}$ using
Axiom A4.6; (2) follows the definition of $\mname{IS-EFFECTIVE-IN}$
using part 5 of Lemma~\ref{lem:poly}; (3) follows from (1), (2), Axiom
B6, the Induction Principle for Constructions, Alpha-Equivalence,
Universal Generalization, and Universal Instantiation; (4)--(10) are
proved below; (11) follows from (4)--(10) by the definition of
$\mname{is-poly}_{\epsilon \tarrow o}$ and propositional logic (proof
by cases); (12) follows from (11) by the Deduction Theorem; (13)
follows from (12) by Beta-Reduction by Substitution and the Equality
Rules; (14) follows from (13) by Universal Generalization; and (15)
follows from (14) and (3) by Modus Ponens.

\medskip

\noindent 
\textbf{Proof of (4)}
\begin{align}
& T_{\mathbb R} \vdash {}
  \mname{is-diff}_{(\iota \tarrow \iota) \tarrow o} \, 
  (\LambdaApp x_\iota \mdot x_\iota) \And 
  \mname{deriv}_{(\iota \tarrow \iota) \tarrow (\iota \tarrow \iota)} 
  (\LambdaApp x_\iota \mdot x_\iota) = 
  \LambdaApp x_\iota \mdot 1_\iota\\
& T_{\mathbb R} \vdash {}
  \sembrack{\synbrack{x_\iota}}_\iota = x_\iota\\
& T_{\mathbb R} \vdash {}
  \sembrack{\synbrack{1_\iota}}_\iota = 1_\iota\\
& T_{\mathbb R} \vdash {}
  \mname{poly-diff}_{\epsilon \tarrow \epsilon \tarrow \epsilon} \, 
  \synbrack{x_\iota} \, \synbrack{x_\iota} = \synbrack{1_\iota}\\
& T_{\mathbb R} \vdash {}
  \mname{is-diff}_{(\iota \tarrow \iota) \tarrow o} \, 
  (\LambdaApp x_\iota \mdot \sembrack{\synbrack{x_\iota}}_\iota) \And {}\nonumber\\
& \hspace{5.5ex}\mname{deriv}_{(\iota \tarrow \iota) \tarrow (\iota \tarrow \iota)} 
  (\LambdaApp x_\iota \mdot \sembrack{\synbrack{x_\iota}}_\iota) = 
  \LambdaApp x_\iota \mdot 
  \sembrack{\mname{poly-diff}_{\epsilon \tarrow \epsilon \tarrow \epsilon} \, 
  \synbrack{x_\iota} \, \synbrack{x_\iota}}_\iota\\
& T_{\mathbb R}, 
  \set{\textbf{D}_o,
  u_\epsilon = \synbrack{x_\iota}} \vdash {}\nonumber\\
& \hspace{5.5ex}\mname{is-diff}_{(\iota \tarrow \iota) \tarrow o} \, 
  (\LambdaApp x_\iota \mdot \sembrack{\synbrack{x_\iota}}_\iota) \And {}\nonumber\\
& \hspace{5.5ex}\mname{deriv}_{(\iota \tarrow \iota) \tarrow (\iota \tarrow \iota)} 
  (\LambdaApp x_\iota \mdot \sembrack{\synbrack{x_\iota}}_\iota) = 
  \LambdaApp x_\iota \mdot 
  \sembrack{\mname{poly-diff}_{\epsilon \tarrow \epsilon \tarrow \epsilon} \, 
  \synbrack{x_\iota} \, \synbrack{x_\iota}}_\iota\\
& T_{\mathbb R} \vdash
  \Neg \mname{IS-EFFECTIVE-IN}(x_\iota, u_\epsilon = \synbrack{x_\iota})\\
& T_{\mathbb R}, 
  \set{\textbf{D}_o, 
  u_\epsilon = \synbrack{x_\iota}} \vdash \textbf{C}_o
\end{align}
(16) is an instance of part 1 of Theorem~\ref{lem:derivatives}; (17)
and (18) are by the Syntactic Law of Disquotation; (19) follows from
the definition of $\mname{poly-diff}_{\epsilon \tarrow \epsilon
  \tarrow \epsilon}$ by Universal Generalization, Universal
Instantiation, Axiom B1.2, and Modus Ponens; (20) follows from
(16)--(19) by the Equality Rules; (21) follows from (20) by Weakening;
(22) is by Axiom B12; and (23) follows from (21), (22), and the
hypothesis $u_\epsilon = \synbrack{x_\iota}$ by Rule ${\rm R}'$.

\medskip

\noindent 
\textbf{Proof of (5)}{\sglsp} Similar to the proof of (4).

\medskip

\noindent 
\textbf{Proof of (6)}{\sglsp} Similar to the proof of (4).

\medskip

\noindent 
\textbf{Proof of (7)}{\sglsp} Similar to the proof of (4).

\medskip

\noindent 
\textbf{Proof of (8)}{\sglsp} Similar to the proof of (9).

\medskip

\noindent 
\textbf{Proof of (9)}
\begin{align}
& T_{\mathbb R} \vdash {}
  \mname{is-diff}_{(\iota \tarrow \iota) \tarrow o} \,
  (\LambdaApp x_\iota \mdot \sembrack{v_\epsilon}_\iota) \And 
  \mname{is-diff}_{(\iota \tarrow \iota) \tarrow o} \, 
  (\LambdaApp x_\iota \mdot \sembrack{w_\epsilon}_\iota) \Implies {}\nonumber\\ 
& \hspace*{7ex}(\mname{is-diff}_{(\iota \tarrow \iota) \tarrow o} \,
  (\LambdaApp x_\iota \mdot \sembrack{v_\epsilon}_\iota 
  +_{(\iota \tarrow \iota) \tarrow (\iota \tarrow \iota) \tarrow (\iota \tarrow \iota)}
  \LambdaApp x_\iota \mdot \sembrack{w_\epsilon}_\iota) \And {}\nonumber\\
& \hspace*{8ex}\mname{deriv}_{(\iota \tarrow \iota) \tarrow (\iota \tarrow \iota)} \, 
  (\LambdaApp x_\iota \mdot \sembrack{v_\epsilon}_\iota 
  +_{(\iota \tarrow \iota) \tarrow (\iota \tarrow \iota) \tarrow (\iota \tarrow \iota)}
  \LambdaApp x_\iota \mdot \sembrack{w_\epsilon}_\iota) = {}\nonumber\\
& \hspace*{8ex}(\mname{deriv}_{(\iota \tarrow \iota) \tarrow (\iota \tarrow \iota)} \, 
  \LambdaApp x_\iota \mdot \sembrack{v_\epsilon}_\iota)\nonumber\\
& \hspace*{9ex}
  +_{(\iota \tarrow \iota) \tarrow (\iota \tarrow \iota) \tarrow (\iota \tarrow \iota)}
  \nonumber\\ 
& \hspace*{8ex}(\mname{deriv}_{(\iota \tarrow \iota) \tarrow (\iota \tarrow \iota)} \, 
  \LambdaApp x_\iota \mdot \sembrack{w_\epsilon}_\iota))
\end{align}
\begin{align}
& T_{\mathbb R}, \set{\textbf{D}_o, \textbf{E}^{2}_o} \vdash {}\nonumber\\
& \hspace*{5ex}\mname{is-diff}_{(\iota \tarrow \iota) \tarrow o} \,
  (\LambdaApp x_\iota \mdot \sembrack{v_\epsilon}_\iota) \And 
  \mname{is-diff}_{(\iota \tarrow \iota) \tarrow o} \, 
  (\LambdaApp x_\iota \mdot \sembrack{w_\epsilon}_\iota) \Implies {}\nonumber\\ 
& \hspace*{7ex}(\mname{is-diff}_{(\iota \tarrow \iota) \tarrow o} \,
  (\LambdaApp x_\iota \mdot \sembrack{v_\epsilon}_\iota 
  +_{(\iota \tarrow \iota) \tarrow (\iota \tarrow \iota) \tarrow (\iota \tarrow \iota)}
  \LambdaApp x_\iota \mdot \sembrack{w_\epsilon}_\iota) \And {}\nonumber\\
& \hspace*{8ex}\mname{deriv}_{(\iota \tarrow \iota) \tarrow (\iota \tarrow \iota)} \, 
  (\LambdaApp x_\iota \mdot \sembrack{v_\epsilon}_\iota 
  +_{(\iota \tarrow \iota) \tarrow (\iota \tarrow \iota) \tarrow (\iota \tarrow \iota)}
  \LambdaApp x_\iota \mdot \sembrack{w_\epsilon}_\iota) = {}\nonumber\\
& \hspace*{8ex}(\mname{deriv}_{(\iota \tarrow \iota) \tarrow (\iota \tarrow \iota)} \, 
  \LambdaApp x_\iota \mdot \sembrack{v_\epsilon}_\iota)\nonumber\\
& \hspace*{9ex}
  +_{(\iota \tarrow \iota) \tarrow (\iota \tarrow \iota) \tarrow (\iota \tarrow \iota)}
  \nonumber\\ 
& \hspace*{8ex}(\mname{deriv}_{(\iota \tarrow \iota) \tarrow (\iota \tarrow \iota)} \, 
  \LambdaApp x_\iota \mdot \sembrack{w_\epsilon}_\iota))
\end{align}
\begin{align}
& T_{\mathbb R}, \set{\textbf{D}_o, \textbf{E}^{2}_o} \vdash 
  v_\epsilon \sqsubset u_\epsilon
\end{align}
\begin{align}
& T_{\mathbb R}, \set{\textbf{D}_o, \textbf{E}^{2}_o} \vdash 
  w_\epsilon \sqsubset u_\epsilon
\end{align}
\begin{align}
& T_{\mathbb R}, \set{\textbf{D}_o, \textbf{E}^{2}_o} \vdash 
  \textbf{P}_{\epsilon \tarrow o} \, v_\epsilon
\end{align}
\begin{align}
& T_{\mathbb R}, \set{\textbf{D}_o, \textbf{E}^{2}_o} \vdash 
  \textbf{P}_{\epsilon \tarrow o} \, w_\epsilon
\end{align}
\begin{align}
& T_{\mathbb R}, \set{\textbf{D}_o, \textbf{E}^{2}_o} \vdash\nonumber\\
& \hspace*{5ex}\mname{is-diff}_{(\iota \tarrow \iota) \tarrow o} \, 
  (\LambdaApp x_\iota \mdot \sembrack{v_\epsilon}_\iota) \And {}\nonumber\\
& \hspace*{5ex}\mname{deriv}_{(\iota \tarrow \iota) \tarrow (\iota \tarrow \iota)}
  (\LambdaApp x_\iota \mdot \sembrack{v_\epsilon}_\iota) = 
  \LambdaApp x_\iota \mdot
  \sembrack{\mname{poly-diff}_{\epsilon \tarrow \epsilon \tarrow \epsilon} \, 
  v_\epsilon \, \synbrack{x_\iota}}_\iota
\end{align}
\begin{align}
& T_{\mathbb R}, \set{\textbf{D}_o, \textbf{E}^{2}_o} \vdash\nonumber\\
& \hspace*{5ex}\mname{is-diff}_{(\iota \tarrow \iota) \tarrow o} \, 
  (\LambdaApp x_\iota \mdot \sembrack{w_\epsilon}_\iota) \And {}\nonumber\\
& \hspace*{5ex}\mname{deriv}_{(\iota \tarrow \iota) \tarrow (\iota \tarrow \iota)}
  (\LambdaApp x_\iota \mdot \sembrack{w_\epsilon}_\iota) = 
  \LambdaApp x_\iota \mdot
  \sembrack{\mname{poly-diff}_{\epsilon \tarrow \epsilon \tarrow \epsilon} \, 
  w_\epsilon \, \synbrack{x_\iota}}_\iota
\end{align}
\begin{align}
& T_{\mathbb R}, \set{\textbf{D}_o, \textbf{E}^{2}_o} \vdash {}\nonumber\\
& \hspace*{5ex}\mname{is-diff}_{(\iota \tarrow \iota) \tarrow o} \,
  (\LambdaApp x_\iota \mdot \sembrack{v_\epsilon}_\iota 
  +_{(\iota \tarrow \iota) \tarrow (\iota \tarrow \iota) \tarrow (\iota \tarrow \iota)}
  \LambdaApp x_\iota \mdot \sembrack{w_\epsilon}_\iota) \And {}\nonumber\\
& \hspace*{5ex}\mname{deriv}_{(\iota \tarrow \iota) \tarrow (\iota \tarrow \iota)} \, 
  (\LambdaApp x_\iota \mdot \sembrack{v_\epsilon}_\iota 
  +_{(\iota \tarrow \iota) \tarrow (\iota \tarrow \iota) \tarrow (\iota \tarrow \iota)}
  \LambdaApp x_\iota \mdot \sembrack{w_\epsilon}_\iota) = {}\nonumber\\
& \hspace*{5ex}(\mname{deriv}_{(\iota \tarrow \iota) \tarrow (\iota \tarrow \iota)} \, 
  \LambdaApp x_\iota \mdot \sembrack{v_\epsilon}_\iota)\nonumber\\
& \hspace*{6ex}
  +_{(\iota \tarrow \iota) \tarrow (\iota \tarrow \iota) \tarrow (\iota \tarrow \iota)}
  \nonumber\\ 
& \hspace*{5ex}(\mname{deriv}_{(\iota \tarrow \iota) \tarrow (\iota \tarrow \iota)} \, 
  \LambdaApp x_\iota \mdot \sembrack{w_\epsilon}_\iota)
\end{align}
\begin{align}
& T_{\mathbb R}, \set{\textbf{D}_o, \textbf{E}^{2}_o} \vdash {}\nonumber\\
& \hspace*{5ex}\mname{is-diff}_{(\iota \tarrow \iota) \tarrow o} \,
  (\LambdaApp x_\iota \mdot \sembrack{v_\epsilon}_\iota 
  +_{(\iota \tarrow \iota) \tarrow (\iota \tarrow \iota) \tarrow (\iota \tarrow \iota)}
  \LambdaApp x_\iota \mdot \sembrack{w_\epsilon}_\iota) \And {}\nonumber\\
& \hspace*{5ex}\mname{deriv}_{(\iota \tarrow \iota) \tarrow (\iota \tarrow \iota)} \, 
  (\LambdaApp x_\iota \mdot \sembrack{v_\epsilon}_\iota 
  +_{(\iota \tarrow \iota) \tarrow (\iota \tarrow \iota) \tarrow (\iota \tarrow \iota)}
  \LambdaApp x_\iota \mdot \sembrack{w_\epsilon}_\iota) = {}\nonumber\\
& \hspace*{5ex}\LambdaApp x_\iota \mdot
  \sembrack{\mname{poly-diff}_{\epsilon \tarrow \epsilon \tarrow \epsilon} \, 
  v_\epsilon \, \synbrack{x_\iota}}_\iota\nonumber\\
& \hspace*{6ex}
  +_{(\iota \tarrow \iota) \tarrow (\iota \tarrow \iota) \tarrow (\iota \tarrow \iota)}
  \nonumber\\ 
& \hspace*{5ex}\LambdaApp x_\iota \mdot
  \sembrack{\mname{poly-diff}_{\epsilon \tarrow \epsilon \tarrow \epsilon} \, 
  w_\epsilon \, \synbrack{x_\iota}}_\iota
\end{align}
\begin{align}
& T_{\mathbb R}, \set{\textbf{D}_o, \textbf{E}^{2}_o} \vdash {}\nonumber\\
& \hspace*{5ex}\mname{is-diff}_{(\iota \tarrow \iota) \tarrow o} \,
  (\LambdaApp x_\iota \mdot (\sembrack{v_\epsilon}_\iota 
  +_{\iota \tarrow \iota \tarrow \iota}
  \sembrack{w_\epsilon}_\iota)) \And {}\nonumber\\
& \hspace*{5ex}\mname{deriv}_{(\iota \tarrow \iota) \tarrow (\iota \tarrow \iota)} \, 
  (\LambdaApp x_\iota \mdot (\sembrack{v_\epsilon}_\iota 
  +_{\iota \tarrow \iota \tarrow \iota}
  \sembrack{w_\epsilon}_\iota)) = {}\nonumber\\
& \hspace*{5ex}\LambdaApp x_\iota \mdot
  (\sembrack{\mname{poly-diff}_{\epsilon \tarrow \epsilon \tarrow \epsilon} \, 
  v_\epsilon \, \synbrack{x_\iota}}_\iota
  +_{\iota \tarrow \iota \tarrow \iota}
  \sembrack{\mname{poly-diff}_{\epsilon \tarrow \epsilon \tarrow \epsilon} \, 
  w_\epsilon \, \synbrack{x_\iota}}_\iota)
\end{align}
\begin{align}
& T_{\mathbb R}, \set{\textbf{D}_o, \textbf{E}^{2}_o} \vdash {}\nonumber\\
& \hspace*{5ex}\mname{is-diff}_{(\iota \tarrow \iota) \tarrow o} \,
  (\LambdaApp x_\iota \mdot 
  (\sembrack{\synbrack{\commabrack{v_\epsilon} 
  +_{\iota \tarrow \iota \tarrow \iota} \commabrack{w_\epsilon}}}_\iota)) 
  \And {}\nonumber\\
& \hspace*{5ex}\mname{deriv}_{(\iota \tarrow \iota) \tarrow (\iota \tarrow \iota)} \, 
  (\LambdaApp x_\iota \mdot 
  (\sembrack{\synbrack{\commabrack{v_\epsilon} 
  +_{\iota \tarrow \iota \tarrow \iota} \commabrack{w_\epsilon}}}_\iota))
  = {}\nonumber\\
& \hspace*{5ex}\LambdaApp x_\iota \mdot
  (\sembrack{\synbrack{\commabrack{
  \mname{poly-diff}_{\epsilon \tarrow \epsilon \tarrow \epsilon} \, 
  v_\epsilon \, \synbrack{x_\iota}}
  +_{\iota \tarrow \iota \tarrow \iota}
  \commabrack{\mname{poly-diff}_{\epsilon \tarrow \epsilon \tarrow \epsilon} \, 
  w_\epsilon \, \synbrack{x_\iota}}}}_\iota)
\end{align}
\begin{align}
& T_{\mathbb R}, \set{\textbf{D}_o, \textbf{E}^{2}_o} \vdash {}\nonumber\\
& \hspace*{5ex}\mname{poly-diff}_{\epsilon \tarrow \epsilon \tarrow \epsilon} \, 
  \synbrack{\commabrack{v_\epsilon} 
  +_{\iota \tarrow \iota \tarrow \iota} 
  \commabrack{w_\epsilon}} \, \synbrack{x_\iota} = {}\nonumber\\
& \hspace*{5ex}\synbrack{\commabrack{
  \mname{poly-diff}_{\epsilon \tarrow \epsilon \tarrow \epsilon} \, 
  v_\epsilon \, \synbrack{x_\iota}}
  +_{\iota \tarrow \iota \tarrow \iota}
  \commabrack{\mname{poly-diff}_{\epsilon \tarrow \epsilon \tarrow \epsilon} \, 
  w_\epsilon \, \synbrack{x_\iota}}}
\end{align}
\begin{align}
& T_{\mathbb R}, \set{\textbf{D}_o, \textbf{E}^{2}_o} \vdash {}\nonumber\\
& \hspace*{5ex}\mname{is-diff}_{(\iota \tarrow \iota) \tarrow o} \,
  (\LambdaApp x_\iota \mdot 
  (\sembrack{\synbrack{\commabrack{v_\epsilon} 
  +_{\iota \tarrow \iota \tarrow \iota} \commabrack{w_\epsilon}}}_\iota)) 
  \And {}\nonumber\\
& \hspace*{5ex}\mname{deriv}_{(\iota \tarrow \iota) \tarrow (\iota \tarrow \iota)} \, 
  (\LambdaApp x_\iota \mdot 
  (\sembrack{\synbrack{\commabrack{v_\epsilon} 
  +_{\iota \tarrow \iota \tarrow \iota} \commabrack{w_\epsilon}}}_\iota))
  = {}\nonumber\\
& \hspace*{5ex}\LambdaApp x_\iota \mdot
  (\sembrack{  
  \mname{poly-diff}_{\epsilon \tarrow \epsilon \tarrow \epsilon} \, 
  \synbrack{\commabrack{v_\epsilon} 
  +_{\iota \tarrow \iota \tarrow \iota} 
  \commabrack{w_\epsilon}} \, \synbrack{x_\iota}}_\iota)
\end{align}
\begin{align}
& T_{\mathbb R}, \set{\textbf{D}_o, \textbf{E}^{2}_o} \vdash
  \Neg \mname{IS-EFFECTIVE-IN}(x_\iota,u_\epsilon =
  \synbrack{\commabrack{v_\epsilon} +_{\iota \tarrow \iota \tarrow \iota} 
  \commabrack{w_\epsilon}})
\end{align}
\begin{align}
& T_{\mathbb R}, \set{\textbf{D}_o, \textbf{E}^{2}_o} \vdash \textbf{C}_o
\end{align}
\begin{align}
& T_{\mathbb R}, \set{\textbf{D}_o, \textbf{E}^{2}_o} \vdash 
  \Neg \mname{IS-EFFECTIVE-IN}(v_\epsilon, \textbf{C}_o)
\end{align}
\begin{align}
& T_{\mathbb R}, \set{\textbf{D}_o, \textbf{E}^{2}_o} \vdash 
  \Neg \mname{IS-EFFECTIVE-IN}(w_\epsilon, \textbf{C}_o)
\end{align}
\begin{align}
& T_{\mathbb R}, \set{\textbf{D}_o, \textbf{E}^{2}_o} \vdash 
  \Neg \mname{IS-EFFECTIVE-IN}(v_\epsilon, \textbf{D}_o)
\end{align}
\begin{align}
& T_{\mathbb R}, \set{\textbf{D}_o, \textbf{E}^{2}_o} \vdash 
  \Neg \mname{IS-EFFECTIVE-IN}(w_\epsilon, \textbf{D}_o)
\end{align}
\begin{align}
& T_{\mathbb R}, \set{\textbf{D}_o, 
  \ForsomeApp v_\epsilon \mdot \ForsomeApp w_\epsilon \mdot \textbf{E}^{2}_o} 
  \vdash \textbf{C}_o
\end{align}
(24) is an instance of part 5 of Theorem~\ref{lem:derivatives}; (25)
follows from (24) by Weakening; (26) and (27) follow from the
hypothesis $\textbf{E}^{2}_o$ and Axioms B5.1--6 by Universal
Generalization, Universal Instantiation, the Equality Rules and
propositional logic; (28) and (29) follow from (26), (27), and the
hypothesis $\textbf{D}_o$ by Universal Instantiation, parts 5 and 7 of
Lemma~\ref{lem:poly}, the Equality Rules, and propositional logic;
(30) and (31) follow from (28), (29), and the hypothesis
$\textbf{E}^{2}_o$ by Universal Instantiation, parts 4 and 6 of
Lemma~\ref{lem:poly}, the Equality Rules, and propositional logic;
(32) follows from (25), (30), and (31) by propositional logic; (33)
follows from (30)--(32) by the Equality Rules and propositional logic;
(34) follows from (33) by the definition of $+_{(\iota \tarrow \iota)
  \tarrow (\iota \tarrow \iota) \tarrow (\iota \tarrow \iota)}$; (35)
follows from (34) and the hypothesis $\textbf{E}^{2}_o$ by Axiom
B10.3, Lemma~\ref{lem:poly}, quasiquotation, and propositional logic;
(36) follows from the definition of $\mname{poly-diff}_{\epsilon
  \tarrow \epsilon \tarrow \epsilon}$ and the hypothesis
$\textbf{E}^{2}_o$ by Universal Generalization, Universal
Instantiation, Axiom B1.2, and propositional logic; (37) follows from
(35) and (36) by the Equality Rules; (38) is by Axiom B12; (39)
follows from the hypothesis $\textbf{E}^{2}_o$, (37), and (38) by Rule
${\rm R}'$ and propositional logic; (40)--(43) follow from the
hypothesis $\textbf{E}^{2}_o$ and the definition of
$\mname{IS-EFFECTIVE}$ by Beta-Reduction by Substitution, part 4 of
Lemma~\ref{lem:axiomB11.2-app} and Lemma~\ref{lem:poly}; and (44)
follows from (41)--(43) by Lemma~\ref{lem:exist-rule}.

\medskip

\noindent 
\textbf{Proof of (10)}{\sglsp} Similar to the proof of (9).

\medskip

\noindent
This finally completes the proof of Theorem~\ref{thm:meaning-form}.
\end{proof}

\begin{prop}\label{prop:mf-app}
\[\proves{T_{\mathbb R}}{\mname{deriv}_{(\iota \tarrow \iota) \tarrow (\iota \tarrow
  \iota)} (\LambdaApp x_\iota \mdot x_{\iota}^{2}) = \LambdaApp
x_\iota \mdot 2_\iota * x_{\iota}}.\]
\end{prop}

\begin{proof}
\begin{align} \setcounter{equation}{0}
T_{\mathbb R} \vdash {} 
& \ForallApp u_\epsilon \mdot
  (\mname{is-poly}_{\epsilon \tarrow o} \, u_\epsilon \Implies {}\nonumber\\
& \hspace*{2ex}
  \mname{deriv}_{(\iota \tarrow \iota) \tarrow (\iota \tarrow \iota)}
  (\LambdaApp x_\iota \mdot \sembrack{u_\epsilon}_\iota) = 
  \LambdaApp x_\iota \mdot
  \sembrack{\mname{poly-diff}_{\epsilon \tarrow \epsilon \tarrow \epsilon} \, 
  u_\epsilon \, \synbrack{x_\iota}}_\iota)\\
T_{\mathbb R} \vdash {}
& \mname{is-poly}_{\epsilon \tarrow o} \, 
  \synbrack{x_{\iota}^{2}} \Implies {}\nonumber\\
& \hspace*{2ex}
  \mname{deriv}_{(\iota \tarrow \iota) \tarrow (\iota \tarrow \iota)}
  (\LambdaApp x_\iota \mdot 
  ((\LambdaApp u_\epsilon \mdot \sembrack{u_\epsilon}_\iota) \, 
  \synbrack{x_{\iota}^{2}})) = {}\nonumber\\
& \hspace*{2ex} 
  \LambdaApp x_\iota \mdot
  ((\LambdaApp u_\epsilon \mdot 
  \sembrack{\mname{poly-diff}_{\epsilon \tarrow \epsilon \tarrow \epsilon} \, 
  u_\epsilon \, \synbrack{x_\iota}}_\iota) \, 
  \synbrack{x_{\iota}^{2}}))\\
T_{\mathbb R} \vdash {} 
& \mname{is-poly}_{\epsilon \tarrow o} \, 
  \synbrack{x_{\iota}^{2}}\\
T_{\mathbb R} \vdash {}
& \mname{deriv}_{(\iota \tarrow \iota) \tarrow (\iota \tarrow \iota)}
  (\LambdaApp x_\iota \mdot 
  ((\LambdaApp u_\epsilon \mdot \sembrack{u_\epsilon}_\iota) \, 
  \synbrack{x_{\iota}^{2}})) = {}\nonumber\\
& \LambdaApp x_\iota \mdot
  ((\LambdaApp u_\epsilon \mdot 
  \sembrack{\mname{poly-diff}_{\epsilon \tarrow \epsilon \tarrow \epsilon} \, 
  u_\epsilon \, \synbrack{x_\iota}}_\iota) \, 
  \synbrack{x_{\iota}^{2}})\\
T_{\mathbb R} \vdash {}
& (\LambdaApp u_\epsilon \mdot \sembrack{u_\epsilon}_\iota) \, \synbrack{x_{\iota}^{2}} =
  \sembrack{\synbrack{x_{\iota}^{2}}}_\iota\\
T_{\mathbb R} \vdash {}
& (\LambdaApp u_\epsilon \mdot 
  \sembrack{\mname{poly-diff}_{\epsilon \tarrow \epsilon \tarrow \epsilon} \, 
  u_\epsilon \, \synbrack{x_\iota}}_\iota) \, \synbrack{x_{\iota}^{2}} = 
  \sembrack{\mname{poly-diff}_{\epsilon \tarrow \epsilon \tarrow \epsilon} \, 
  \synbrack{x_{\iota}^{2}} \, \synbrack{x_\iota}}_\iota\\
T_{\mathbb R} \vdash {}
& \mname{poly-diff}_{\epsilon \tarrow \epsilon \tarrow \epsilon} \, 
  \synbrack{x_{\iota}^{2}} \, \synbrack{x_\iota} = \synbrack{2_\iota * x_\iota}\\
T_{\mathbb R} \vdash {}
& \sembrack{\synbrack{x_{\iota}^{2}}}_\iota = x_{\iota}^{2}\\
T_{\mathbb R} \vdash {}
& \sembrack{\synbrack{2_\iota * x_{\iota}}}_\iota = 2_\iota * x_{\iota}\\
T_{\mathbb R} \vdash {}
& \mname{deriv}_{(\iota \tarrow \iota) \tarrow (\iota \tarrow \iota)} 
  (\LambdaApp x_\iota \mdot x_{\iota}^{2}) = 
  \LambdaApp x_\iota \mdot 2_\iota * x_{\iota}
\end{align}
(1) is Theorem~\ref{thm:meaning-form}; (2) follows from (1) by
Universal Instantiation; (3) follows from the definition of
$\mname{is-poly}_{\epsilon \tarrow o}$ by Beta-Reduction by
Substitution, Existential Generalization, and propositional logic; (4)
follows from (3) and (2) by Modus Ponens; (5) follows from (3) by part
4 of Lemma~\ref{lem:poly} by Modus Ponens; (6)~follows from (3) and
part 6 of Lemma~\ref{lem:poly} by Modus Ponens; (7) follows from the
definitions of $\mname{is-poly}_{\epsilon \tarrow o}$ and
$\mname{poly-diff}_{\epsilon \tarrow \epsilon \tarrow \epsilon}$ by
Universal Generalization, Universal Instantiation, Axiom B1.2, and
propositional logic; (8) and (9) are by the Syntactic Law of
Disquotation; and (10) follows from (4)--(9) by the Equality Rules.
\end{proof}

\section{Related Work}\label{sec:related-work}

\subsection{Metareasoning with Reflection}

\emph{Metareasoning} is reasoning about the behavior of a reasoning
system such as a proof system for a logic.  Metareasoning about the
proof system for a \emph{logic} $L$ is performed in a proof system for
a \emph{metalogic} $M$ where $M$ may be $L$ itself.  Since a proof
system involves manipulating expressions as syntactic objects,
metareasoning starts with reasoning about syntax.  This can be done in
a number of ways.  Kurt G\"odel's famously used \emph{G\"odel numbers}
in~\cite{Goedel31} to encode expressions and thereby reduce reasoning
about expressions to reasoning about natural numbers.  The technique
of a \emph{deep embedding} --- in which a particular language of
expressions is represented by an inductive type of values --- is the
most common means used today to reason about the syntax of a
language~\cite{BoultonEtAl93,ContejeanEtAl07,WildmoserNipkow04}.

Metareasoning is the most interesting when the metalogic $M$ is the
same as the logic $L$ and reasoning about $L$'s proof system is
integrated into reasoning within $L$'s proof system.  This is commonly
called \emph{reflection}.  The integration of meta-level reasoning in
object-level reasoning requires some form of \emph{quotation} and some
form of \emph{evaluation}.  Stanfania Costantini presents a general
survey of metareasoning and reflection in~\cite{Costantini02}, and
John Harrison in his excellent paper~\cite{Harrison95} surveys the
applications of reflection to computer theorem proving while arguing
that LCF-style proof assistants do not have an inherent need for
reflection.

Harrison identifies two kinds of reflection: logical and
computational.  \emph{Logical reflection} employs metareasoning about
$L$'s proof system within itself to reveal logical properties about
$L$.  G\"odel, Tarski, and others have used reflection in this form to
show the limits of formal logic~\cite{Goedel31,Tarski35a} and to
explore the logical impact of various \emph{reflection
  principles}~\cite{Harrison95,Koellner09}.  \emph{Computational
  reflection} incorporates algorithms that manipulate expressions and
other meta-level objects into the logic's proof system.  Examples of
such algorithms are the differentiation algorithm for polynomials
defined in section~\ref{sec:examples} and the ring tactic in the Coq
proof assistant~\cite{Boutin97,GregoireMahboubi05}.

Computational reflection has been explored and exploited in several
computer theorem systems.  In the seminal paper~\cite{BoyerMoore81},
Robert Boyer and J Moore developed a global infrastructure for
incorporating symbolic algorithms into the Nqthm~\cite{BoyerMoore88}
theorem prover.  This approach is also used in
ACL2~\cite{KaufmannMoore97}, the successor to Nqthm;
see~\cite{HuntEtAl05}.  Over the last 30 years, the Nuprl group lead
by Robert Constable has produced a large body of work on metareasoning
and reflection for theorem
proving~\cite{AllenEtAl90,Barzilay05,Constable95,Howe92,KnoblockConstable86,Nogin05,Yu07}
that has been implemented in the Nuprl~\cite{Constable86} and
MetaPRL~\cite{HickeyEtAl03} systems.  Proof by reflection has become a
mainstream technique in the Coq~\cite{Coq8.7.2} proof assistant with the
development of tactics based on symbolic computations like the Coq
ring tactic~\cite{Boutin97,GregoireMahboubi05} and the formalizations
in Coq of the \emph{four color theorem}~\cite{Gonthier08} and the
\emph{Feit-Thompson odd-order theorem}~\cite{GonthierEtAl13} led by
Georges Gonthier.
See~\cite{Boutin97,BraibantPous11,Chlipala13,GonthierEtAl15,GregoireMahboubi05,JamesHinze09,OostdijkGeuvers02}
for a selection of the work done on using reflection in Coq.
Agda~\cite{Norell07,Norell09} supports reflection in both programming
and proving; see~\cite{VanDerWalt12,VanDerWaltSwierstra12}.  Martin
Giese and Bruno Buchberger present in~\cite{GieseBuchberger07} the
design for a global infrastructure for employing reflection in the
Theorema~\cite{BuchbergerEtAl06} theorem prover.  See the following
references for research on using metareasoning and reflection in other
systems: 
Idris~\cite{Christiansen:2016,Christiansen:2014,Christiansen:2016:Thesis},
Isabelle/HOL~\cite{ChaiebNipkow08},
Lean~\cite{ebner2017metaprogramming},
Maude~\cite{ClavelMeseguer02}, 
PVS~\cite{VonHenkeEtAl98}, 
and
reFLect~\cite{MelhamEtAl13}.

The programming language community has likewise looked at reflection.
Its use for metaprogramming will be covered in the next section, but
trying to come to grasp with these ideas produced some interesting
papers on reflective theories~\cite{Mendhekar:1996} and
reification~\cite{DBLP:conf/lfp/FriedmanW84}. Of particular note is
that, even in a pure programming context, unrestricted reflection
leads to problems~\cite{DBLP:journals/lisp/Wand98}.  Kavvos' recent
D.Phil thesis~\cite{Kavvos2017} has a very interesting overview of the
impossibility of building a quotation operator (with certain
properties) and other dangers.  His literature review is also quite
extensive.

The Nqthm/ACL2~\cite{BoyerMoore81,HuntEtAl05} and
Theorema~\cite{GieseBuchberger07} approaches to computational
reflection are the approaches in the literature that are closest to
{\churchqe}.  Like {\churchqe}, these approaches utilize a global
reflection infrastructure for a traditional logic.

\subsection{Metaprogramming with Reflection}

\emph{Metaprogramming} is writing computer programs to manipulate and
generate computer programs in some programming language $L$.
Metaprogramming is especially useful when the ``metaprograms'' can be
written in $L$ itself.  This is facilitated by implementing in $L$
metaprogramming techniques for $L$ that involve the manipulation of
program code.  See~\cite{DemersMalenfant95} for a survey of how this
kind of ``reflection'' can be done for the major programming
paradigms.  

\bsp We listed in section~\ref{sec:introduction} several programming
languages that support metaprogramming with reflection: 
Lisp,
Agda~\cite{Norell07,Norell09,VanDerWalt12}, 
Elixir~\cite{Elixir18},
F\#~\cite{FSharp18},
Idris~\cite{Christiansen:2016,Christiansen:2014,Christiansen:2016:Thesis},
MetaML~\cite{TahaSheard00}, 
MetaOCaml~\cite{MetaOCaml11},
reFLect~\cite{GrundyEtAl06}, 
Scala~\cite{Odersky16,Scalameta18}, and
Template Haskell~\cite{SheardJones02}.  
These languages represent
fragments of computer code as values in an inductive type and include
quotation, quasiquotation, and evaluation operations.  For example,
these operations are called \emph{quote}, \emph{backquote}, and
\emph{eval} in the Lisp programming language.  The metaprogramming
language Archon~\cite{Stump09} developed by Aaron Stump offers an
interesting alternate approach in which program code is manipulated
directly instead of manipulating representations of computer code.
\esp

The reflection infrastructure in a programming language provides the
basis for \emph{multistage programming}~\cite{Taha04} in which code
generation and code execution are interleaved to produce programs that
are both general and efficient.  The code generation and execution can
take place at compile-time or run-time.
See~\cite{BergerEtAlArxiv16,CalcagnoEtAl03,MoggiFagorzi03} for
research on developing models for multistage programming
and~\cite{DaviesPfenning01,NanevskiEtAl08} for research on type
systems that support multistage programming.

\subsection{Theories of Quotation}

The semantics of the quotation operator $\synbrack{\cdot}$ is based on
the \emph{disquotational theory of quotation}~\cite{Quotation12}.
According to this theory, a quotation of an expression $e$ is an
expression that denotes $e$ itself.  In {\churchqe},
$\synbrack{\textbf{A}_\alpha}$ denotes a value that represents the
syntactic structure of $\textbf{A}_\alpha$.  Andrew Polonsky presents
in~\cite{Polonsky11} a set of axioms for quotation operators of this
kind.  There are several other theories of quotation that have been
proposed~\cite{Quotation12}.  For instance, quotation can be viewed as
an operation that constructs literals for syntactic values.  Florian
Rabe explores in~\cite{Rabe15} this approach to quotation.

\subsection{Theories of Truth}

Truth is a major subject in philosophy~\cite{Truth13}.  A theory of
truth seeks to explain what truth is and how the liar and other
related paradoxes can be resolved.  A \emph{truth
  predicate}~\cite{Truth13} is the face of a \emph{theory of truth}:
the properties of a truth predicate characterize a theory of
truth~\cite{Leitgeb07}.  A \emph{semantics theory of truth} defines a
truth predicate for a formal language, while an \emph{axiomatic theory
  of truth}~\cite{Halbach11,AxiomThyTruth13} specifies a truth
predicate for a formal language by means of an axiomatic theory.

In {\churchqe}, $\sembrack{\textbf{A}_\epsilon}_o$ asserts the truth
of the formula represented by $\textbf{A}_\epsilon$, and thus the
evaluation operator $\sembrack{\cdot}_o$ is a truth predicate.  Hence
{\churchqe} provides a semantic theory of truth via it semantics and
an axiomatic theory of truth via its proof system.  Since our goal is
not to explicate the nature of truth, it is not surprising that the
semantic and axiomatic theories of truth provided by {\churchqe} are
not very innovative.  Theories of truth --- starting with Tarski's
work~\cite{Tarski33,Tarski35,Tarski35a} in the 1930s --- have
traditionally been restricted to the truth of sentences, i.e.,
formulas with no free variables.  However, the {\churchqe} semantic
and axiomatic theories of truth admit formulas with free variables.

\subsection{Reasoning in the Lambda Calculus about Syntax}

Corrado B\"ohm and Alessandro Berarducci present
in~\cite{BoehmBerarducci85} a method for representing an inductive
type of values as a collection of lambda-terms.  Then functions
defined on the members of the inductive type can also be represented
as lambda terms.  Both the lambda terms representing the values and
those representing the functions defined on the values can be typed in
the second-order lambda calculus (System F)~\cite{Girard72,Reynolds74}
as shown in~\cite{BoehmBerarducci85}.  B\"ohm and his collaborators
present in~\cite{BerarducciBoehm93,BoehmEtAl94} a second, more
powerful method for representing inductive types as collections of
lambda-terms in which the lambda terms are not as easily typeable as
in the first method.  These two methods provide the means to
efficiently formalize syntax-based mathematical algorithms in the
lambda calculus.

Using the fact that inductive types can be directly represented in the
lambda calculus, Torben \AE. Mogensen in~\cite{Mogensen94} represents
the inductive type of lambda terms in lambda calculus itself as well
as defines a global evaluation operator in the lambda calculus.  (See
Henk Barendregt's survey paper~\cite{Barendregt97} on the impact of
the lambda calculus for a nice description of this work.)
Nevertheless these representations were only partially typed. The
\emph{finally tagless} approach to embedded
representations~\cite{CaretteKS09} kicked off a series of papers on
typed
self-representation~\cite{atkey2009syntax,atkey2009unembedding,brown2015self,BrownPalsberg16,BrownP17,BrownP18,JayP11,rendel2009typed}
which eventually succeeded at providing elegant solutions.

\subsection{Undefinedness}

Undefinedness naturally occurs in two places in {\churchqe}.  It
occurs when a syntax constructor is applied to inappropriate arguments
and when the evaluation operator $\sembrack{\cdot}_\alpha$ is applied
to an expression $\synbrack{\textbf{B}_\beta}$ where $\alpha \not=
\beta$.  We would prefer a cleaner version of {\churchqe} that
formalizes the traditional approach to undefinedness~\cite{Farmer04}.
Then improper constructions would not be needed and checking whether
an expression $ \textbf{A}_\epsilon$ denotes a construction or an
evaluation $\sembrack{\textbf{A}_\epsilon}_\alpha$ is meaningful would
be reduced to checking for definedness.  We argue in~\cite{Farmer04}
that a logic that supports the traditional approach to undefinedness is
much closer to mathematical practice than traditional logics and can
be effectively implemented.

We show in~\cite{Farmer08a} how to formalize the traditional approach
to undefinedness in a traditional logic.  The paper~\cite{Farmer08a}
presents {\qzerou}, a version of Andrews' {\qzero} that takes this
approach to undefinedness.  {\qzerou} is a simplified version of
{\lutins}~\cite{Farmer90,Farmer93b,Farmer94}, the logic of the the
{\imps} theorem proving
system~\cite{FarmerEtAl93,FarmerEtAl96}. Roughly speaking, {\qzerouqe}
is {\qzerou} plus quotation and evaluation.
{\churchuqe}~\cite{Farmer17} is a variant of {\churchqe} in which
undefinedness is incorporated in {\churchqe} in the same way that it
is incorporated in {\qzerou} and {\qzerouqe}.

\section{Conclusion}\label{sec:conclusion}

Quotation and evaluation provide a basis for metaprogramming as seen
in Lisp and other programming languages.  We believe that these
mechanisms can also provide a basis for metareasoning in traditional
logics like first-order logic or simple type theory~\cite{Farmer08}.
However, incorporating quotation and evaluation into a traditional
logic is much more challenging than incorporating them into a
programming language due to the Evaluation, Variable, and Double
Substitution Problems we described in the Introduction.

In this paper we have introduced {\churchqe}, a logic based on
{\qzero}~\cite{Andrews02}, Andrews' version of Church's type theory,
that includes quotation and evaluation.  We have presented the syntax
and semantics of {\churchqe} as well as a proof system for
{\churchqe}.  The syntax of {\churchqe} has the machinery of Church's
type theory plus an inductive type $\epsilon$ of syntactic values, a
partial quotation operator, and a typed evaluation operator.  The
semantics of {\churchqe} is based on Henkin-style general
models~\cite{Henkin50}.  Constructions --- certain expressions of type
$\epsilon$ --- represent the syntactic structures of eval-free
expressions (i.e., expressions that do not contain the evaluation
operator); they serve as the syntactic values in the semantics.  The
proof system for {\churchqe} is an extension of the proof system for
{\qzero}.  We proved that it is sound for all formulas (Requirement
R1) and complete for eval-free formulas (R2).  We also showed it can
be used to reason about constructions~(R3), can instantiate free
variables occurring within evaluations~(R4), and can prove formulas
containing evaluations such as schemas and meaning formulas for
syntax-based mathematical algorithms~(R5).

The Evaluation Problem is completely avoided in {\churchqe} by
restricting the quotation operator to eval-free expressions.  The
Variable Problem is solved by (1) using the more restrictive semantic
notion of ``a variable is effective in an expression'' in place of the
syntactic notion of ``a variable is free in an expression'' and (2)
adding beta-reduction axioms for quotations and evaluations to the
beta-reduction axioms used by Andrews in the proof system for
{\qzero}~\cite[p.~213]{Andrews02}.  The Double Substitution Problem is
solved by not allowing beta-reductions that embody a double
substitution.

Using examples, we have shown that {\churchqe} is suitable for
reasoning about the interplay of syntax and semantics, expressing
quasiquotations, and stating and proving schemas and meaning formulas.
In particular, we proved within the proof system for {\churchqe} the
meaning formula for an symbolic differentiation algorithm for
polynomials.  The proof of this result
(Theorem~\ref{thm:meaning-form}) is an comprehensive test of the
efficacy of {\churchqe}'s proof system.

{\churchqe} is much simpler than {\qzerouqe}~\cite{FarmerArxiv14}, a
richer, but more complicated, version of {\qzero} with undefinedness,
quotation, and evaluation.  In {\qzerouqe}, quotation may be applied
to expressions containing evaluations, expressions may be undefined
and functions may be partial, and substitution is implemented
explicitly as a logical constant.  Allowing quotation to be applied to
all expressions makes {\qzerouqe} much more expressive than
{\churchqe} but also much more difficult to implement since
substitution in the presence of evaluations is highly complex.  The
decision to represent ``a variable is free in an expression'' in the
logic but represent substitution only in the metalogic gives the proof
system for {\churchqe} much greater fluency that the proof system for
{\qzerouqe}.

We believe that {\churchqe} is the first version of simple type theory
with global quotation and evaluation that has a practicable proof
system.  We also believe that our approach for incorporating quotation
and evaluation into Church's type theory --- introducing an inductive
type of constructions, a partial quotation operator, and a typed
evaluation operator --- can be applied to other logics including
many-sorted first-order logic.  We have shown that developing the
needed syntax and semantics is relatively straightforward, while
developing a proof system for the logic is fraught with difficulties.

In our future research we will seek to answer the following three
questions:

\be

  \item Can {\churchqe} (or a logic like {\churchqe} with global
    quotation and evaluation) be effectively implemented as a software
    system?

  \item Is {\churchqe} an effective logic for developing defining,
    applying, proving properties about syntax-based mathematical
    algorithms.

  \item Is {\churchqe} an effective logic for formalizing graphs of
    biform theories?

\ee
Since {\churchqe} is a version of Church's type theory, the most
promising approach to answering the first question is to implement
{\churchqe} by extending HOL Light~\cite{Harrison09}, a simple
implementation of the HOL proof assistant~\cite{GordonMelham93}.  We
have developed a system called HOL Light
QE~\cite{CaretteFarmerLaskowskiArxiv18} by modifying HOL Light to
include global quotation and evaluation operators.  As future work, we
intend to continue the development of HOL Light QE and to show that
HOL Light QE can be effectively used to develop syntax-based
mathematical algorithms.

A \emph{biform theory}~\cite{CaretteFarmer08,Farmer07b} is a basic
unit of mathematical knowledge that consists of a set of
\emph{concepts} that denote mathematical values, \emph{transformers}
that denote symbolic algorithms, and \emph{facts} about the concepts
and transformers.  Since transformers manipulate the syntax of
expressions, biform theories are difficult to formalize in a
traditional logic.  The notion of a biform theory is a key component
of a framework for integrating axiomatic and algorithmic mathematics
that is being developed under the MathScheme
project~\cite{CaretteEtAl11} at McMaster University, led by Jacques
Carette and the author.  One of the main goals of the MathScheme is to
see if a logic like {\churchqe} can be used to develop a library of
biform theories connected by meaning preserving theory morphisms.  As
part of a case study~\cite{CaretteFarmer17}, we have formalized a
graph of biform theories encoding natural number arithmetic in
{\churchuqe}~\cite{Farmer17}, a variant of {\churchqe} with
undefinedness and theory morphisms.  Our next step in this direction
will be to formalize this same graph of biform theories in the HOL
Light QE system we mentioned above.

\addcontentsline{toc}{section}{Acknowledgments}
\section*{Acknowledgments} 

The author is grateful to Marc Bender, Jacques Carette, Michael
Kohlhase, Pouya Larjani, and Florian Rabe for many valuable
discussions on the use of quotation and evaluation in logic.  Peter
Andrews deserves special thanks for writing \emph{An Introduction to
  Mathematical Logic and Type Theory: To Truth through
  Proof}~\cite{Andrews02}.  The ideas embodied in {\churchqe} heavily
depend on the presentation of {\qzero} given in this superb textbook.
This research was supported by NSERC.  Finally, the author would like
to thank the referees for their insightful comments and suggestions.

\addcontentsline{toc}{section}{References}
\bibliography{$HOME/research/lib/imps}
\bibliographystyle{plain}

\end{document}